\documentclass[12pt]{article}
\usepackage[margin=1in]{geometry}
\usepackage{times}
\usepackage{color}
\usepackage[normalem]{ulem}
\usepackage{graphicx}
\usepackage[colorlinks,citecolor=blue,urlcolor=blue,linkcolor=blue,bookmarks=false,hypertexnames=true]{hyperref} 
\usepackage{bbm}
\usepackage{multirow}
\usepackage{amsmath}
\usepackage{amsthm}
\usepackage{amssymb}
\usepackage{natbib}
\usepackage{rotating}

\newtheorem{theorem}{Theorem}

\newtheorem{lemma}[theorem]{Lemma}

\newcommand{\atan}{\mathrm{atan2}}
\newcommand{\diag}{\mathrm{diag}}

\newcommand{\ilc}[1]{\tcp*[h]{#1}}

\newtheorem{prop}{Proposition}

\usepackage[linesnumbered,ruled,vlined]{algorithm2e}

\title{A mixed effects cosinor modelling framework for circadian gene expression}
\author{%
	Michael T. Gorczyca \\
	MTG Research Consulting\\
	{\texttt{mtg62@cornell.edu}}
}
\begin{document}

\maketitle
\begin{abstract}
\noindent The cosinor model is frequently used to represent the oscillatory behavior of different genes over time. When data are collected from multiple individuals, the cosinor model is estimated with recorded gene expression levels and the 24 hour day-night cycle time at which gene expression levels are observed. However, the timing of many biological processes are based on individual-specific internal timing systems that are offset relative to day-night cycle time. When these individual-specific offsets are unknown, they pose a challenge in performing statistical analyses with a cosinor model. Specifically, when each individual participating in a study has a different offset, the parameter estimates of a population cosinor model obtained with day-night cycle time are attenuated. These attenuated parameter estimates also attenuate test statistics, which inflate type II error rates in identifying genes with oscillatory behavior. To address this attenuation bias, this paper proposes a method when data are collected in a longitudinal design. This method involves first estimating individual-specific and population cosinor models for each gene, and then translating the times at which an individual's gene expression levels are recorded based on the parameter estimates of these models. Simulation studies confirm that this method mitigates bias in estimation and inference. Illustrations with data from three circadian biology studies highlight that this method produces parameter estimates and test statistics akin to those obtained when each individual's offset is known.

 \end{abstract}


\section{Introduction} \label{sec:1}

Human physiology, metabolism, and behavior all exhibit diurnal variations consistent with a circadian rhythm \citep{Marcheva2013}. These rhythms are generated by an internal self-sustained oscillator located in the hypothalamic suprachiasmatic nucleus, which is often referred to as a circadian clock \citep{Bae2001, Hastings2018, Herzog2017, Yamaguchi2003}. An understanding of a patient's circadian clock has implications in maximizing that patient's quality of care, as survival rates for open-heart surgery \citep{Montaigne2018}, effectiveness of chemotherapy \citep{Dallmann2016}, and response to a vaccine \citep{Long2016} each vary based on the time of day that treatment is administered. However, there is currently limited use of the circadian clock in the timing of treatment administration \citep{Dallmann2014, Wittenbrink2018}.

One challenge with integrating the circadian clock into timing treatment administration is that an individual's circadian clock time (internal circadian time, or ICT) is often offset relative to day-night cycle time (Zeitgeber time, or ZT) \citep{Duffy2011, Lewy1999, Wittenbrink2018}. Laboratory tests can determine the time at which melatonin onset occurs under dim-light conditions (dim-light melatonin onset time, or DLMO time), which is the gold-standard marker of this offset \citep{Hughey2017, Kennaway2023, Lewy1999, Ruiz2020, Wittenbrink2018}. However, these laboratory tests require that multiple tissue samples are collected from an individual under controlled conditions and analyzed \citep{Kantermann2015, Kennaway2019, Kennaway2020, Reid2019}. This process places a costly and labor-intensive burden on investigators who study how different experimental interventions effect the oscillatory behavior of multiple genes, which has implications in identifying treatment strategies based on the circadian clock. These tests can also fail to produce precise DLMO time estimates depending on the devices used in DLMO time determination \citep{Kennaway2019, Kennaway2020}. However, the use of ZT instead of ICT in developing statistical models to represent the oscillatory behavior of different genes can introduce bias in model parameter estimates, which would bias statistical inferences \citep{Gorczycab2024, Gorczycaa2024, Sollberger1962, Weaver1995}.

This paper is motivated by the challenge of determining DLMO time in order to model the oscillatory behavior of different genes, and considers a scenario where a longitudinal design is utilized to record gene expression levels from each individual. Initially, we assess the suitability of a linear mixed effects cosinor model, which is commonly used to represent gene expression over time \citep{Archer2014, delolmo2022, Fontana2012, Hou2021, MllerLevet2013}, in accounting for the offset of ICT relative to ZT when each individual's offset is unknown. This assessment reveals that fixed parameter estimates for this model and test statistics computed from it are biased towards zero, or suffer from attenuation bias, when each individual possesses a unique offset. To mitigate this bias without performing laboratory tests, this paper proposes a method to identify an individual's offset prior to regression. Specifically, given multiple samples collected from each individual over time, this method involves first estimating individual-specific and population cosinor models, and then adjusting the ZTs recorded for that individual based on the parameter estimates of these models. 

The remainder of this paper is organized as follows. In Section \ref{sec:2}, background on the mixed effects cosinor model, motivating results, and an overview of the proposed method are presented. In Section \ref{sec:3}, Monte Carlo simulation studies are performed to assess the utility of the proposed method. In Section \ref{sec:4}, the proposed method is applied on data from three circadian biology studies, where it is found to produce parameter estimates and inferences comparable to those where each individual's DLMO time has been determined from laboratory tests. Finally, in Section \ref{sec:5}, the proposed method and directions for future work are discussed.

\section{Methodology} \label{sec:2}

\subsection{Background and notation} \label{sec:2.1}

Suppose a longitudinal circadian biology experiment is conducted with $M$ individuals, where $n_i$ tissue samples are collected from the $i$-th individual over time. In this paper, it is assumed that each collected tissue sample is processed to extract the expression levels of $G$ different genes. Specifically, let $\boldsymbol{Y_i^{(g)}}$ denote an $n_i\times 1$ vector $[Y^{(g)}_{i,1},\ldots, Y^{(g)}_{i,n_i}]^T$, in which $Y_{i,j}^{(g)}$ denotes the $j$-th recording of the $g$-th gene for the $i$-th individual. Further, let $\boldsymbol{X_i}$ denote an $n_i\times 1$ vector $[X_{i,1},\ldots, X_{i,n_i}]^T$, in which $X_{i,j}$ denotes the Zeitgeber time at which $Y_{i,j}^{(g)}$ occurs for all $g = 1,\ldots, G$. It is emphasized that this paper adopts a notation convention where estimators, statistics, and other quantities related to the $g$-th gene are denoted with $(g)$ in their superscript.

When each sample is independent of every other sample, a cosinor model is often specified for modelling gene expression given time, which has the nonlinear amplitude-phase representation
\begin{align} \label{eq:cos}
    Y^{(g)}_{i,j} &= \mu_0^{(g)} + \theta^{(g)}_1\cos\left(\frac{\pi X_{i,j}}{12} + \theta_2^{(g)}\right)+\epsilon^{(g)}_{i,j}.
\end{align} 
Here, $\mu_0^{(g)}$ denotes the mean expression levels of the $g$-th gene; $\theta_1^{(g)}$ denotes the amplitude of the $g$-th gene, or the deviation from mean expression levels to peak expression levels; and $\theta_2^{(g)}$ denotes the phase-shift of the $g$-th gene, which relates to the time at which gene expression levels peak \citep{Cornelissen2014}. In this paper, the gene-specific random noise  $\epsilon_{i,j}^{(g)}\sim \mathcal{N}(0, (\sigma^{(g)})^2)$. When the model in (\ref{eq:cos}) is specified, it can be transformed into the linear model
\begin{equation} \label{eq:3}
    Y^{(g)}_{i,j} = \mu_0^{(g)} + \beta^{(g)}_1\sin\left(\frac{\pi X_{i,j}}{12}\right)+\beta^{(g)}_2\cos\left(\frac{\pi X_{i,j}}{12}\right) + \epsilon^{(g)}_{i,j},
\end{equation}
where the identities
\begin{align}
    \theta_1^{(g)} = \sqrt{(\beta^{(g)}_{1})^2+(\beta^{(g)}_{2})^2}, \ \quad & \theta_2^{(g)} = \atan(-\beta^{(g)}_{1}, \beta^{(g)}_{2}), \nonumber \\
    \beta^{(g)}_{1} = -\theta_1^{(g)}\sin(\theta_2^{(g)}), \quad \quad \quad  & \beta^{(g)}_{2} = \theta_1^{(g)}\cos(\theta_2^{(g)}) \label{eq:alt_to_orig}
\end{align}
can be used to transform the linear model in (\ref{eq:3}) into the nonlinear model in (\ref{eq:cos}), and vice versa \citep{Tong1976}. It is noted that transforming this nonlinear model into a linear model enables unbiased estimation of its parameters without additional technical assumptions \citep[Theorem 6.7]{Boos2013}.

A common extension of the model in (\ref{eq:cos}) is to enable each individual to influence their repeated outcomes. Specifically, mixed effects models can enable this influence \citep{Davidian1995, Hedeker2006, McCulloch2000}, where the cosinor model would instead be specified as
\begin{align} \label{eq:cos_m}
    Y^{(g)}_{i,j} &= \mu_0^{(g)} + m_{i,0}^{(g)} + (\theta^{(g)}_1+c^{(g)}_{i,1})\cos\left(\frac{\pi X_{i,j}}{12} + \theta_2^{(g)}+c_{i,2}^{(g)}\right)+\epsilon^{(g)}_{i,j}.
\end{align} 
Here, the parameters $\mu_0^{(g)}$, $\theta^{(g)}_1$, and $\theta^{(g)}_2$ from (\ref{eq:cos}) maintain their interpretation, and represent constant parameter estimands for the population, or fixed effects. Additionally, $m_{i,0}^{(g)}$ denotes the $i$-th individual's influence on their mean expression levels of the $g$-th gene; $c^{(g)}_{i,1}$ denotes the $i$-th individual's influence on their amplitude of the $g$-th gene; and $c^{(g)}_{i,2}$ denotes the $i$-th individual's influence on their phase-shift of the $g$-th gene. It is emphasized that $[m_{i,0}^{(g)}, c^{(g)}_{i,1}, c^{(g)}_{i,2}]^T \sim P^{(g)}$, or parameters related to an individual's influence on their repeated outcomes are considered random effects generated from a gene-specific probability distribution $P^{(g)}$. In this paper, it is assumed that the expectation of each random effect is zero. Further, it is not assumed that each gene-specific $c^{(g)}_{i,2}$ equals an individual's dim-light melatonin onset (DLMO) time. 

Many circadian biology studies assume that (\ref{eq:cos_m}) can be expressed as the linear mixed effects model
\begin{align} \label{eq:lin_m}
    Y^{(g)}_{i,j} &= \left\{\mu^{(g)}_0+\beta^{(g)}_1\sin(X_{i,j})+\beta^{(g)}_2\cos(X_{i,j})\right\} \nonumber \\
    & \quad + \left\{m^{(g)}_{i,0}+b^{(g)}_{i,1}\sin(X_{i,j})+b^{(g)}_{i,2}\cos(X_{i,j}) \right\} + \epsilon^{(g)}_{i,j},
\end{align} 
where $[m^{(g)}_{i,0}, b^{(g)}_{i,1}, b^{(g)}_{i,2}]^T \sim \mathcal{N}(0, \Psi^{(g)})$ and are independent of the random noise $\epsilon^{(g)}_{i,j}$ for all $g$ \citep{Archer2014, delolmo2022, Fontana2012, Hou2021, MllerLevet2013}. Under these assumptions,
\begin{align*}
    \boldsymbol{Y_i^{(g)}} \sim \mathcal{N}(\boldsymbol{W_i}\beta^{(g)}, \boldsymbol{V_i^{(g)}}),
\end{align*} 
where $\beta^{(g)} =[{\mu}^{(g)}_0, {\beta}^{(g)}_1,  {\beta}^{(g)}_2]^T$ denotes the vector of fixed effect estimands for the $g$-th gene,
\begin{align*}
    \boldsymbol{W_i} &= \begin{bmatrix} 1 & \sin\left(\frac{\pi X_{i,1}}{12}\right) & \cos\left(\frac{\pi X_{i,1}}{12}\right) \\
    \vdots & \vdots & \vdots \\
    1 & \sin\left(\frac{\pi X_{i,n_i}}{12}\right) & \cos\left(\frac{\pi X_{i,n_i}}{12}\right)
    \end{bmatrix}
\end{align*} 
denotes the design matrix for estimating fixed effects, and
\begin{align*}
\boldsymbol{V_i^{(g)}} = (\sigma^{(g)})^2I_{n_i} + \boldsymbol{W_i} \Psi^{(g)}\boldsymbol{W_i}^T
\end{align*} 
denotes the covariance matrix of the $g$-th gene's expression levels for the $i$-th individual, with $I_{n_i}$ denoting an $n_i \times n_i$ identity matrix. The corresponding likelihood function that is maximized for estimating the parameters of this model is defined as
\begin{align*}
    L(\beta, \Psi, \sigma^2 \mid  \boldsymbol{W_1},\ldots,\boldsymbol{W_M}, \boldsymbol{Y_1^{(g)}},\ldots,\boldsymbol{Y_M^{(g)}}) = \prod_{i=1}^M\frac{\exp\left\{-\frac{(\boldsymbol{Y_i^{(g)}}-\boldsymbol{W_i}\beta)^T(\boldsymbol{V_i})^{-1}(\boldsymbol{Y_i^{(g)}}-\boldsymbol{W_i}\beta)}{2} \right\}}{(2\pi)^{n_i/2} \sqrt{|\boldsymbol{V_i}|}},
\end{align*} 
where $\boldsymbol{V_i} = \sigma^2I_{n_i} + \boldsymbol{W_i} \Psi\boldsymbol{W_i}^T$. When $(\sigma^{(g)})^2$ and $\Psi^{(g)}$ are known, the expression
\begin{align*}
    \hat{\beta}^{(g)} = \begin{bmatrix}
    \hat{\mu}^{(g)}_0 \\
    \hat{\beta}^{(g)}_1 \\
    \hat{\beta}^{(g)}_2
    \end{bmatrix} = \left(\sum_{i=1}^M\boldsymbol{W_i}^T(\boldsymbol{V_i^{(g)}})^{-1}\boldsymbol{W_i}\right)^{-1}\left(\sum_{i=1}^M\boldsymbol{W_i}^T(\boldsymbol{V_i^{(g)}})^{-1}\boldsymbol{Y_i^{(g)}}\right)
\end{align*} 
yields maximum likelihood estimates of the fixed effects [\citealp[Page 78]{Davidian1995}; \citealp[Equation 6.19]{McCulloch2000}]. When $(\sigma^{(g)})^2$ and $\Psi^{(g)}$ are unknown, it is not possible to obtain a closed-form expression of these parameter estimates, and an expectation-maximization algorithm is utilized to obtain them [\citealp[Chapter 4]{Hedeker2006}; \citealp{Laird1987}; \citealp[Chapter 14]{McCulloch2000}].

After parameter estimation, an investigator could perform hypothesis tests to assess the statistical significance of a gene's oscillatory behavior. This paper considers hypothesis testing with the fixed effects parameter vector $\hat{\beta}^{(g)}$, with which Wald-type test statistics can be computed as
\begin{align}
    \tau^{(g)} = (\hat{\beta}^{(g)} - \hat{\beta}^{(g)}_{\text{Null}})^TL^T\left[L\left\{\sum_{i=1}^M\boldsymbol{W_i}^T(\boldsymbol{\hat{V}_i^{(g)}})^{-1}\boldsymbol{W_i}\right\}^{-1}L^T\right]^{-1}L(\hat{\beta}^{(g)} - \hat{\beta}^{(g)}_{\text{Null}}). \label{eq:wald}
\end{align} 
Here, $L$ denotes a $d\times 3$ matrix, where $d\leq 3$; $\boldsymbol{\hat{V}_i^{(g)}}$ denotes an empirical estimate of the matrix $\boldsymbol{V_i^{(g)}}$; and $\hat{\beta}^{(g)}_{\text{Null}}$ denotes the empirical parameter vector estimate under the null hypothesis, which is $H_0: L(\beta^{(g)} - \beta^{(g)}_{\text{Null}}) = 0$ \citep{Halekoh2014}. This paper defines 
\begin{align*}
    L = \begin{bmatrix}
        0 & 1 & 0 \\
        0 & 0 & 1
    \end{bmatrix}
\end{align*} 
and $\hat{\beta}^{(g)}_{\text{Null}} = [0, 0, 0]^T$, which can be considered a determination of whether or not a gene displays oscillatory behavior \citep{Zong2023}. This test statistic has an asymptotic $\chi^2_2$ distribution, and the null hypothesis would be rejected at a pre-determined $\alpha$-level if $\tau^{(g)}$ surpasses the $1-\alpha$ percentile of this distribution \citep{Halekoh2014}. 

\subsection{Fixed effects estimation in the presence of random phase-shifts} \label{sec:2.2}
A long-standing challenge in estimating a cosinor model with data from multiple individuals is that these quantities are biased when each individual has a distinct, unknown phase-shift parameter, or there is ``phase variation'' \citep{Marron2015, Ramsay2003, Sollberger1962, Weaver1995}. Recent investigations into this phase variation were driven by the discrepancy between an individual's internal circadian time (ICT) and their recorded time (ZT), which was interpreted as covariate measurement error or misrecorded time data. Based on this interpretation, empirical and theoretical analyses were conducted to determine how cosinor model parameter estimates and test statistics are biased when ZT is used for regression \citep{Gorczycab2024, Gorczycaa2024}. Building upon these insights, this paper evaluates the accuracy of approximating the nonlinear mixed effects cosinor model in (\ref{eq:cos_m}) with the linear mixed effects cosinor model in (\ref{eq:lin_m}) when each individual's random phase-shift parameter is unknown. Specifically, we first introduce the following proposition, which quantifies this accuracy.

\begin{prop} \label{prop:1}
Suppose each $n_i = n$ and $X_{i,j} = 24(j-1)/n$. Assume the model in (\ref{eq:lin_m}) is specified to estimate the response in (\ref{eq:cos_m}), with both $(\sigma^{(g)})^2$ and $\Psi^{(g)}$ known such that $\Psi^{(g)} = \diag(\psi^{(g)}_1, \psi^{(g)}_2, \psi^{(g)}_3)$, or equals a $3\times 3$ diagonal matrix where $\psi^{(g)}_k$ denotes the $k$-th element along the diagonal. Further, assume $[m_{i,0}^{(g)}, c^{(g)}_{i,1}, c^{(g)}_{i,2}]^T \sim P^{(g)}$, where the corresponding probability density function of $P^{(g)}$ is symmetric around zero. Then the expectation of the fixed parameter estimates
\begin{align*}
    \mathbb{E}(\hat{\beta}^{(g)}) = \begin{bmatrix}
    \mathbb{E}(\hat{\mu}^{(g)}_0) \\
    \mathbb{E}(\hat{\beta}^{(g)}_1) \\
    \mathbb{E}(\hat{\beta}^{(g)}_2)
    \end{bmatrix} = \begin{bmatrix}
    \mu^{(g)}_0 \\
    \phi_{c^{(g)}_2}(1)\beta^{(g)}_1 \\
    \phi_{c^{(g)}_2}(1)\beta^{(g)}_2
    \end{bmatrix},
\end{align*} 
where $\phi_{c^{(g)}_2}(t)$ denotes the characteristic function for $c^{(g)}_{i,2}$ evaluated at $t$. Further, the Wald test statistic in (\ref{eq:wald}) computed with the expectation of these parameter estimates can be expressed as
\begin{align*}
    \tau^{(g)} 
    &= \frac{Mn\phi^2_{c^{(g)}_2}(1)}{2} \left\{\left(\frac{1}{(\sigma^{(g)})^2} - \frac{n\psi^{(g)}_2}{(\sigma^{(g)})^2\{n\psi^{(g)}_2+2(\sigma^{(g)})^2\}}\right)(\beta_1^{(g)})^2 \right. \\
    &\quad \quad + \left. \left(\frac{1}{(\sigma^{(g)})^2} - \frac{n\psi^{(g)}_3}{(\sigma^{(g)})^2(n\psi^{(g)}_3+2(\sigma^{(g)})^2\}}\right)(\beta^{(g)}_2)^2 \right\}.
\end{align*} 
\end{prop}

\noindent A derivation for this result is provided in \ref{app:B}. To clarify the setup for this result, when each $X_{i,j} = 24(j-1)/n$, tissue samples are collected from an equispaced design. This design is optimal for minimizing the estimation variance of a cosinor model in (\ref{eq:3}) under multiple statistical criteria \citep[Pages 241-243]{Pukelsheim2006} as well as for the mixed effects cosinor model in (\ref{eq:lin_m}) \citep{Mentre1997}. Further, when each $n_i = n$, no samples are removed prior to estimation, which can occur when a collected tissue sample is poor quality \citep{Laurie2010}. It is noted that symmetry is assumed for the probability density function of $P^{(g)}$. As a consequence, the marginal density function for $c^{(g)}_{i,2}$ is also symmetric, which results in $\phi_{c_2^{(g)}}(t)$ being a real-valued function that is bounded by one in magnitude.

The significance of Proposition \ref{prop:1} is it presents a scenario where an investigator employs a common experimental design and has additional knowledge of $(\sigma^{(g)})^2$ and $\Psi^{(g)}$. In this scenario, the magnitude of the fixed parameter estimates $\hat{\beta}_1^{(g)}$ and $\hat{\beta}_2^{(g)}$ are biased towards zero, as the characteristic function $\phi_{c_2^{(g)}}(t)$ is bounded by one in magnitude. Further, application of (\ref{eq:alt_to_orig}) yields
\begin{align*}
    \sqrt{\mathbb{E}(\hat{\beta}^{(g)}_1)^2 + \mathbb{E}(\hat{\beta}^{(g)}_2)^2}  = |\phi_{c^{(g)}_2}(1)|\theta^{(g)}_1, \quad \atan\{-\mathbb{E}(\hat{\beta}^{(g)}_1), \mathbb{E}(\hat{\beta}^{(g)}_2)\}  = \theta^{(g)}_2,
\end{align*} 
which indicates that the population amplitude parameter estimate is biased, while the population phase-shift parameter estimate is unbiased. These results also indicate that the test statistic $\tau^{(g)}$ is attenuated by $\phi^2_{c^{(g)}_2}(1)$ relative to the test statistic obtained in the scenario where there is no phase variation. It is noted that recent investigations derived the same bias in model parameter estimates for the cosinor model in (\ref{eq:3}; \citealt{Gorczycaa2024}). Further, when there is no phase-shift variation and each sample is independent, the test statistic simplifies to $\tau^{(g)} = \{Mn(\theta^{(g)}_1)^2\}/\{2(\sigma^{(g)})^2\}$, which has been identified in other research efforts \citep{Silverthorne2024, Zong2023}.

To understand when the linear mixed effects model in (\ref{eq:lin_m}) produces unbiased parameter estimates if the random phase-shifts are unknown, an alternate representation that is equivalent to it is
\begin{align} \label{eq:res_2}
Y^{(g)}_{i,j} &= (\mu^{(g)}_0 + m^{(g)}_{i, 0}) + \theta^{(g)}_1\cos(X_{i,j}+\theta^{(g)}_{2}) + c^{(g)}_{i,1}\cos(X_{i,j}+c^{(g)}_{i,2}) + \epsilon^{(g)}_{i,j},
\end{align} 
which separates fixed and random effects into different nonlinear components \citep{Mikulich2003}. The following proposition transforms the model in (\ref{eq:res_2}) into the model in (\ref{eq:cos_m}).
\begin{prop} \label{prop:2}
The model in (\ref{eq:res_2}) can be expressed as 
\begin{align*}
Y^{(g)}_{i,j} &= (\mu^{(g)}_0 + m^{(g)}_{i, 0}) + \sqrt{\left\{\theta^{(g)}_1\sin(\theta^{(g)}_2) + c^{(g)}_{i,1}\sin(c^{(g)}_{i,2})\right\}^2 + \left\{\theta^{(g)}_1\cos(\theta^{(g)}_2) + c^{(g)}_{i,1}\cos(c^{(g)}_{i,2})\right\}^2} \nonumber \\
    & \quad \quad \times \cos\left[\frac{\pi X_{i,j}}{12} + \atan\left\{\theta^{(g)}_1\sin(\theta^{(g)}_2) + c^{(g)}_{i,1}\sin(c^{(g)}_{i,2}), \theta^{(g)}_1\cos(\theta^{(g)}_2) + c^{(g)}_{i,1}\cos(c^{(g)}_{i,2})\right\}\right] + \epsilon^{(g)}_{i,j},
\end{align*} 
where $\atan(a, b)$ denotes the two argument arctangent function.
\end{prop}
\noindent A derivation for this result is provided in \ref{app:C}. Proposition \ref{prop:2} implies that when each individual's random phase-shift is unknown, population parameter estimates obtained from specifying the linear mixed effects model in (\ref{eq:lin_m}) are biased unless the conditions
\begin{align*}
\theta_1^{(g)} &= \mathbb{E}\left[\sqrt{\left\{\theta^{(g)}_1\sin(\theta^{(g)}_2) + c^{(g)}_{i,1}\sin(c^{(g)}_{i,2})\right\}^2 + \left\{\theta^{(g)}_1\cos(\theta^{(g)}_2) + c^{(g)}_{i,1}\cos(c^{(g)}_{i,2})\right\}^2} \right], \\
\theta_2^{(g)} &= \mathbb{E}\left[\atan\left\{\theta^{(g)}_1\sin(\theta^{(g)}_2) + c^{(g)}_{i,1}\sin(c^{(g)}_{i,2}), \theta^{(g)}_1\cos(\theta^{(g)}_2) + c^{(g)}_{i,1}\cos(c^{(g)}_{i,2})\right\} \right]
\end{align*} 
are satisfied.

% I am here
\subsection{Overview of method} \label{sec:2.3}
Propositions \ref{prop:1} and \ref{prop:2} highlight that the parameter estimates of a linear mixed effects cosinor model and test statistics computed with this model are biased when there is phase variation and each individual's phase-shift parameter is unknown. A practical consequence of not addressing these biases is that study conclusions made from a linear mixed effects cosinor model would be inaccurate. For example, the attenuation bias in parameter estimates and test statistics presented in Proposition \ref{prop:1} would result in genes being misclassified as having statistically insignificant oscillatory behavior in a circadian biology study.

To address this bias, an investigator could collect a sufficient number of samples from each individual to estimate individual-specific cosinor models for each gene. In this scenario, an investigator would then utilize two-stage methods, which leverage individual-specific nonlinear model parameter estimates as building blocks for computing population model parameters and test statistics \citep[Chapter 5]{Davidian1995}. A challenge in applying these methods to the cosinor model is that random phase-shifts are distinct across genes and not necessarily indicative of dim-light melatonin onset (DLMO) time across individuals. For example, approximately 50\% of mammalian genes display tissue-dependent oscillatory behavior \citep{Mure2018, Ruben2018, Zhang2014}. Additionally, some biological processes, such as the secretion of some hormones relevant to thyroid function or pregnancy, exhibit minimal diurnal variations \citep{Kennaway2019, Kennaway2020}. This lack of oscillatory behavior implies that individual-specific phase-shift estimates for these genes could represent random noise and not align with each individual's DLMO time. Further, a small sample size can exacerbate the unreliability of these parameter estimates \citep{Brooks2023}.

To recover parameter estimates and inferences that would be obtained when each individual's DLMO time is known, we propose is a method based on two-stage methods that involves the following steps. For each step, relevant comments are provided to clarify their significance and implications, and an algorithmic summary is provided in Algorithm \ref{salgo:1}:
\begin{description}
    \item[Step 1.] For each gene (indexed by $g$), estimate a linear mixed effects cosinor model using every individual's data, or $(\boldsymbol{X_i}, \boldsymbol{Y^{(g)}_i}; i = 1,\ldots, M)$.
    \item[Step 2.] For each gene-specific linear mixed effects cosinor model estimated in Step 1, compute the population phase-shift parameter estimate ($\hat{\theta}_{2}^{(g)}$) with the identities in (\ref{eq:alt_to_orig}). 
    \item[Step 3.] For each gene (indexed by $g$), estimate a linear cosinor model with each individual's data (indexed by $i$), or $(\boldsymbol{X_i}, \boldsymbol{Y^{(g)}_i})$.
    \item[Step 4.] For each individual-specific and gene-specific cosinor model estimated in Step 3, compute the individual-specific phase-shift parameter estimate, denoted as $\hat{\theta}_{i,2}^{(g)}$, with the identities in (\ref{eq:alt_to_orig}). 
    \item[Step 5.] Define $$\hat{w}^{(g)}_i = \frac{\frac{1}{\mathrm{Var}(\hat{\theta}^{(g)}_{i,2})}}{\frac{1}{\mathrm{Var}(\hat{\theta}^{(g)}_2)}+\frac{1}{\mathrm{Var}(\hat{\theta}^{(g)}_{i,2})}}$$and 
    \begin{align*}
        \hat{d}^{(g)}_{i} &= \left[\mathrm{atan2}\left\{\hat{w}_i^{(g)}\sin(\hat{\theta}^{(g)}_{i,2}) + (1-\hat{w}_i^{(g)})\sin(\hat{\theta}^{(g)}_{2}), \right. \right. \\
        & \quad \left. \left. \hat{w}_i^{(g)}\cos(\hat{\theta}^{(g)}_{i,2}) + (1-\hat{w}_i^{(g)})\cos(\hat{\theta}^{(g)}_{2})\right\} - \hat{\theta}^{(g)}_2\right] \mathrm{mod} \ 2\pi.
    \end{align*} 
    Here, $\mathrm{Var}(\hat{\theta}^{(g)}_2)$ denotes the estimated variance of $\hat{\theta}^{(g)}_2$, and $\mathrm{Var}(\hat{\theta}^{(g)}_{i,2})$ denotes the estimated variance of $\hat{\theta}^{(g)}_{i,2}$.
    \item[Step 6.]  Define $\boldsymbol{\tilde{X}_i} = \boldsymbol{X_i}+\tilde{d}_{i}$, where $$\tilde{d}_{i} = \frac{12}{\pi}\mathrm{atan2}\left\{\left(\frac{1}{\sum_{g=1}^G\frac{1}{1-\Omega^{(g)}_{ }}}\right)\sum_{g=1}^G\frac{\sin(\hat{d}_{i}^{(g)})}{1-\Omega^{(g)}_{ }},\ \left(\frac{1}{\sum_{g=1}^G\frac{1}{1-\Omega^{(g)}_{ }}}\right)\sum_{g=1}^G\frac{\cos(\hat{d}_{i}^{(g)})}{1-\Omega^{(g)}_{ }}\right\},$$with \begin{align} \Omega^{(g)}_{ } = \sqrt{\left\{\frac{1}{M}\sum_{i=1}^M\sin(\theta_{i,2}^{(g)})\right\}^2+\left\{\frac{1}{M}\sum_{i=1}^M\cos(\theta_{i,2}^{(g)})\right\}^2}. \label{eq:R_val} \end{align} 
    \item[Step 7.] Estimate a linear mixed effects cosinor model with gene-specific translated data from every individual, or $\{(\boldsymbol{\tilde{X}_i}, \boldsymbol{Y_i^{(g)}}), i = 1,\ldots, M\}$ .
\end{description}
\noindent Steps 1-4 obtain individual-specific and population phase-shift parameter estimates, denoted as $\hat{\theta}_{i,2}^{(g)}$ and $\hat{\theta}_{i}^{(g)}$, respectively. Step 5 then computes a weighted circular average of $\hat{\theta}_{i,2}^{(g)}$ and $\hat{\theta}_2^{(g)}$, where the corresponding unweighted average $\bar{\omega}$ would be defined as
\begin{align*}
    \bar{\omega} = \mathrm{atan2}\left\{\frac{1}{N}\sum_{j=1}^N\sin(\omega_j), \frac{1}{N}\sum_{j=1}^N\cos(\omega_j)\right\}
\end{align*} 
given measurements $\omega_j$, with $j=1,\ldots, N$ \citep[Section 2.2]{Mardia1999}. Here, $\hat{w}^{(g)}_i$ is an inverse-variance weight based on the phase-shift estimates \citep{Cochran1953}, where the variance of a phase-shift estimate is derived in \ref{app:D}. The difference between this weighted average and $\hat{\theta}_2^{(g)}$ is then computed modulo $2\pi$. The motivation for Step 5 is that if $\hat{d}_i^{(g)}$ is computed without accounting for the circular nature of the phase-shift parameter, then
\begin{align}
    \hat{d}^{(g)}_{i} &= \hat{w}_i^{(g)}\hat{\theta}^{(g)}_{i,2} + (1-\hat{w}_i^{(g)})\hat{\theta}^{(g)}_{2} - \hat{\theta}^{(g)}_2 \nonumber \\
    &= \hat{w}_i^{(g)}(\hat{\theta}^{(g)}_{i,2}-\hat{\theta}^{(g)}_{2}). \label{eq:s5}
\end{align} 
If it is also assumed that $\theta_{i,2}^{(g)} = \theta_{2}^{(g)} + c_{i,2}^{(g)}$, which is common in nonlinear mixed effects modelling \citep[Equation 4.4]{Davidian1995}, then (\ref{eq:s5}) implies that $d^{(g)}_i = w_i^{(g)}c_{i,2}^{(g)}$, where the weight $w_i^{(g)}$ increases as variance in estimation of $\hat{\theta}^{(g)}_{i,2}$ decreases. It is noted that this weighting has been interpreted as shrinkage estimation of individual-specific parameter estimates toward population estimates \citep{Davidian1995, Steimer1984}. In Step 6, an additional inverse-variance weighted circular average of each $\hat{d}_{i}^{(g)}$ is computed across all $G$ genes. To clarify, the quantity $1-\Omega^{(g)}_{ }$ is the circular variance of each $\theta^{(g)}_{i,2}$ across all $M$ individuals \citep[Section 2.3]{Mardia1999}, where $\Omega^{(g)}_{ }$ is defined in (\ref{eq:R_val}). The times of sample collection for the $i$-th individual are then translated by the quantity obtained from this weighted average, which is denoted as $\tilde{d}_{i}$. In Step 7, a linear mixed effects model is estimated with this translated data.

It is noted that the population phase-shift estimates $\theta_2^{(g)}$ obtained with this method can be different from those obtained with a linear mixed effects model where the times of sample collection are not translated. If an investigator is interested in analyzing estimates of $\theta_2^{(g)}$ after translation, this paper recommends that an additional adjustment is performed to ensure each $\hat{\theta}_2^{(g)}$ is the same regardless of whether $\boldsymbol{\tilde{X}_i}$ or $\boldsymbol{X_i}$ is used for estimation. This recommendation is supported by Proposition \ref{prop:1}, which implies that phase-shift estimates are unbiased when generated by a symmetric probability distribution. 

\section{Simulation Study} \label{sec:3}

A simulation study is conducted to evaluate the method proposed in Section \ref{sec:2.3}. Briefly, six settings from the $\text{BioCycle}_{\text{Form}}$ simulation software are considered, where each setting represents a scenario in which a simulated gene displays periodic behavior over a 24 hour interval \citep{Agostinelli2016, Ceglia2018}. It is noted that for the first simulation setting in this study, the cosinor model correctly represents the oscillatory behavior of a gene. For the remaining five simulation settings, the cosinor model is mis-specified. Additional details about each simulation setting are provided in \ref{app:E}.

For each setting in this simulation study, 2,000 simulation trials are performed. Further, in each simulation trial, two data sets are generated: one in which each individual has a random phase-shift parameter, and one in which there are no random phase-shift parameters. It is emphasized that this simulation study assumes each individual's dim-light melatonin onset (DLMO) time equals that individual's random phase-shift parameter. As a consequence, the second data set represents a scenario where each individual's DLMO time is known. 

This simulation study considers three estimation frameworks to obtain linear mixed effects cosinor models from these two data sets:
\begin{description}
    \item[Framework 1.] A linear mixed effects cosinor model is estimated using the method in Section \ref{sec:2.3} with the data set where each individual has an unknown random phase-shift parameter.
    \item[Framework 2.] A linear mixed effects cosinor model is estimated with the data set where each individual has an unknown random phase-shift parameter.
    \item[Framework 3.] A linear mixed effects cosinor model is estimated with the data set where there is no phase variation.
 \end{description}
It is also emphasized that only one gene is generated in each simulation trial. When a single gene is available, Step 6 of the method presented in Section \ref{sec:2.3} (Framework 1) would set $\tilde{d}_{i} = \hat{d}^{(g)}_{i}$. 

After conducting 2,000 simulation trials, the mean and standard deviation of the following quantities are reported:
\begin{itemize}
    \item[1.] $\hat{\theta}_1^{(g)}$, or the estimated amplitude in a simulation trial.
    \item[2.] $\tau^{(g)}$, or the computed Wald test statistic from (\ref{eq:wald}) in a simulation trial.
\end{itemize} 
It is noted that the phase-shift estimate is not considered, as Framework 1 can be modified to produce the same phase-shifts as Framework 2.

Table \ref{tab:sim1} presents the results for each simulation setting. Framework 1, or the method in Section \ref{sec:2.3}, consistently mitigates bias in amplitude estimates and test statistics produced in the presence of phase variation. Framework 2, on the other hand, produces corresponding quantities that suffer from attenuation bias.

\section{Real Data Illustrations} \label{sec:4}

Data from three longitudinal circadian biology studies are considered for illustration: Archer and colleagues produced data from two cohorts (a control cohort and an intervention cohort) \cite{Archer2014}; Braun and colleagues produced data from a single cohort \cite{Braun2018}; and M\"{o}ller-Levet and colleagues produced data from two cohorts (a control cohort and an intervention cohort) \cite{MllerLevet2013}. Each data set has been described in detail in their respective studies and have been previously summarized \cite{Gorczycab2024}. Each data set has also been processed in subsequent studies and made publicly available \cite{Huang2024}. The significance of each data set is that their corresponding studies performed laboratory tests to determine the time that each participant starts to produce melatonin under dim-light conditions (dim-light melatonin onset time, or DLMO time). Further, each study considered a participant's DLMO time to be the offset of their internal circadian time (ICT) relative to Zeitgeber time (ZT). It is emphasized that gene expression measurements and offsets are unavailable for some participants in these processed data sets. In this illustration, genes with missing expression measurements and samples from participants with missing offsets within cohort-specific data are excluded before cosinor model estimation and inference. It is also noted that this illustration constructs two additional cohorts: one where data from both cohorts produced by Archer and colleagues are combined, and another where data from both cohorts produced by M\"{o}ller-Levet and colleagues are combined. 

The illustrations follow the evaluation procedure from Section \ref{sec:3}. Specifically, for each sample population, fixed amplitude parameters from (\ref{eq:cos_m}) and Wald test statistics from (\ref{eq:wald}) are computed from Frameworks 1 and 2 given ZT. Framework 3 computes these same quantities given ICT, which are considered the true quantities in this illustration. It is emphasized that each cohort-specific data set consists of multiple genes. To account for this, a linear model-based evaluation is presented in this illustration \citep{Gorczycab2024}. In this evaluation, the covariate of this linear model is specified as a quantity obtained from a framework given ZT, while the response variable is the corresponding true quantity obtained from Framework 3. No intercept term is specified in this linear model, which results in a single regression parameter estimate, $\hat{\gamma}$. The parameter estimate $\hat{\gamma}$ is then used to assess the relationship between the quantities obtained from each framework. If $\hat{\gamma}=1$, a linear relationship exists between a quantity obtained with ZT and a quantity obtained with ICT. If $\hat{\gamma} > 1$, then quantities obtained with ICT are consistently larger than quantities obtained with ZT. If $\hat{\gamma} < 1$, then quantities obtained with ZT are consistently larger than quantities obtained with ICT. This assessment also reports the coefficient of determination ($R^2$) for each linear model, where higher $R^2$ values signify greater precision in linear model fit. In this illustration, a stronger performing framework would obtain $\hat{\gamma}$ a value closer to one and a large $R^2$ value. It is noted that this evaluation is twice. The first evaluation is performed with quantities obtained using every single gene available in a data set. The second evaluation is performed with every ``gold standard circadian gene'' known to display oscillatory behavior available in a data set \citep{Mei2020}.

Table \ref{tab:app} presents linear model parameter estimates $\hat{\gamma}$ and coefficients of determination $R^2$ computed from each cohort's data when the linear model is computed with every gene available. Framework 1, or the method proposed in Section \ref{sec:2.3}, consistently reduces bias in amplitude estimation and Wald test statistic calculation, or produced $\hat{\gamma}$ quantities that are closer to one in numeric value than those produced with Framework 2. Framework 1 outperforms Framework 2 in amplitude estimation for six of the seven sample populations (Framework 2 outperforms Framework 1 on intervention cohort data produced by Archer and colleagues), and Wald test statistic calculation on every data set. Table \ref{tab:app} highlights that the estimated $\hat{\gamma}$ values are consistently larger than one for Framework 2, which implies that estimation of a linear mixed effects model given ZT results in attenuated parameter estimates and test statistics. Corresponding scatter plots for linear model fits are provided in Figures \ref{fig:amp_w} (for amplitude estimates) and \ref{fig:test_w} (for Wald test statistics).

Table \ref{tab:app2} presents linear model parameter estimates $\hat{\gamma}$ and coefficients of determination $R^2$ computed from each cohort's data when the linear model is computed with every gold standard circadian gene available. Similar to Table \ref{tab:app}, Framework 1 obtains $\hat{\gamma}$ values closer to one on five sample populations for amplitude estimation, and six of the seven sample populations for Wald test statistic calculation. Corresponding scatter plots for linear model fits are provided in Figures \ref{fig:amp_g} (for amplitude estimates) and \ref{fig:test_g} (for Wald test statistics).

\section{Discussion} \label{sec:5}

In this paper, a method is introduced that addresses the offset between an individual's internal circadian time (ICT) and the 24 hour day-night cycle time (Zeitgeber time, or ZT) when using a cosinor model to represent a gene's oscillatory behavior. Specifically, this method is designed for scenarios where multiple samples are gathered from each study participant over time, and utilizes phase-shift estimates from both population cosinor models and individual-specific cosinor models to adjust each individual's ZTs of sample collection prior to estimation and inference. To motivate this method, Proposition \ref{prop:1} reveals that the linear mixed effects cosinor model, which is prevalent in circadian biology studies \citep{Archer2014, delolmo2022, Fontana2012, Hou2021, MllerLevet2013}, fails to account for each individual's offset. As a consequence, the linear mixed effects cosinor model yields biased parameter estimates and inferences when specified. Empirical findings from Tables \ref{tab:app} and \ref{tab:app2} empirically validate this finding, as these tables show amplitude estimates and Wald test statistics computed with ZT are consistently smaller than when they are computed with ICT. The method proposed in this paper, on the other hand, consistently mitigates bias in amplitude estimation and test statistic calculation and produces quantities that are more similar to those produced with ICT. 

Despite the success of the method in simulation studies and real data illustrations, it is important to recognize the method's limitations. First, not every gene displays oscillatory behavior consistent with the cosinor model, and the true oscillatory behavior of many genes are unknown. As a consequence, the cosinor model only correctly represents genes that display sinusoidal oscillations. However, this sinusoidal behavior is commonly assumed in many studies \citep{Archer2014, delolmo2022, Fontana2012, Hou2021, MllerLevet2013}, and it enables interpretable parameter estimates and hypothesis tests. Second, this method does not guarantee that the translated times obtained from it will equal each individual's ICTs of sample collection. To clarify, the method translates each individual's ZTs of sample collection based on data-driven phase-shift estimates. However, this method consistently yields parameter estimates and inferences akin to those obtained when an investigator performs laboratory tests to determine each individual's dim-light melatonin onset time (DLMO time), a gold-standard biomarker for the offset of an individual's ICT relative to ZT.

The results presented in this paper open up avenues for future research. First, other frameworks have been proposed for identifying genes with oscillatory behavior \citep{Mei2020}. There could be value in studying how statistical analyses with these frameworks are affected by each study participant having an unknown DLMO time. Second, depending on the assumptions an investigator makes about the oscillatory behavior of a gene over time, there could be utility in extending this methodology to other modelling frameworks.

\section*{Acknowledgements}
The author would like to thank Forest Agostinelli at the University of South Carolina for conversation concerning the BioCycle simulation data sets, which supported development of the simulation study design in Section \ref{sec:3}; Thomas Brooks at the University of Pennsylvania for input on factors that could affect reproducibility in circadian biology studies, which helped refine the scope of this paper and method development; David Kennaway at the University of Adelaide for sharing insights on the cost of DLMO time determination and providing examples of biological processes that exhibit minimal diurnal variations, which enhanced the presentation of this paper; and Tavish McDonald at Lawrence Livermore National Laboratory for technical discussion and reviewing an initial version of this manuscript posted on arXiv, which improved clarity in presentation.

\section*{Data Availability Statement}
The author has made the code scripts used to produce the results presented in Section \ref{sec:4} available at 
\\
\textcolor{blue}{\href{https://bitbucket.org/michaelgorczyca/mixed_cosinor/}{https://bitbucket.org/michaelgorczyca/mixed\_cosinor/}}.

\newpage
\clearpage

\begin{algorithm}[H]
\DontPrintSemicolon
\footnotesize
\caption{Method presented in Section \ref{sec:2.3}.}
\label{salgo:1}
    \SetKwFunction{FCGT}{ComputeGeneTranslation}
    \SetKwFunction{CIT}{ComputeIndividualTranslation}
    \SetKwFunction{CT}{GetTranslatedModel}
    \SetKwFunction{SOne}{Step1}
    \SetKwFunction{STwo}{Step2}
    \SetKwFunction{SThr}{Step3}
    \SetKwFunction{SFou}{Step4}
    \SetKwFunction{FR}{ComputeOmega}
    \SetKwFunction{GMQ}{Step3andStep4}
    \SetKwFunction{at}{atan2}
    \SetKwProg{Fn}{Function}{:}{}
    \SetCommentSty{itshape}
    \Fn{\FCGT{$\hat{\theta}^{(g)}_{2}$, $\hat{\theta}^{(g)}_{i,2}$, $\hat{w}_i^{(g)}$}}{
        \tcc{Performs Step 5 given population phase-shift $\hat{\theta}^{(g)}_{2}$, individual-specific phase-shift $\hat{\theta}^{(g)}_{i,2}$, and inverse-variance weight for individual-specific phase-shift $\hat{w}_i^{(g)}$. Returns the quantity $\hat{d}_i^{(g)}.$}
        $s \gets \hat{w}^{(g)}_i\sin(\theta^{(g)}_{i,2})+(1-\hat{w}^{(g)}_i)\sin(\theta^{(g)}_{2})$\;
        $c \gets \hat{w}^{(g)}_i\cos(\theta^{(g)}_{i,2})+(1-\hat{w}^{(g)}_i)\cos(\theta^{(g)}_{2})$\;
        $\hat{d}_i^{(g)} \gets \{\mathrm{atan2}(s, c) - \hat{\theta}^{(g)}_{2}\}$ \ mod $2\pi$\;
        \KwRet $\hat{d}_i^{(g)}$\;
  }
  \Fn{\CIT{$\hat{d}_i^{(1)},\ldots, \hat{d}_i^{(G)}, \mathrm{Var}(\hat{\theta}^{(g)}_{i,2}), \ldots, \mathrm{Var}(\hat{\theta}^{(g)}_{i,2})$}}{
        \tcc{Computes the quantity $\tilde{d}_i$ from Step 6 given quantities $\hat{d}_i^{(g)}$, $g=1,\ldots, G$ from Step 5 and individual-specific amplitude variance estimates $\mathrm{Var}(\hat{\theta}^{(g)}_{i,2})$. Returns the quantity $\tilde{d}_i.$}
        $\xi \gets 0$\;
        \For{$g \gets 1$ \KwTo $G$}{
            $\xi \gets \xi + \left\{1/\mathrm{Var}(\hat{\theta}^{(g)}_{i,2})\right\}$\;
        }
        $s \gets 0$\;
        $c \gets 0$\;
        \For{$g \gets 1$ \KwTo $G$}{
            $s \gets s + \left\{\sin(\hat{d}^{(g)}_{i})/\mathrm{Var}(\hat{\theta}^{(g)}_{i,2})\right\}/\xi$\;
            $c \gets c + \left\{\cos(\hat{d}^{(g)}_{i})/\mathrm{Var}(\hat{\theta}^{(g)}_{i,2})\right\}/\xi$\;
        }
        $\tilde{d}_i \gets \mathrm{atan2}(s, c)$\;
        \KwRet $\tilde{d}_i$\;
  }
  \Fn{\CT{$\boldsymbol{X_1}, \ldots, \boldsymbol{X_M}, \boldsymbol{Y_1^{(1)}}, \ldots, \boldsymbol{Y_M^{(1)}}, \ldots, \boldsymbol{Y_1^{(G)}}, \ldots, \boldsymbol{Y_M^{(G)}}$}}{
        \tcc{Computes linear mixed effects models, which will be denoted as $\mathcal{M}^{(g)}$ for the $g$-th gene, given covariate data $\boldsymbol{X_1}, \ldots, \boldsymbol{X_M}$ and gene expression data $\boldsymbol{Y_1^{(g)}}, \ldots, \boldsymbol{Y_M^{(g)}}$. Takes as input covariate data and gene expression data across all $G$ genes and all $M$ people. Returns linear mixed effects models for each gene estimated from translated data $\boldsymbol{\tilde{X}_i}$, which is denoted as $\tilde{\mathcal{M}}^{(g)}$ for the $g$-th gene.} % [\hat{\theta}_1^{(g)}, \hat{\theta}_2^{(g)}, \mathrm{Var}(\hat{\theta}_1^{(g)})]
        \For{$g \gets 1$ \KwTo $G$}{
            $\mathcal{M}^{(g)} \gets $ \SOne{$\boldsymbol{X_1}$, $\boldsymbol{Y_1^{(g)}}$,$\ldots$, $\boldsymbol{X_M}$, $\boldsymbol{Y_M^{(g)}}$} \ilc{Implementation omitted, compute mixed model}\;
            $[\hat{\theta}_2^{(g)}, \mathrm{Var}(\hat{\theta}_2^{(g)})] \gets $ \STwo{$\mathcal{M}^{(g)}$} \ilc{Implementation omitted, get relevant parameters}\;
            \For{$i \gets 1$ \KwTo $M$}{ % [\hat{\theta}_{i, 1}^{(g)}, \hat{\theta}_{i, 2}^{(g)}, \mathrm{Var}(\hat{\theta}_{i, 1}^{(g)})]
                $\mathcal{M}^{(g)}_i \gets $ \SThr{$\boldsymbol{X_i}$, $\boldsymbol{Y_i^{(g)}}$} \ilc{Implementation omitted, compute individual model}\;
                $[\hat{\theta}_{i, 2}^{(g)}, \mathrm{Var}(\hat{\theta}_{i, 2}^{(g)})] \gets $ \SFou{$\mathcal{M}^{(g)}_i$} \ilc{Implementation omitted, get relevant parameters}\;
                $\hat{w}^{(g)}_i \gets \{1/\mathrm{Var}(\hat{\theta}_2^{(g)})\}/[\{1/\mathrm{Var}(\hat{\theta}_2^{(g)})\} + \{1/\mathrm{Var}(\hat{\theta}_{i,2}^{(g)})\}]$\;
                $\hat{d}^{(g)}_i \gets $ \FCGT{$\hat{\theta}_2^{(g)},\hat{\theta}_{i, 2}^{(g)}, \hat{w}^{(g)}_i$}\;
            }
            $\Omega^{(g)}_{ } \gets $ \FR{$\hat{\theta}_{1,2}^{(g)}, \ldots, \hat{\theta}_{M,2}^{(g)}$} \ilc{Implementation omitted, obtain quantity from (\ref{eq:R_val})}
        }
        \For{$i \gets 1$ \KwTo $M$}{
            $\tilde{d}_i \gets $ \CIT{$\hat{d}_i^{(1)}, \ldots, \hat{d}_i^{(G)}, R^{(1)},\ldots, \Omega^{(g)}_{ }$}\;
             $\boldsymbol{\tilde{X_i}} \gets \boldsymbol{X_i}+\tilde{d}_i$\;
        }
        \For{$g \gets 1$ \KwTo $M$}{
            $\tilde{\mathcal{M}}^{(g)} \gets $ \SOne{$\boldsymbol{\tilde{X}_1, \boldsymbol{Y_1^{(g)}}}, \ldots, \boldsymbol{\tilde{X}_M, \boldsymbol{Y_M^{(g)}}}$} \ilc{Implementation omitted}\;
        }
        \KwRet $[\tilde{\mathcal{M}}^{(1)}, \ldots, \tilde{\mathcal{M}}^{(G)}]$\;
   }
\end{algorithm}

\newpage
\clearpage

\begin{table*}[!h]
	\caption{Simulation study results for each framework. Framework 1 uses the method proposed in Section \ref{sec:2.3} given Zeitgeber time (ZT), Framework 2 performs linear mixed effects cosinor regression given ZT, and Framework 3 performs linear mixed effects cosinor regression given internal circadian time (ICT). The mean amplitudes and test statistics across 2,000 simulation trials are reported, with their corresponding standard deviations in parentheses. For Frameworks 1 and 2, mean quantities that are closest to the corresponding mean quantity obtained with Framework 3 are denoted in bold.} \label{tab:sim1}
 \centering
 %\resizebox{1.05\textwidth}{!}{
		\begin{tabular}{|c|c|c|c|c|c|c|c|c|c|}
			\hline
   Simulation Setting  & Framework & $\theta_{1}^{(g)}$  & $\tau^{(g)}$ \\
   \hline
         \multirow{3}{*}{1} & 1 & \textbf{0.300 (0.080)} & \textbf{17.404 (10.532)} \\ 
& 2 & 0.275 (0.080) & 14.335 (9.315) \\ 
& 3 & 0.309 (0.083) & 17.559 (10.959) \\ 
         \hline 
         \multirow{3}{*}{2} & 1 & \textbf{0.364 (0.128)} & \textbf{7.994 (5.916)} \\ 
& 2 & 0.309 (0.129) & 5.819 (4.992) \\ 
& 3 & 0.338 (0.134) & 6.944 (5.758) \\ 
         \hline 
         \multirow{3}{*}{3} & 1 & \textbf{0.303 (0.099)} & \textbf{10.038 (6.964)} \\ 
& 2 & 0.253 (0.098) & 6.842 (5.402) \\ 
& 3 & 0.318 (0.103) & 10.495 (7.244) \\
         \hline 
         \multirow{3}{*}{4} & 1 & \textbf{0.167 (0.069)} & \textbf{5.406 (4.908)} \\ 
& 2 & 0.131 (0.067) & 3.391 (3.601) \\ 
& 3 & 0.256 (0.076) & 12.781 (8.178) \\ 
         \hline 
         \multirow{3}{*}{5} & 1 & \textbf{0.222 (0.083)} & \textbf{7.568 (5.887)} \\ 
& 2 & 0.182 (0.082) & 5.079 (4.602) \\ 
& 3 & 0.262 (0.089) & 10.304 (7.313) \\ 
         \hline 
         \multirow{3}{*}{6} & 1 & \textbf{0.335 (0.110)} & \textbf{11.392 (8.952)} \\ 
& 2 & 0.262 (0.107) & 6.707 (5.937) \\ 
& 3 & 0.318 (0.101) & 11.524 (8.151) \\ 
         \hline 
\end{tabular} %}
\end{table*}

\clearpage
\newpage

\begin{table*}[!h]
	\caption{Comparison of both frameworks using every gene from each sample population. Framework 1 estimates a linear mixed effects cosinor model with the method proposed in Section \ref{sec:2.3} given Zeitgeber time (ZT), while Framework 2 estimates a linear mixed effects cosinor model given ZT. The regression parameter estimate $\hat{\gamma}$ is listed alongside the coefficient of determination ($R^2$), where the latter is in parentheses. Bold values indicate a value of $\hat{\gamma}$ that is closest to one (up to the first three decimal digits), which signifies that the quantities obtained from a framework are closer to the quantities obtained from a linear mixed effects cosinor model estimated with internal circadian time.} \label{tab:app}
  \centering
 %\resizebox{1.05\textwidth}{!}{
		\begin{tabular}{|l|c|c|c|c|c|c|c|c|c|}
			\hline
   \multirow{1}{*}{Sample Population} & \multirow{1}{*}{Framework}  &  \multicolumn{1}{c|}{$\theta^{(g)}_1$} & \multicolumn{1}{c|}{$\tau^{(g)}$} \\
   \hline
         \multirow{2}{*}{Archer (Control)} & 1 & \textbf{1.001 (0.991)} & \textbf{1.030 (0.979)} \\ 
        & 2 & 1.019 (0.990) & 1.071 (0.978) \\ 
         \hline 
         \multirow{2}{*}{Archer (Intervention)} & 1 & 0.989 (0.938) & \textbf{1.022 (0.915)} \\ 
    & 2 & \textbf{0.994 (0.936)} & 1.042 (0.911) \\
         \hline 
         \multirow{2}{*}{Archer (Combined)}& 1 & \textbf{1.058 (0.972)} & \textbf{1.157 (0.955)} \\
         & 2 & 1.072 (0.970) & 1.195 (0.953) \\
         \hline 
         \multirow{2}{*}{Braun} & 1 & \textbf{1.005 (0.996)} & \textbf{1.036 (0.980)} \\
         & 2 & 1.021 (0.995) & 1.068 (0.977) \\
         \hline 
         \multirow{2}{*}{M\"{o}ller-Levet (Control)} & 1 & \textbf{1.106 (0.990)} & \textbf{1.176 (0.965)} \\ 
         & 2 & 1.136 (0.987) & 1.242 (0.956) \\
         \hline 
         \multirow{2}{*}{M\"{o}ller-Levet (Intervention)} & 1 & \textbf{1.043 (0.992)} & \textbf{1.100 (0.976)} \\ 
         & 2 & 1.063 (0.991) & 1.135 (0.974) \\
         \hline 
         \multirow{2}{*}{M\"{o}ller-Levet (Combined)} & 1 & \textbf{1.094 (0.994)} & \textbf{1.162 (0.980)} \\ 
         & 2 & 1.110 (0.994) & 1.191 (0.978) \\
         \hline 
\end{tabular} %}
\end{table*}

\clearpage
\newpage

\begin{figure*}[!h]
\centering
\includegraphics[scale=0.13]{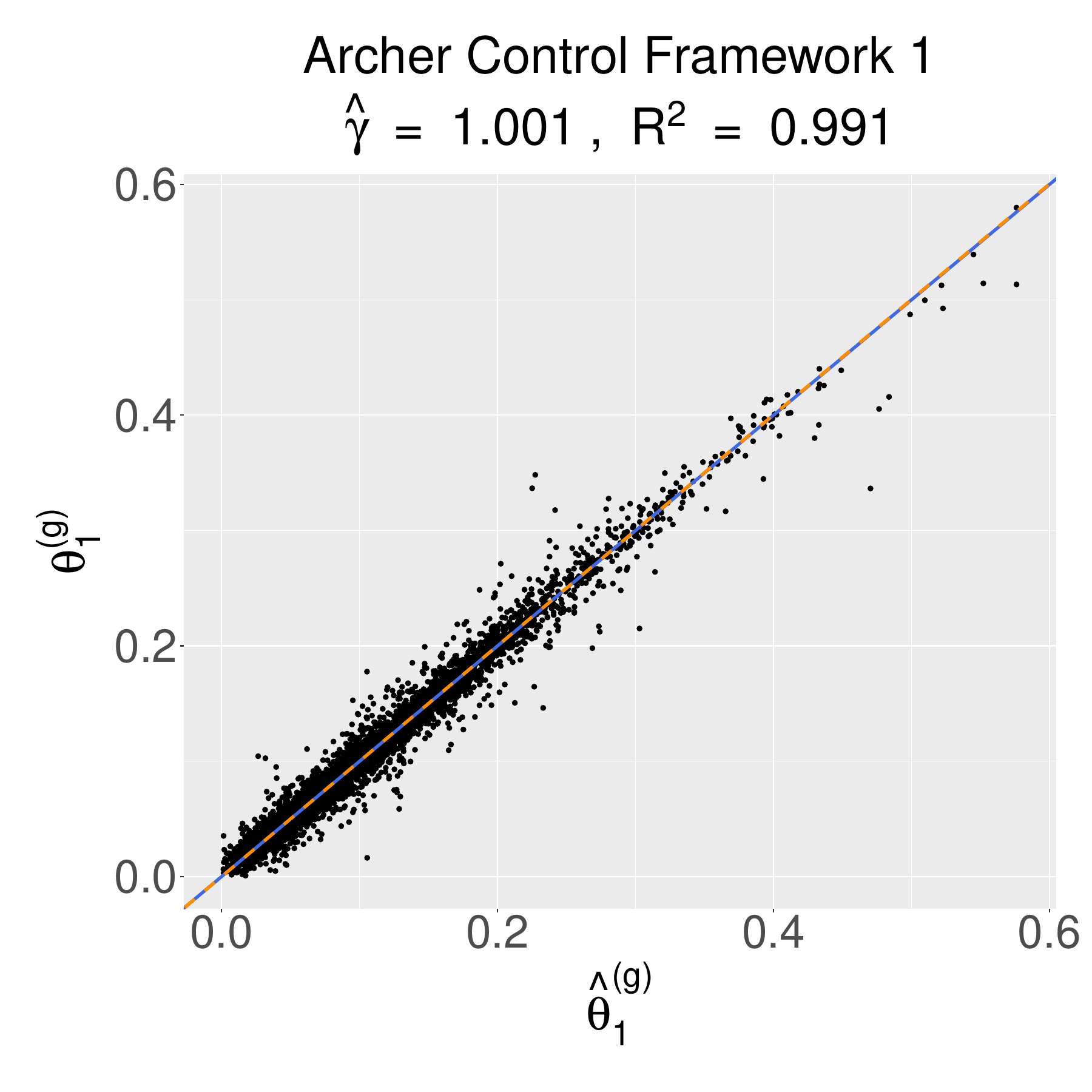}
\includegraphics[scale=0.13]{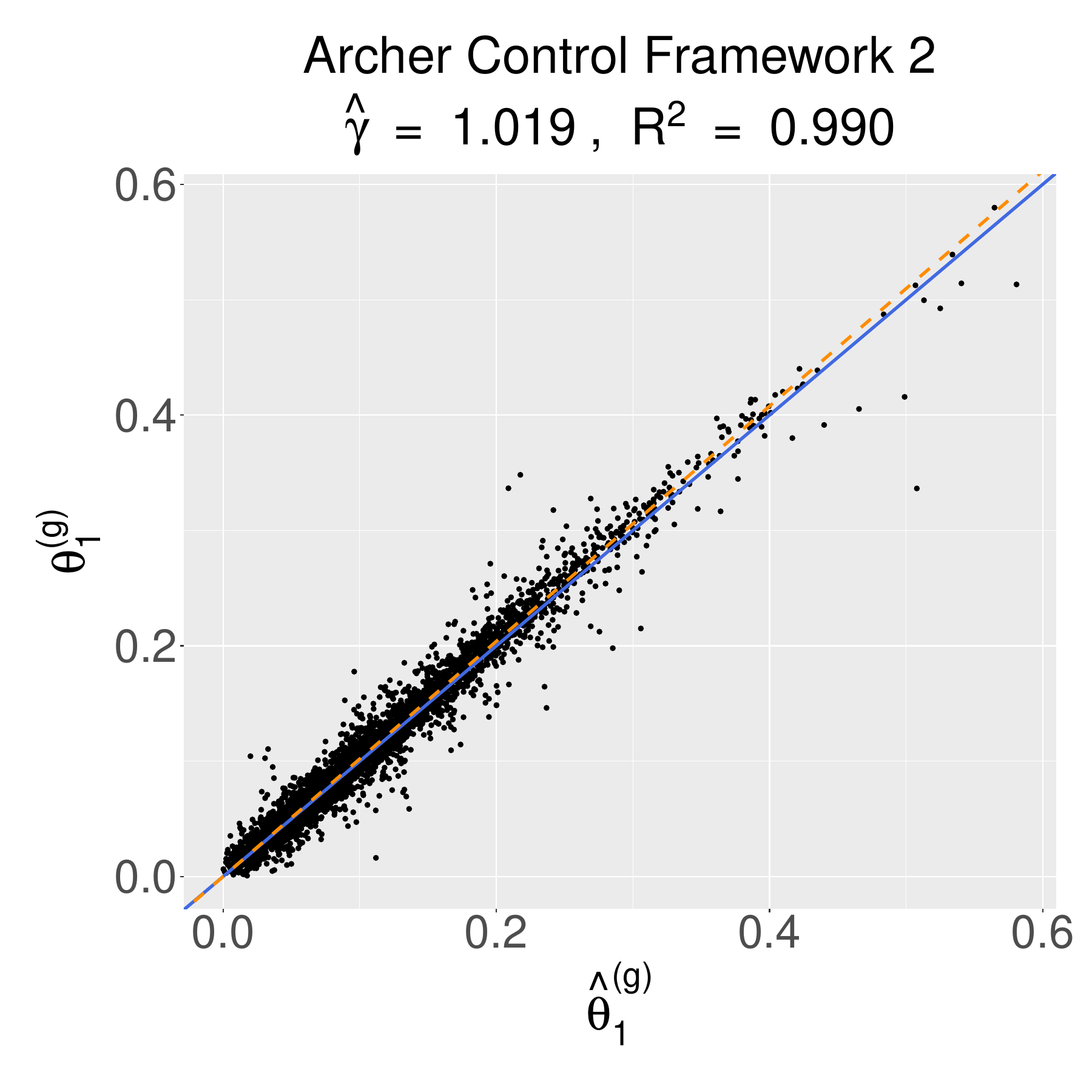} 
\includegraphics[scale=0.13]{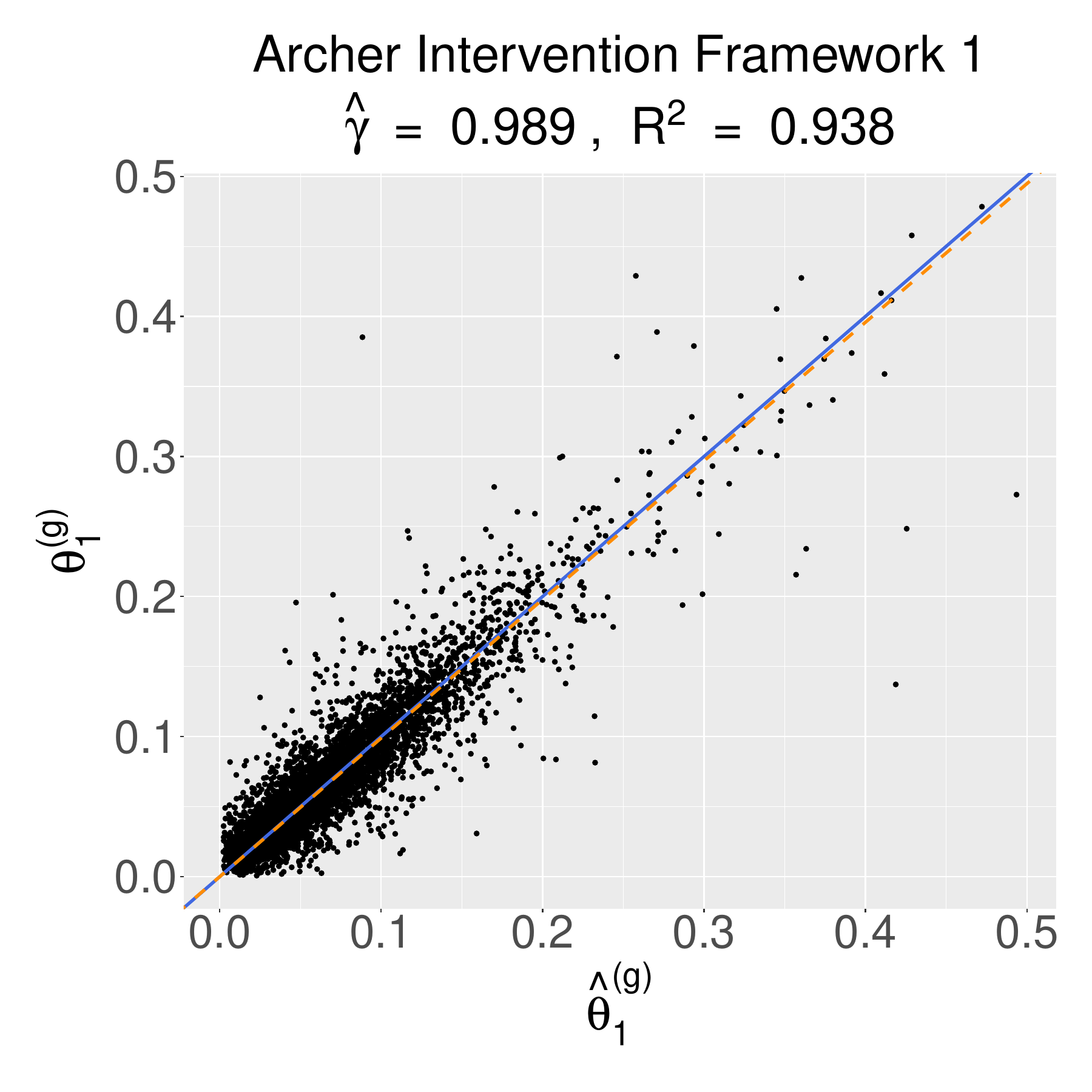} 
\includegraphics[scale=0.13]{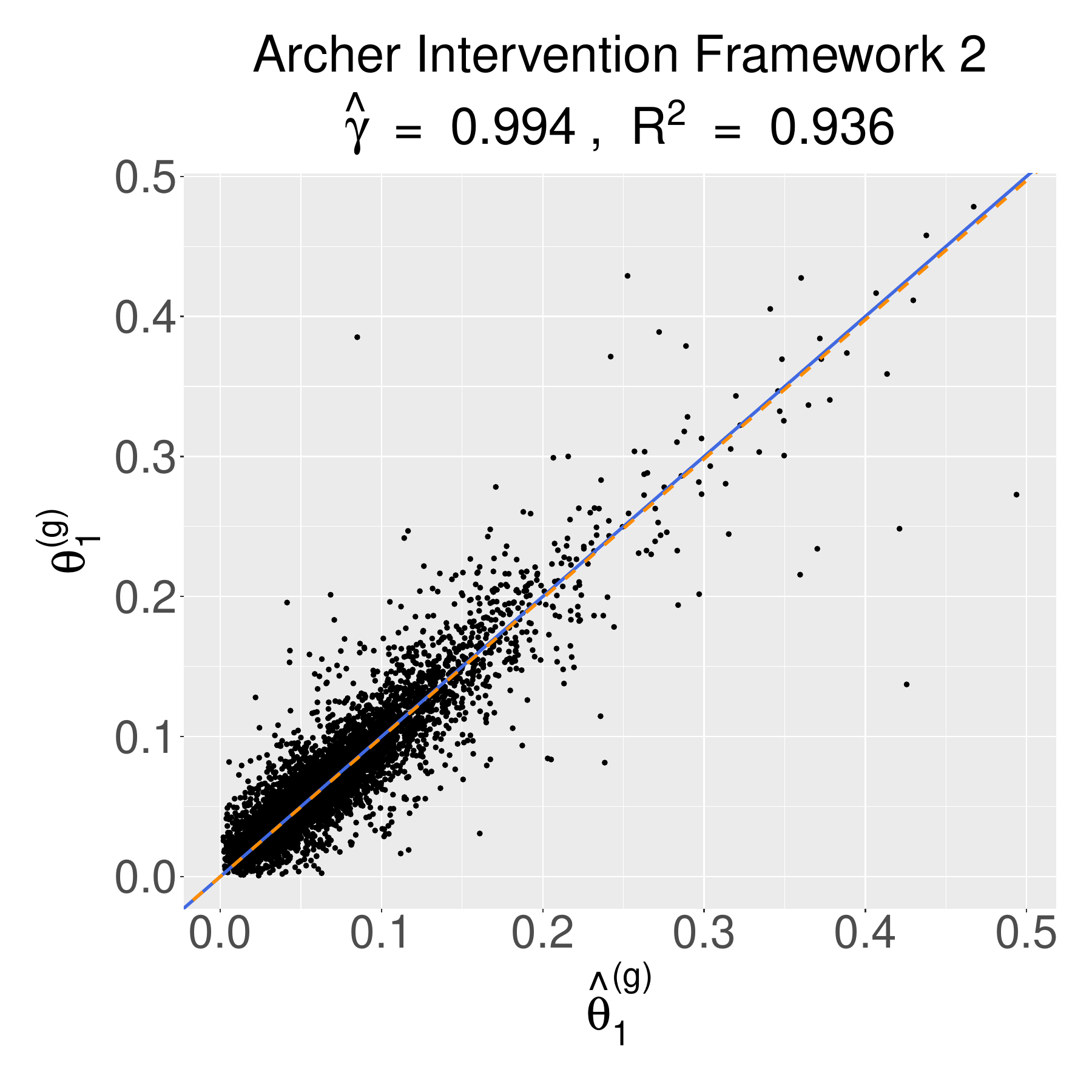}\\
\includegraphics[scale=0.13]{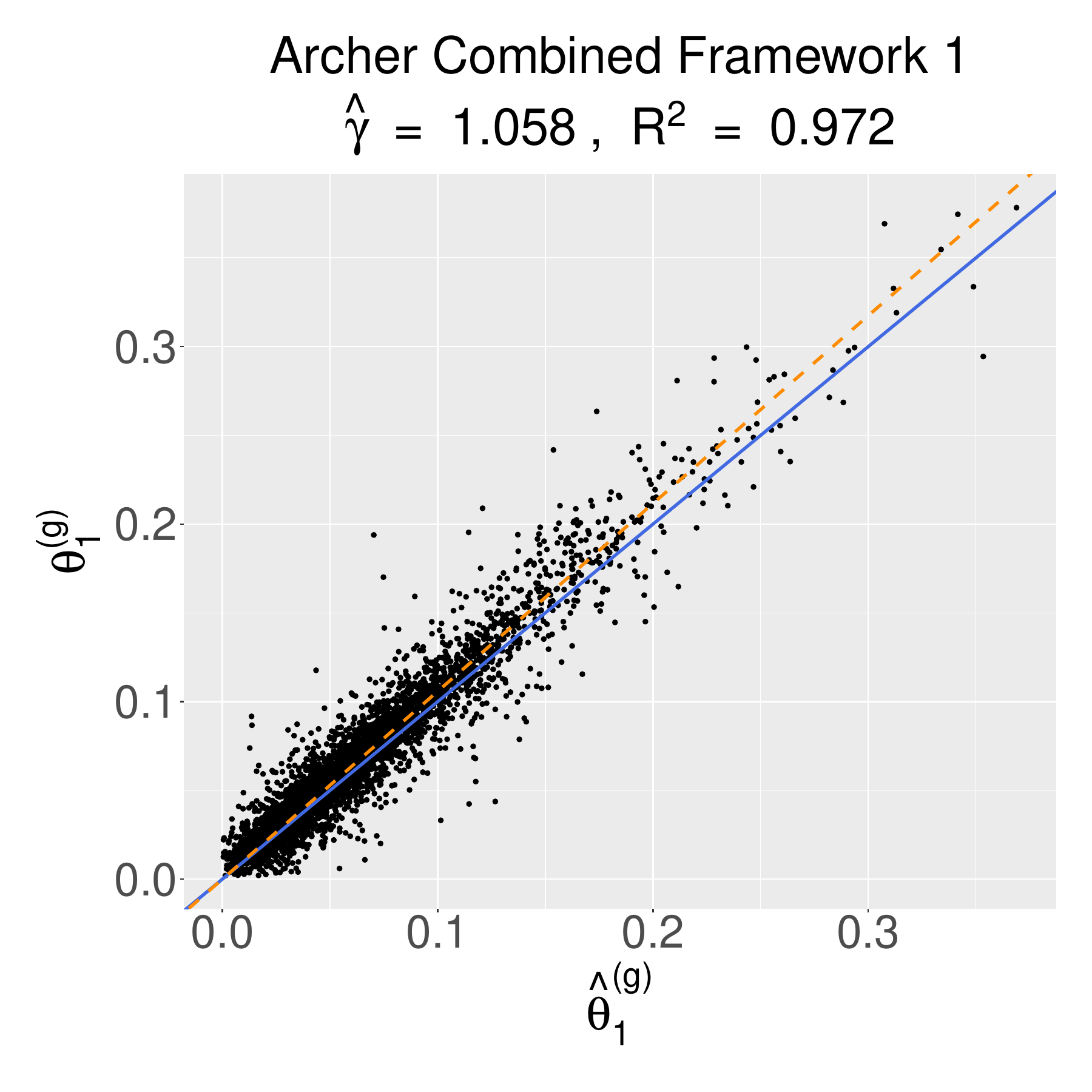}
\includegraphics[scale=0.13]{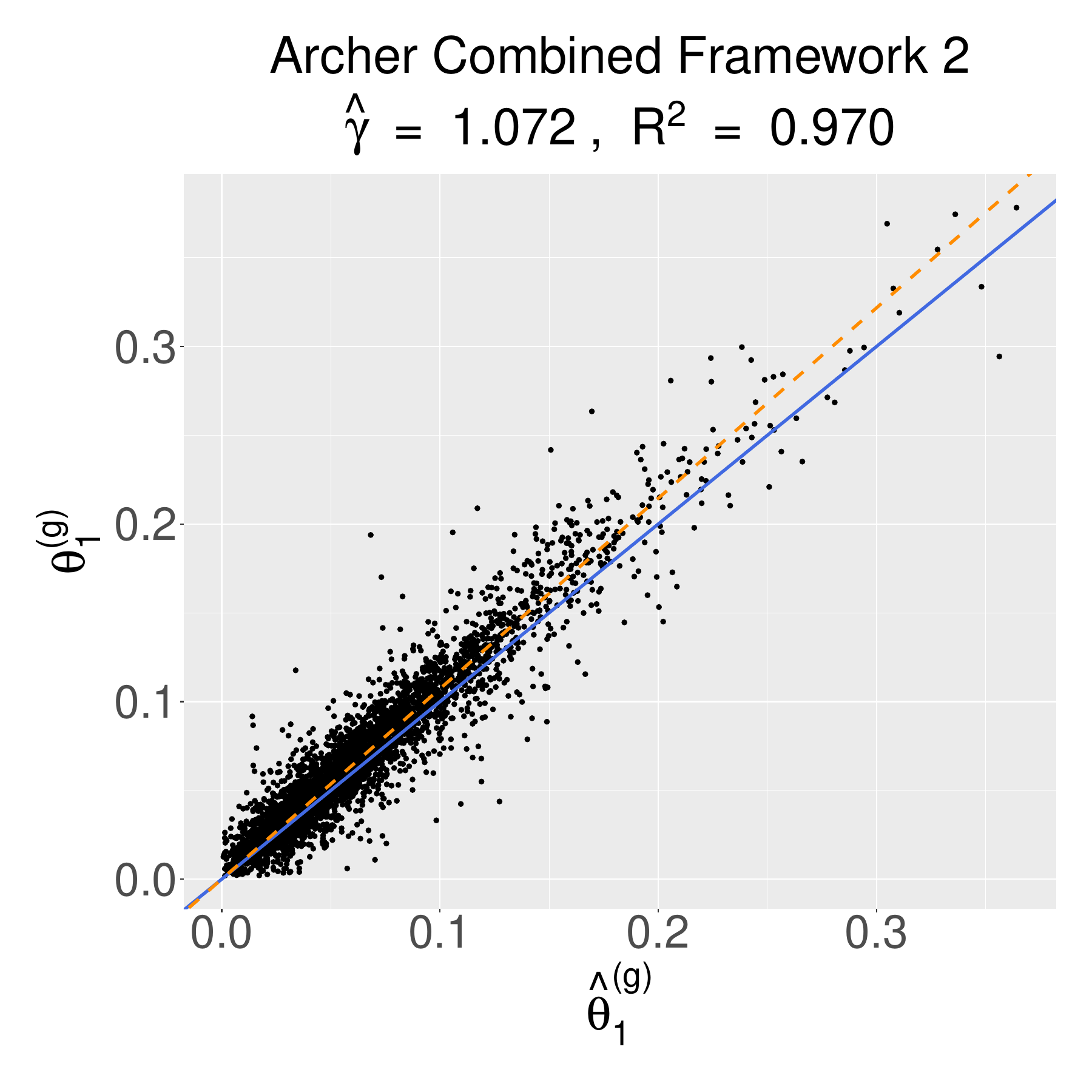} 
\includegraphics[scale=0.13]{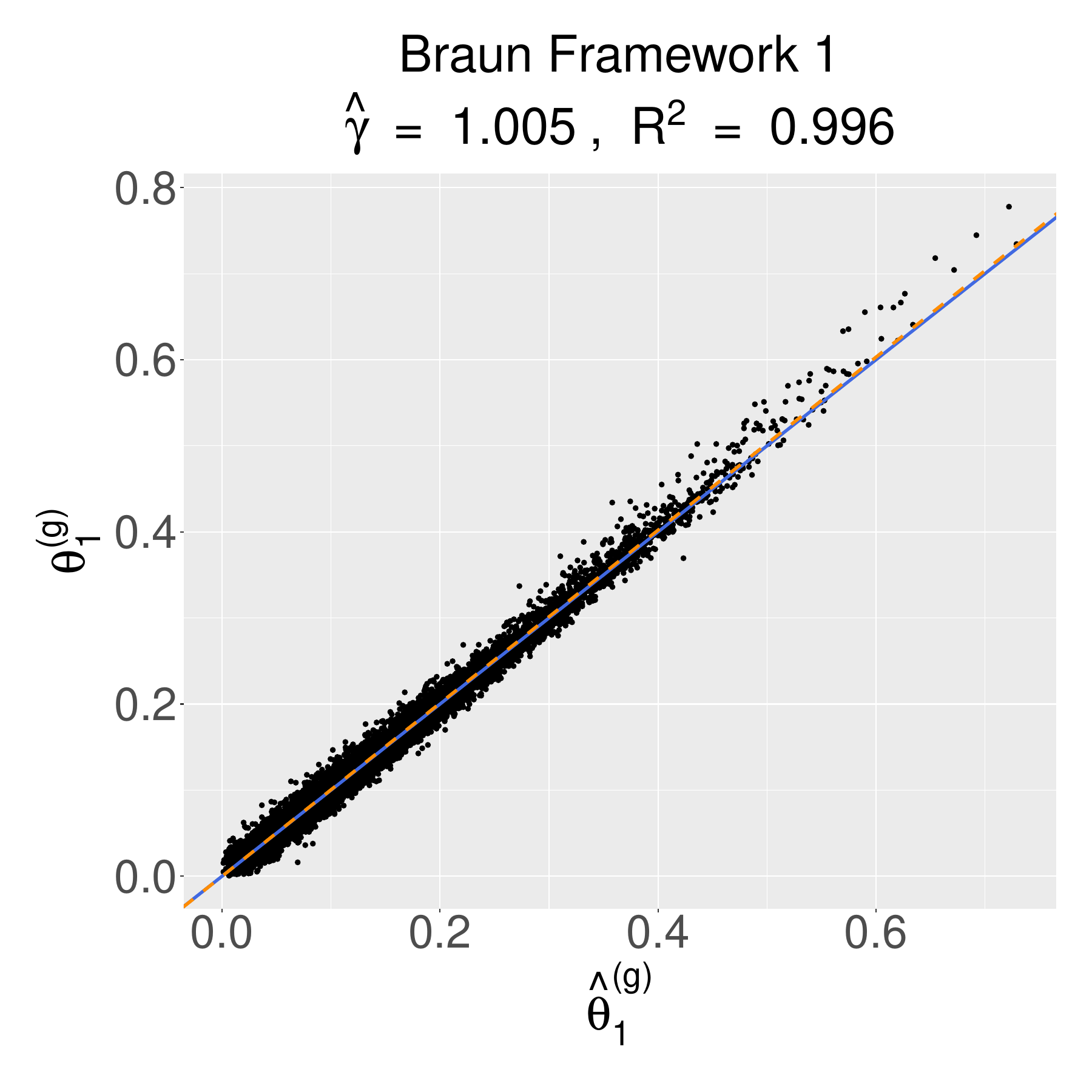}
\includegraphics[scale=0.13]{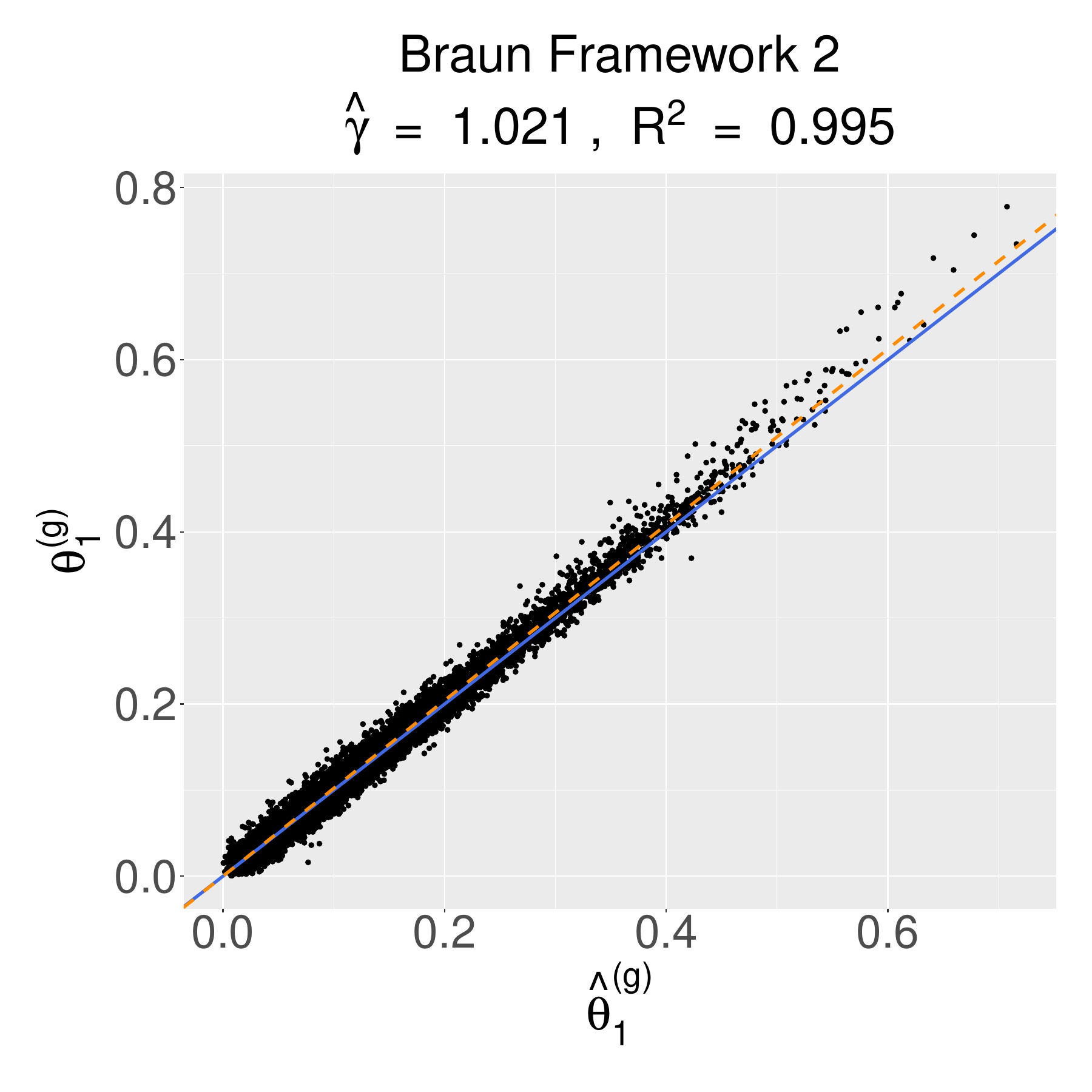} \\
\includegraphics[scale=0.13]{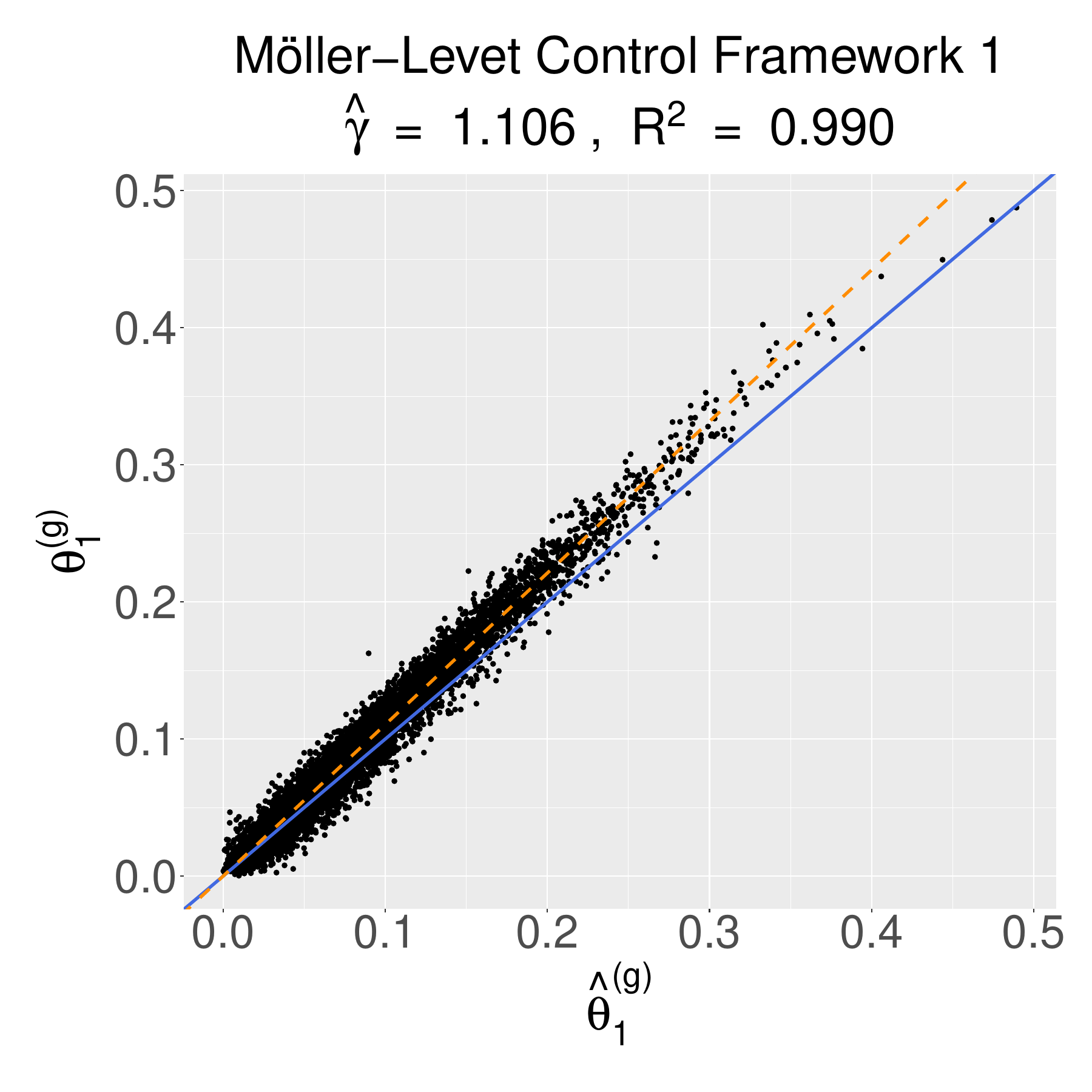} 
\includegraphics[scale=0.13]{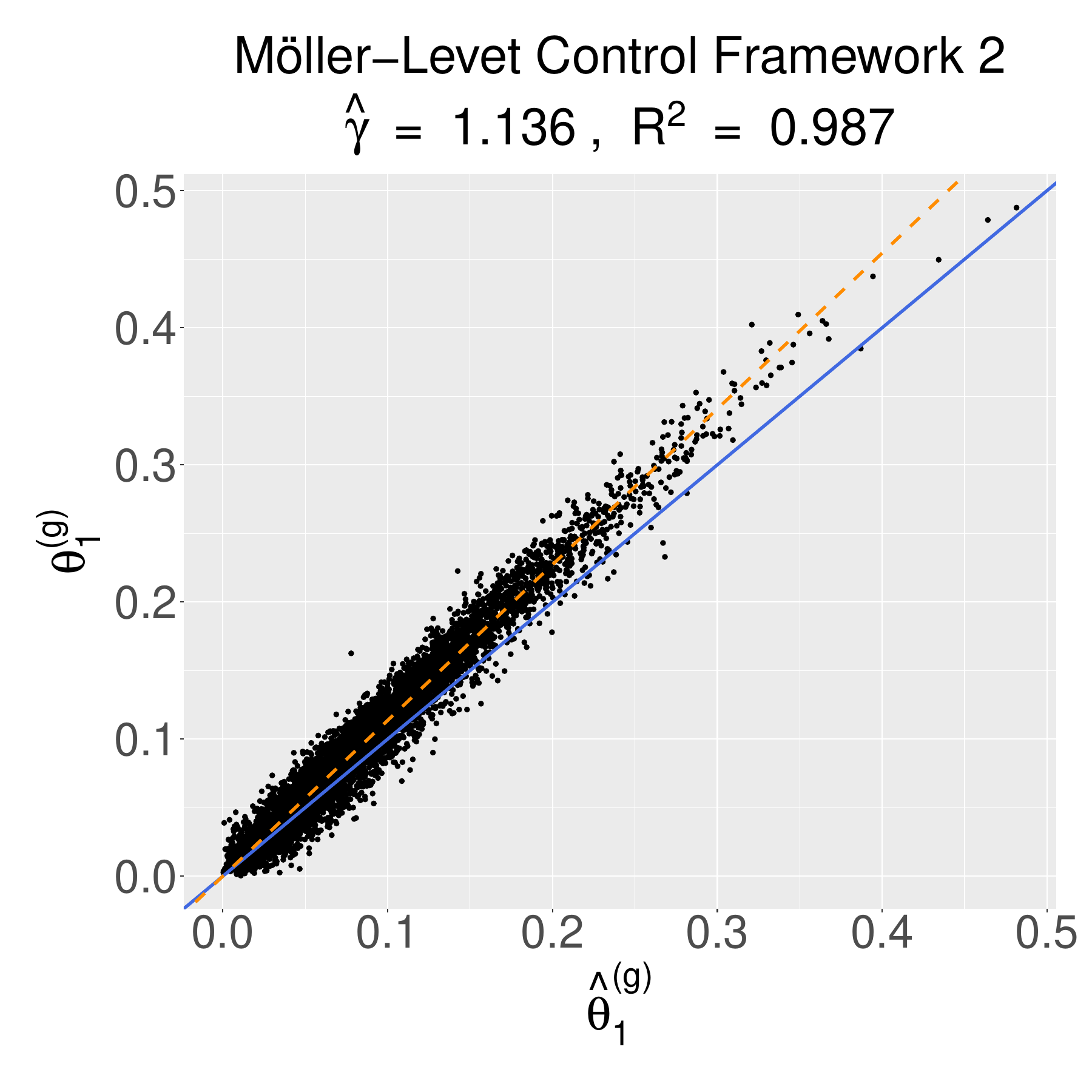}
\includegraphics[scale=0.13]{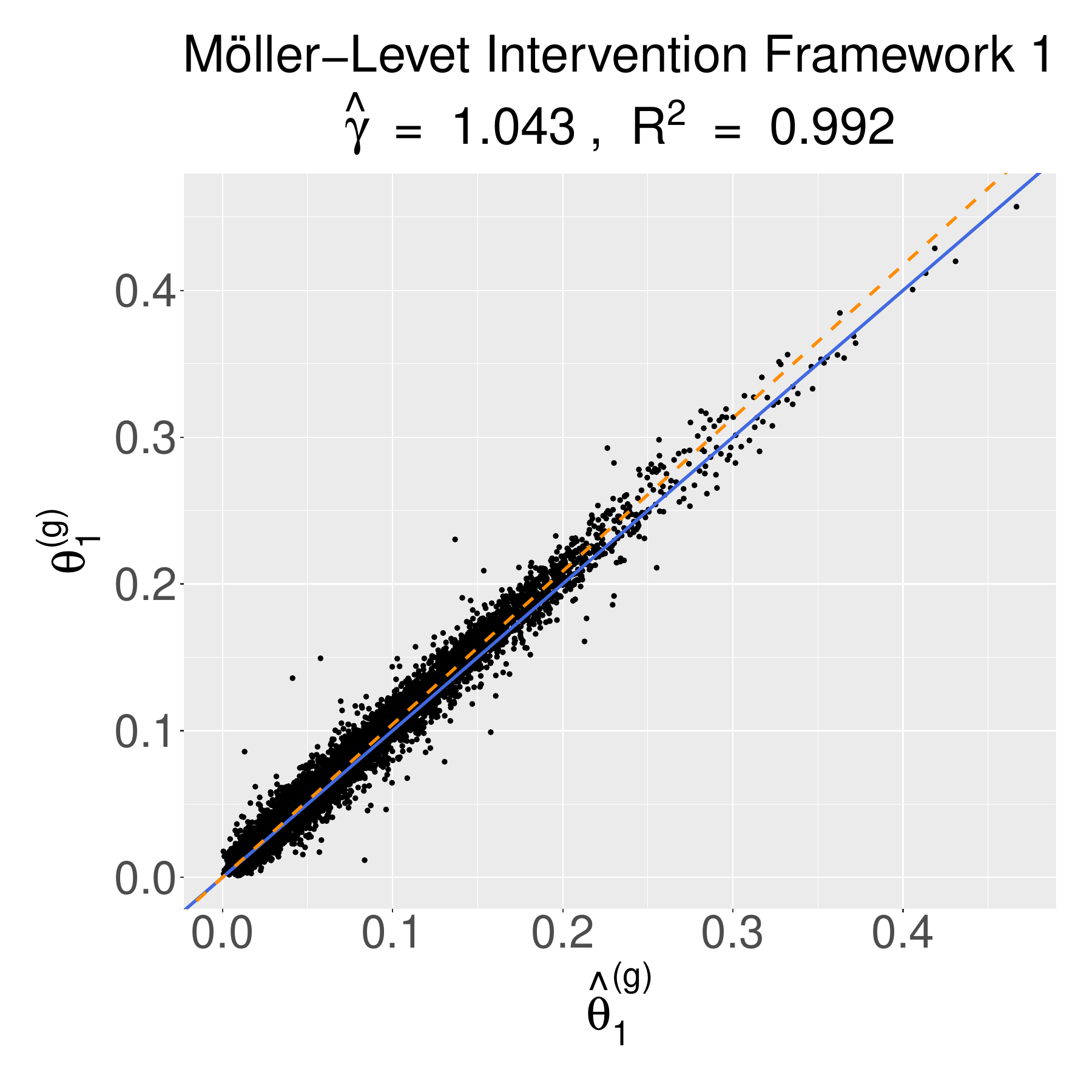}
\includegraphics[scale=0.13]{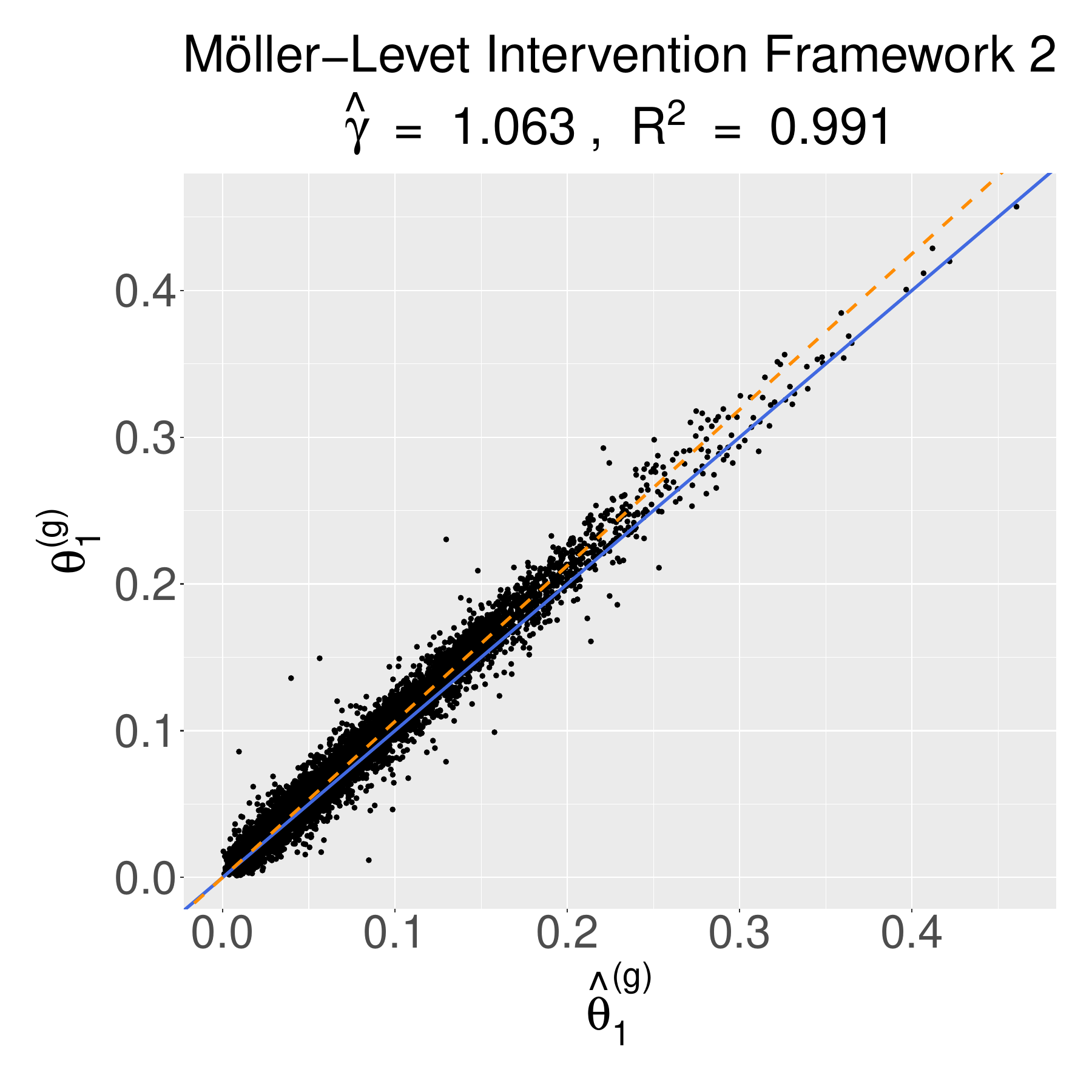} \\
\includegraphics[scale=0.13]{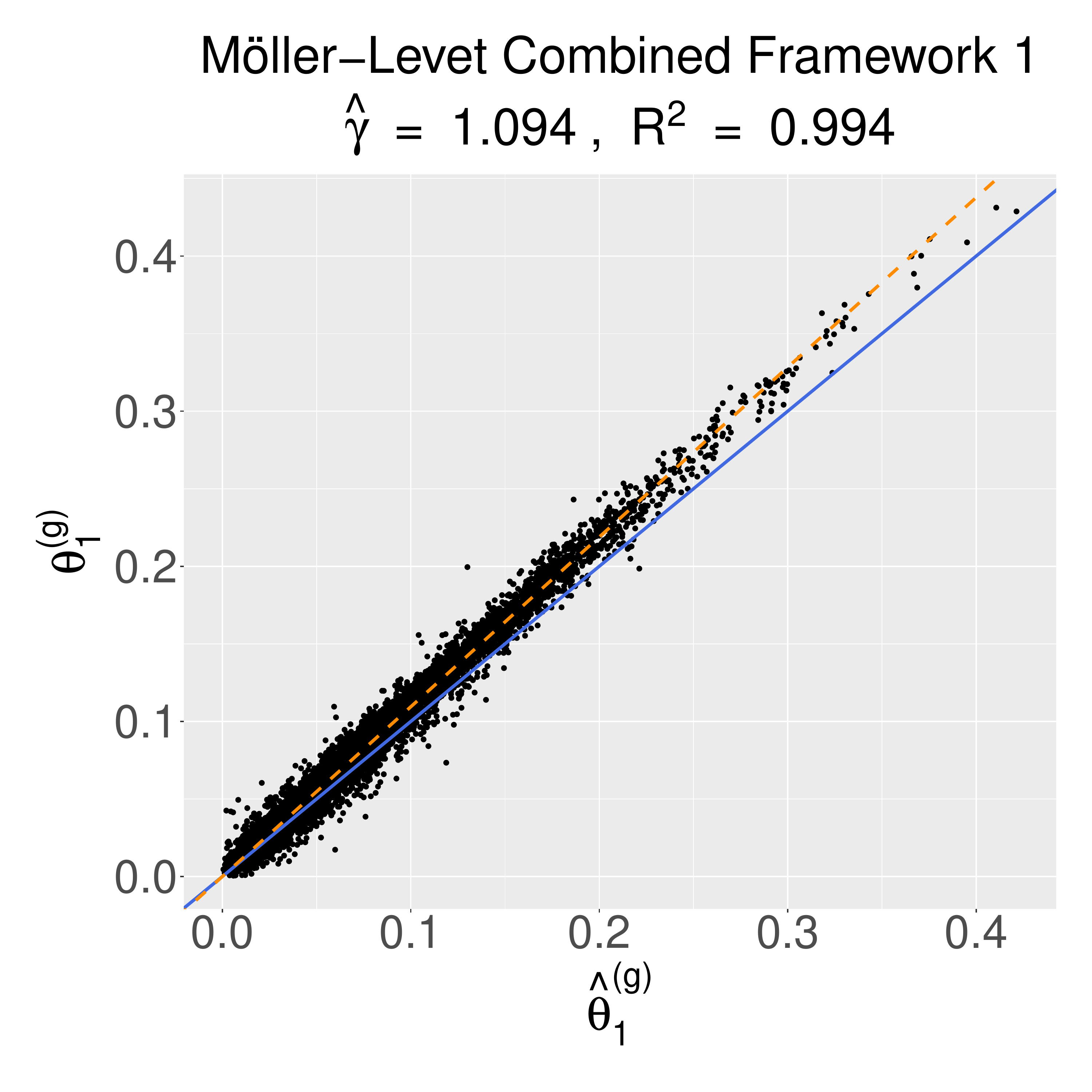}
\includegraphics[scale=0.13]{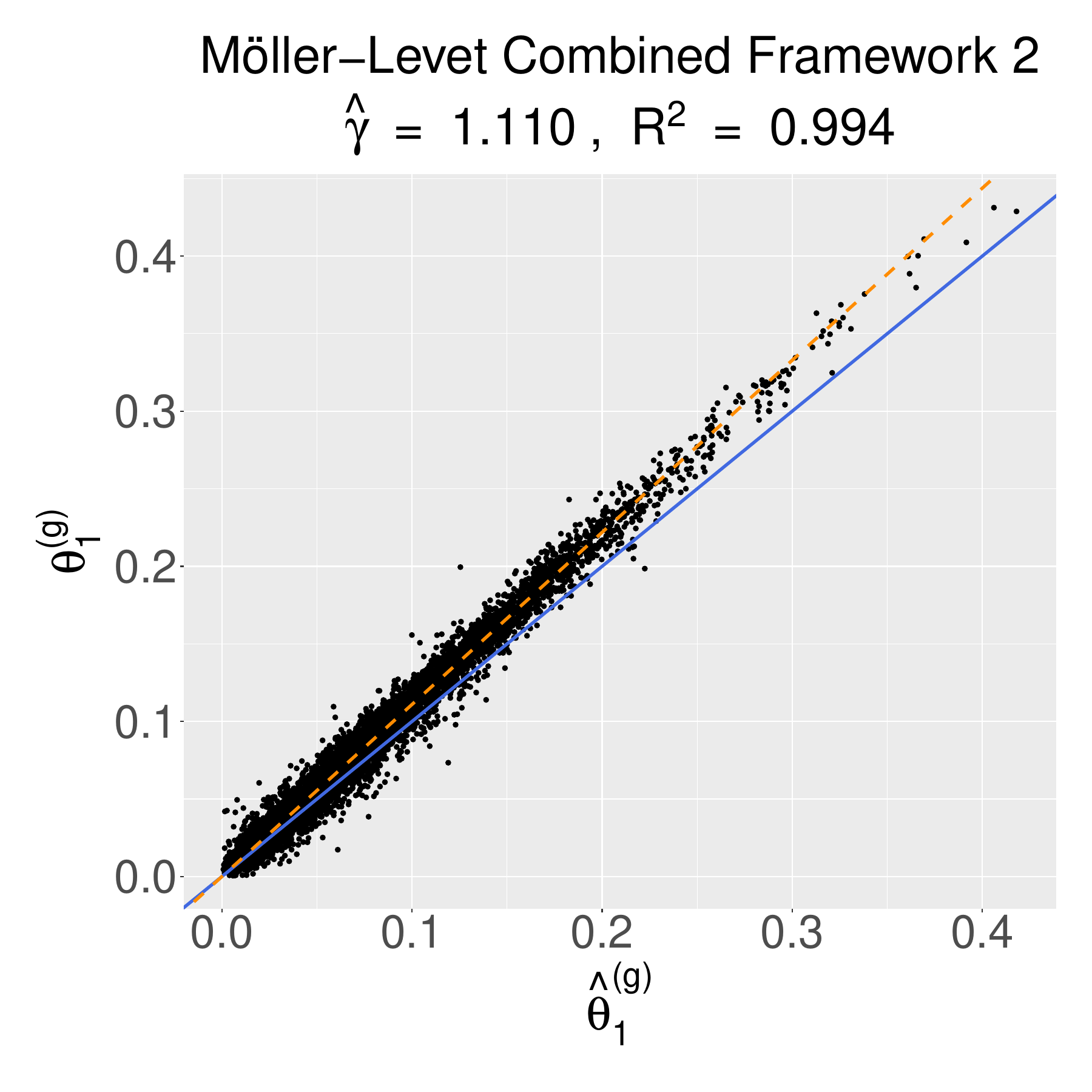}
\caption{Scatter plots of linear model fits using every gene for each Framework (Framework 1 the proposed method from Section \ref{sec:2.3}, Framework 2 linear mixed effects cosinor regression). The y-axis denotes an amplitude estimate obtained from internal circadian time data ($\theta_1^{(g)}$), and the x-axis an amplitude estimate obtained from Zeitgeber time data ($\hat{\theta}_1^{(g)}$). The dashed orange line displays the linear model's fit, and the blue line denotes the fit of a linear model when ${\hat{\gamma}}=1$. Each data point in the scatter plot represents an amplitude estimate obtained for a specific gene.}
    \label{fig:amp_w}
\end{figure*}

\clearpage
\newpage

\begin{figure*}[!h]
\centering
\includegraphics[scale=0.13]{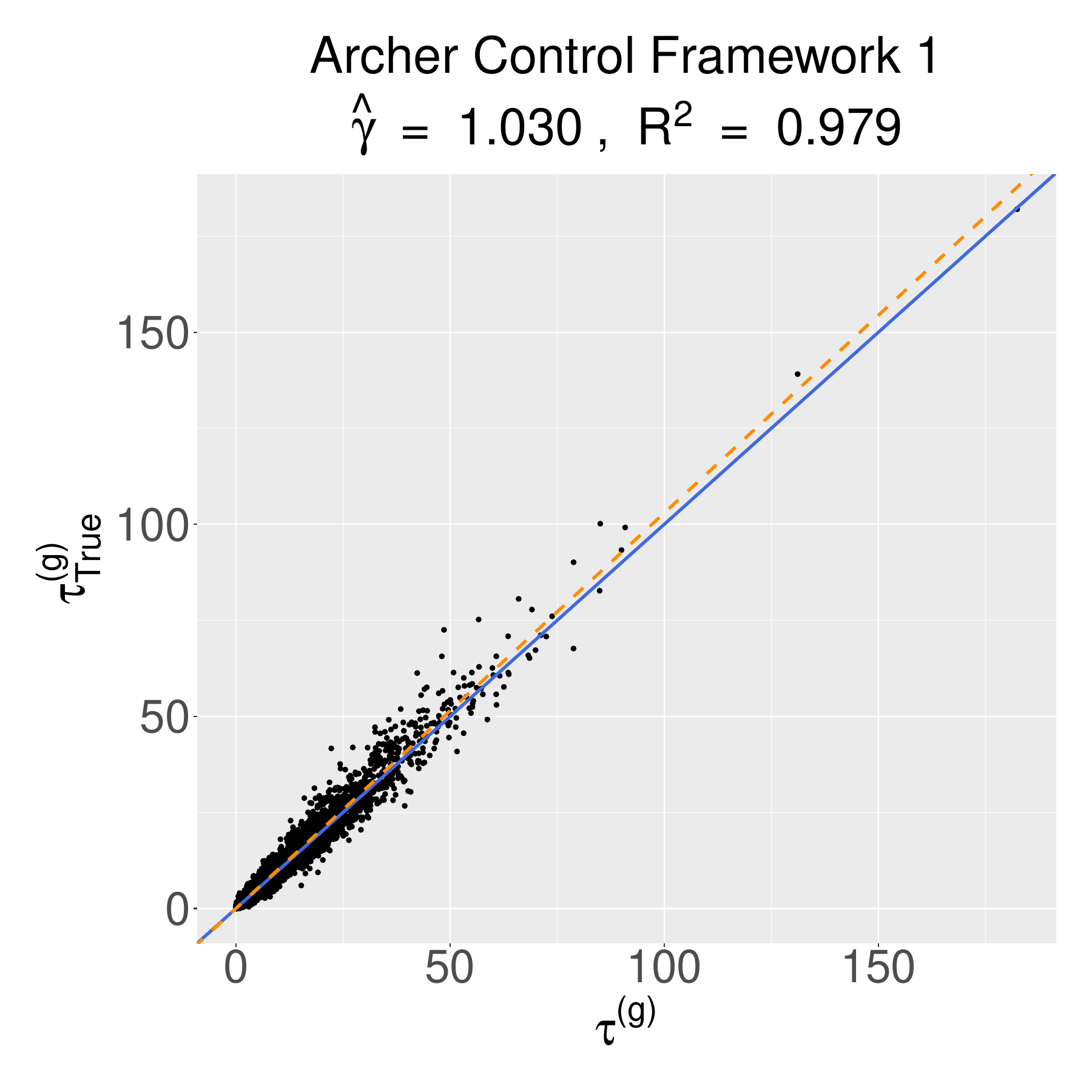}
\includegraphics[scale=0.13]{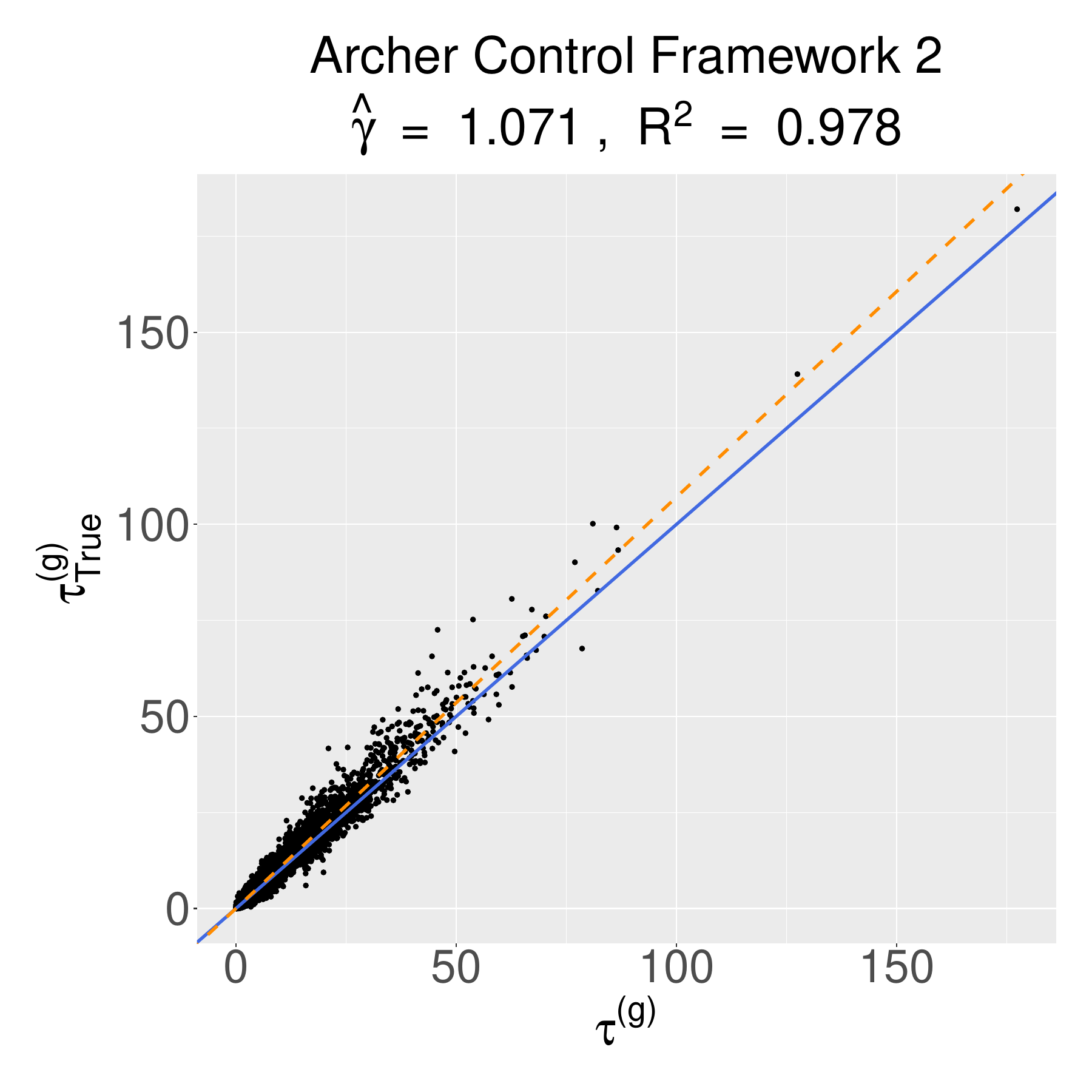} 
\includegraphics[scale=0.13]{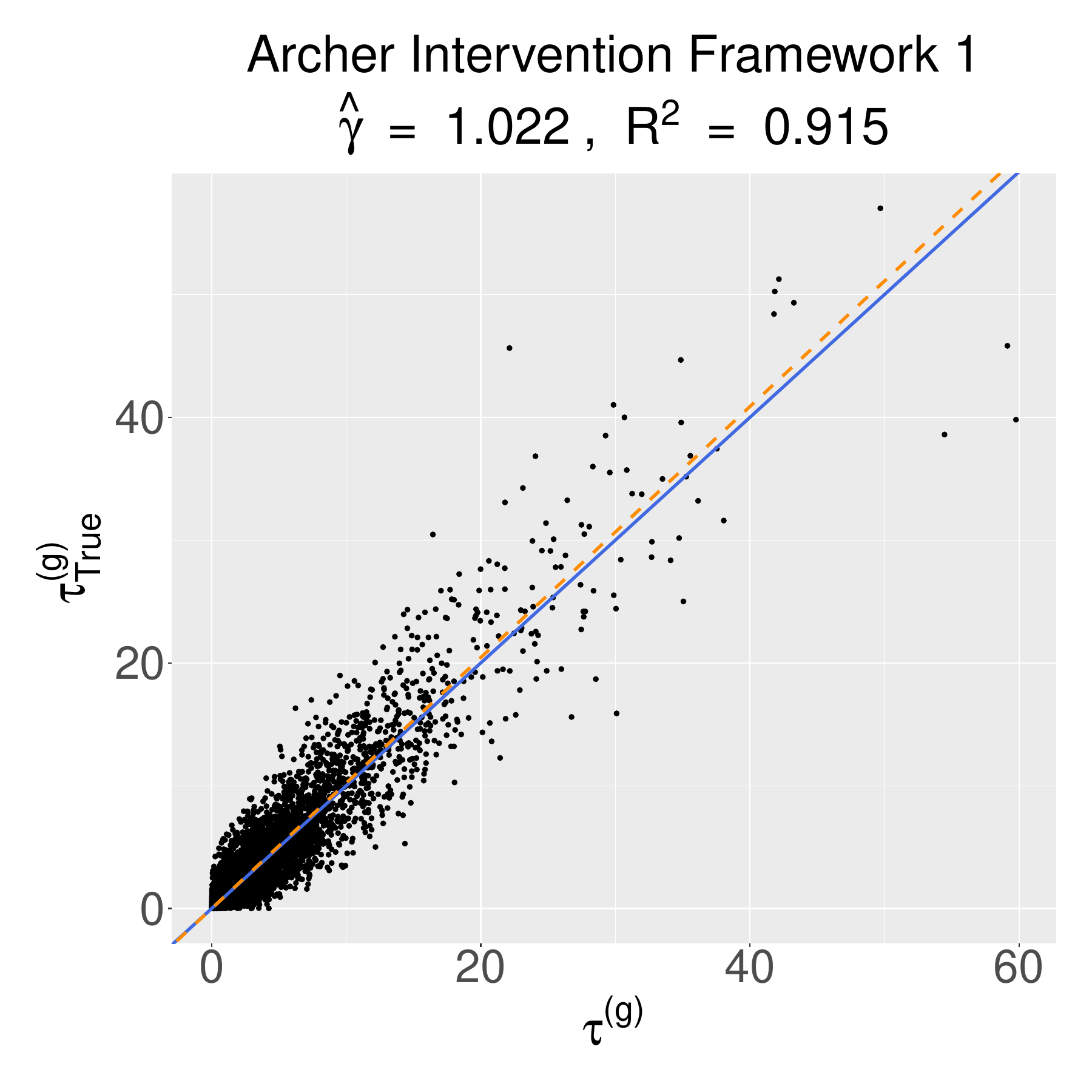} 
\includegraphics[scale=0.13]{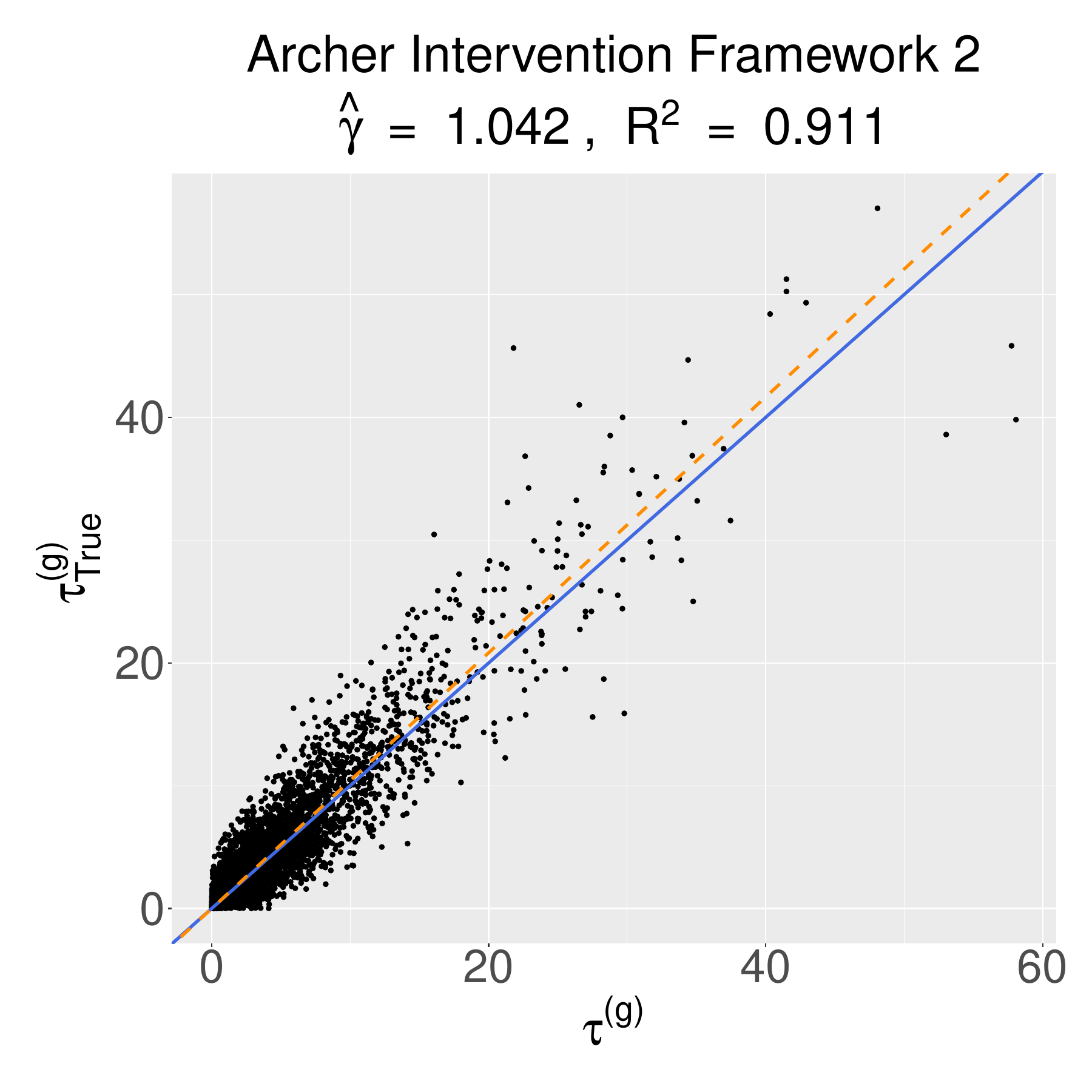}\\
\includegraphics[scale=0.13]{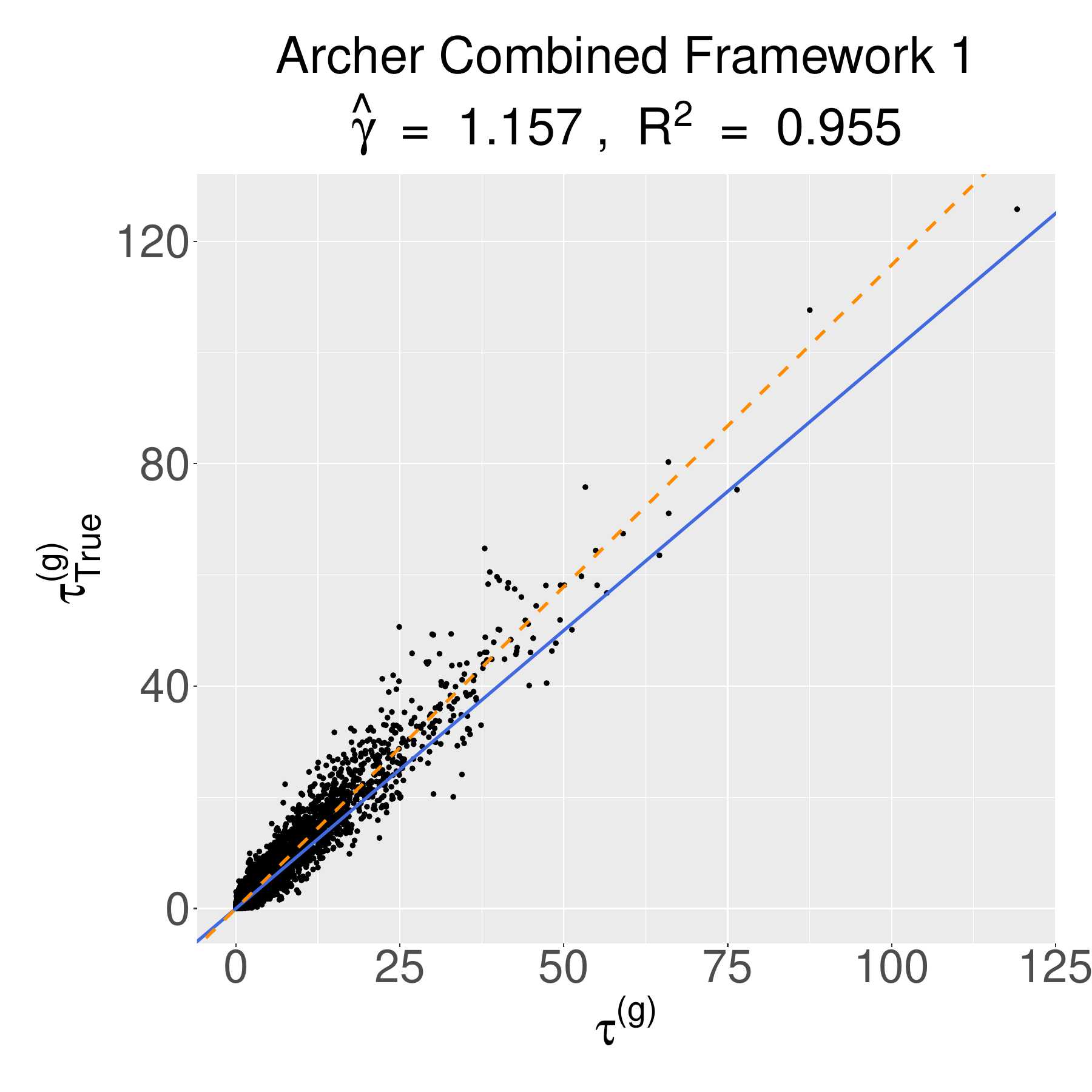}
\includegraphics[scale=0.13]{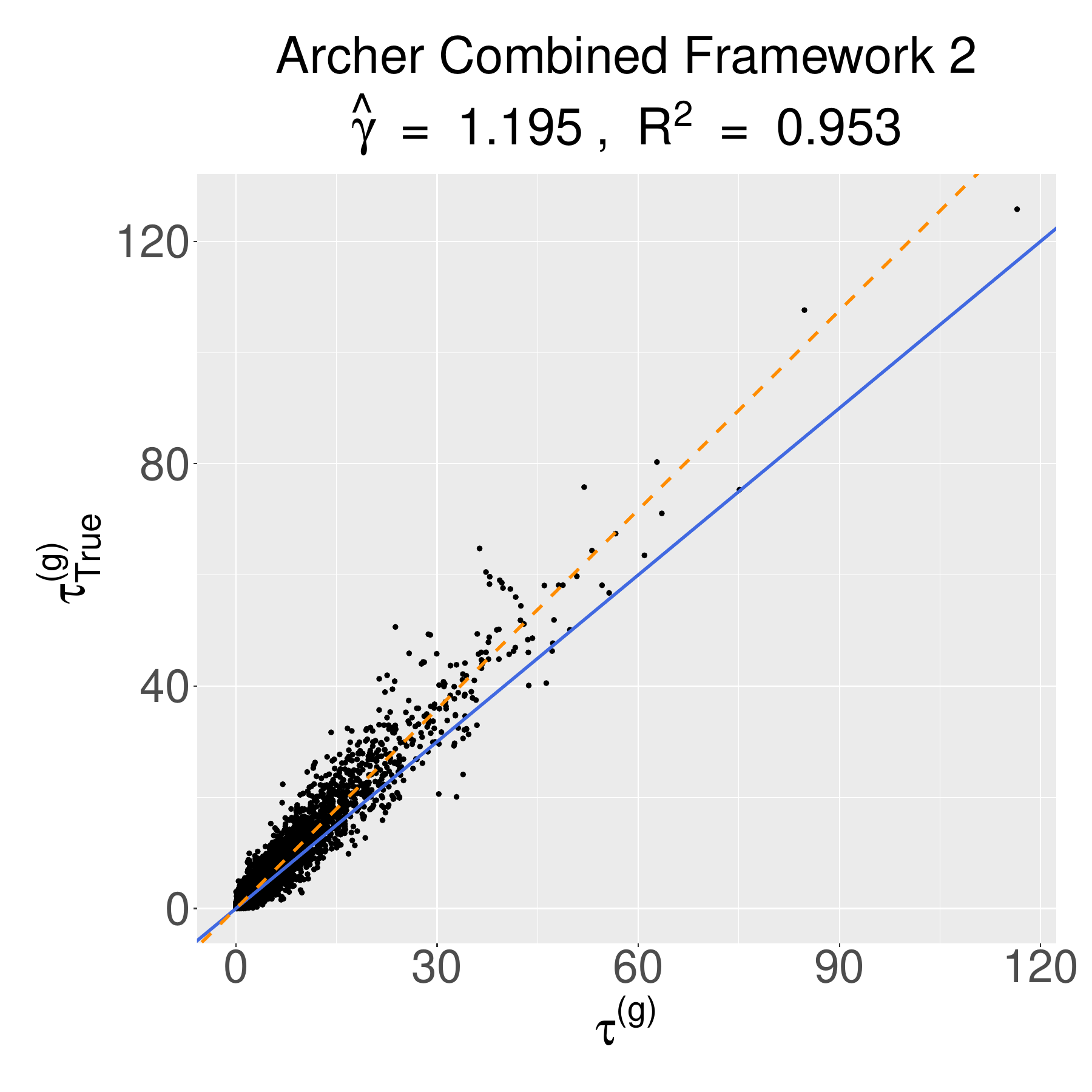} 
\includegraphics[scale=0.13]{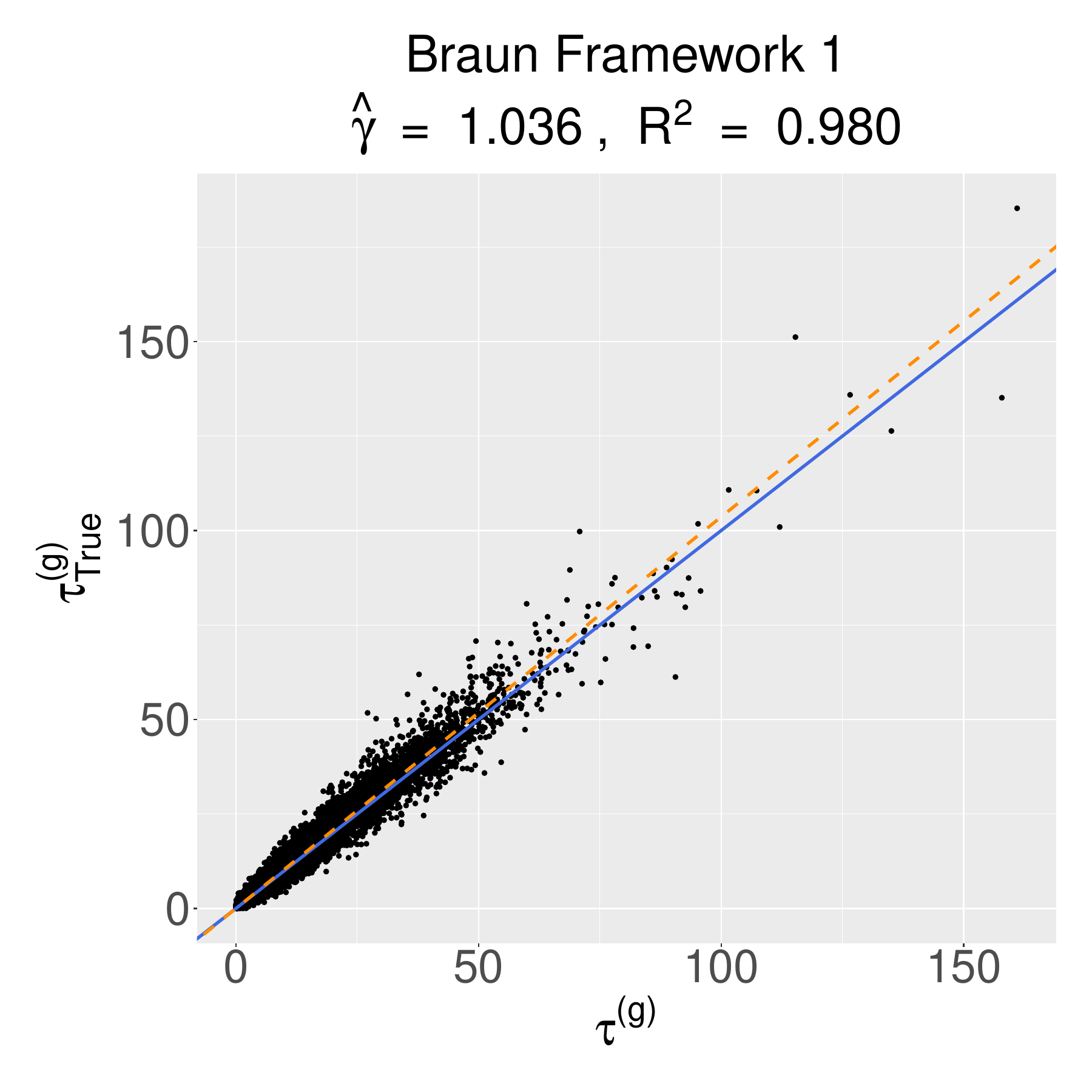}
\includegraphics[scale=0.13]{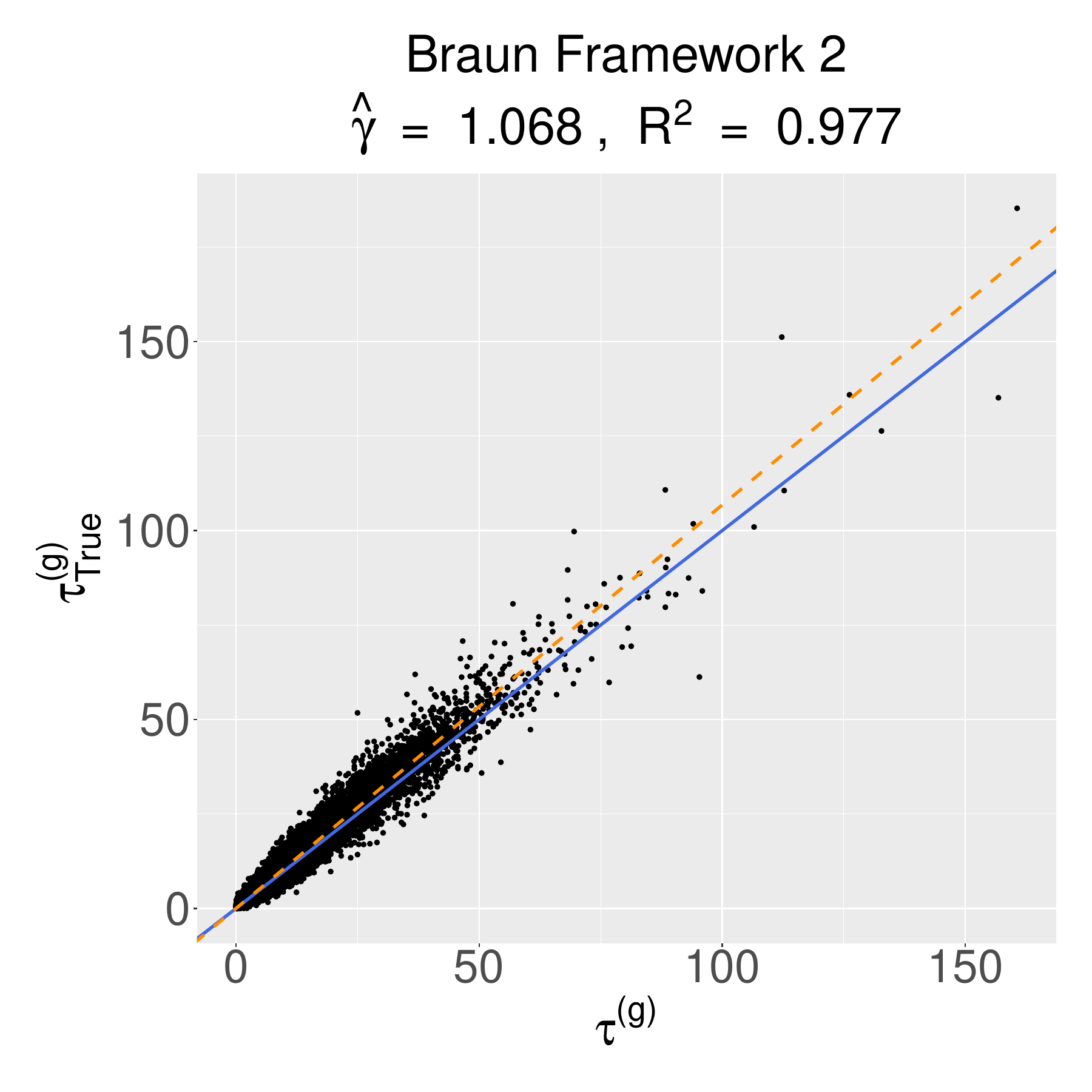} \\
\includegraphics[scale=0.13]{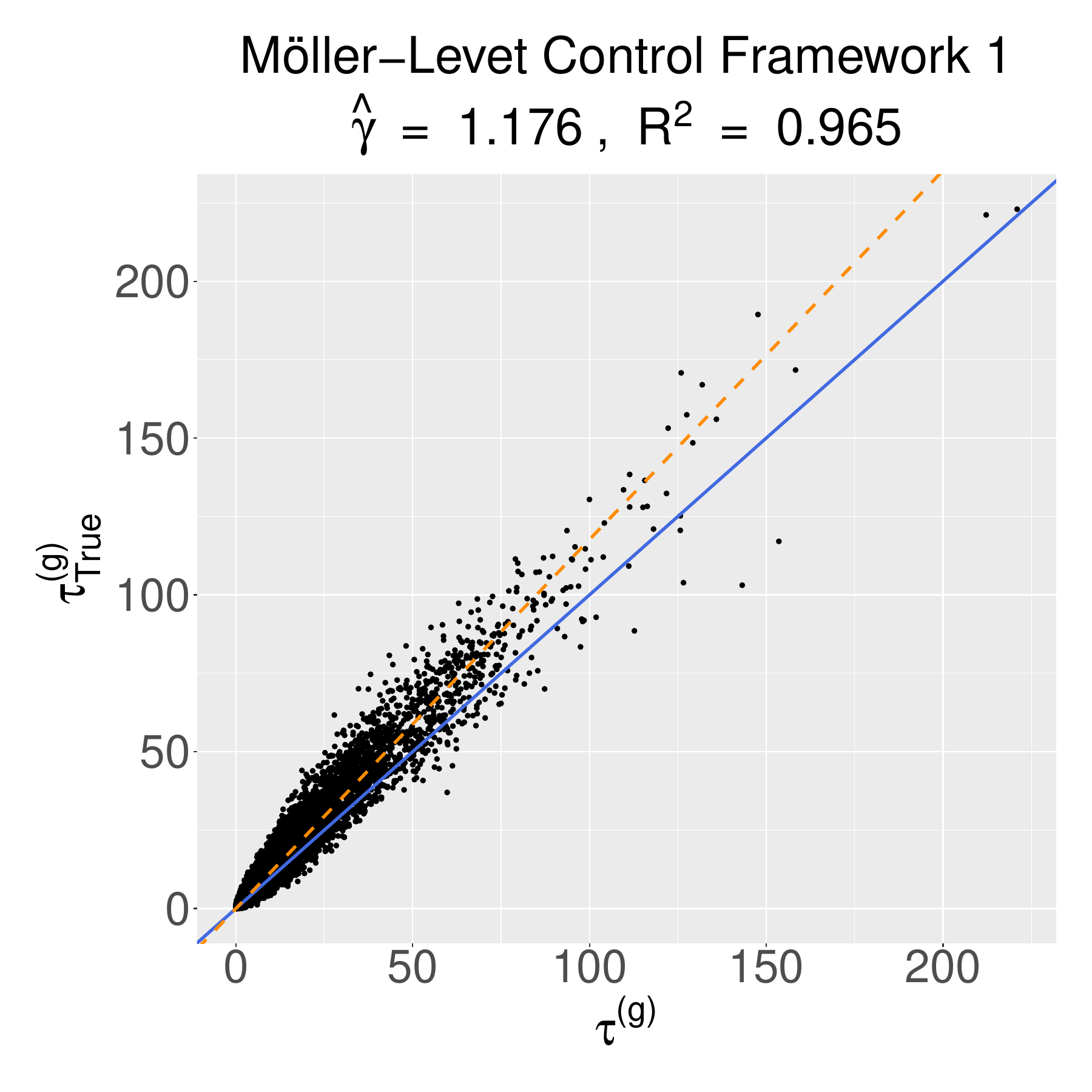} 
\includegraphics[scale=0.13]{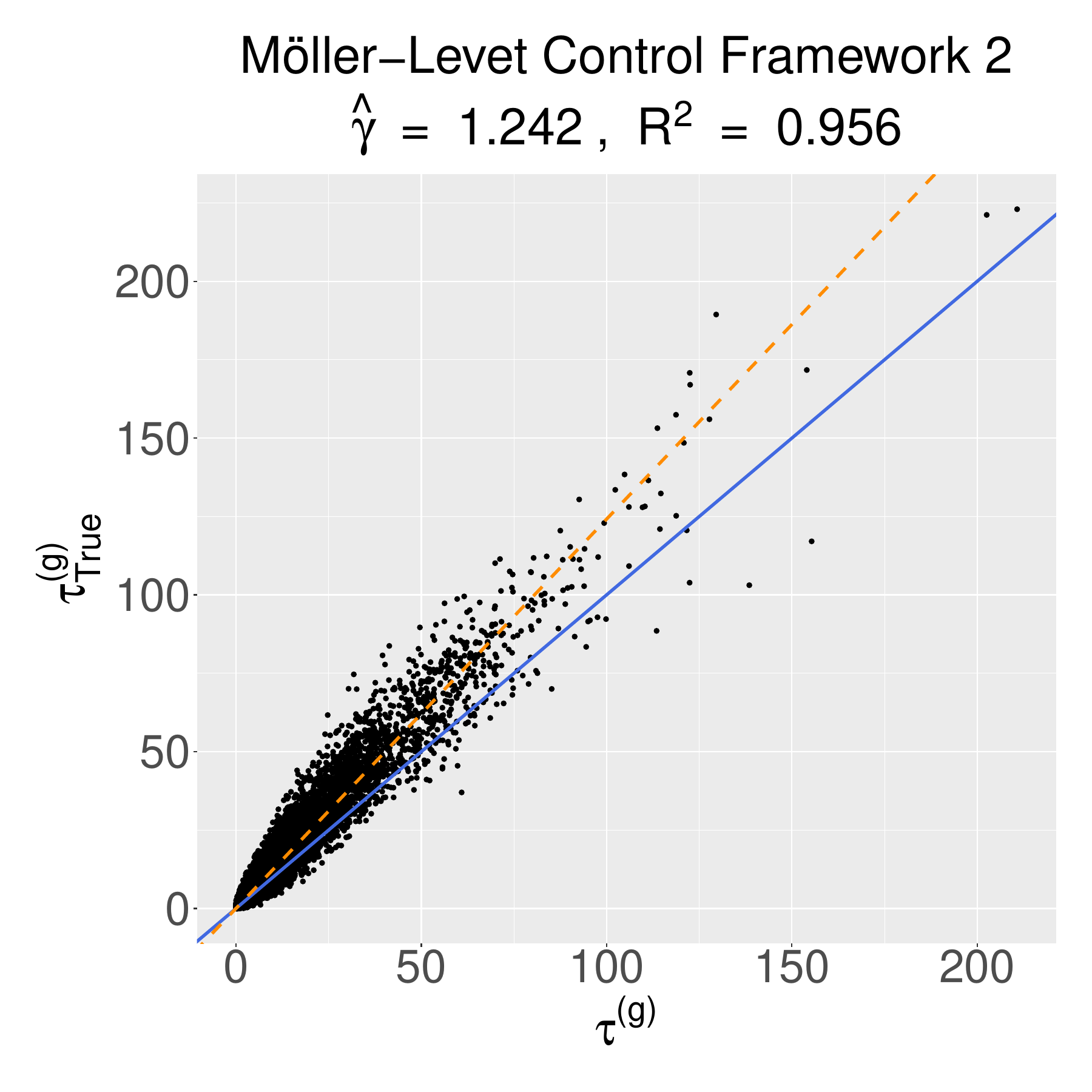}
\includegraphics[scale=0.13]{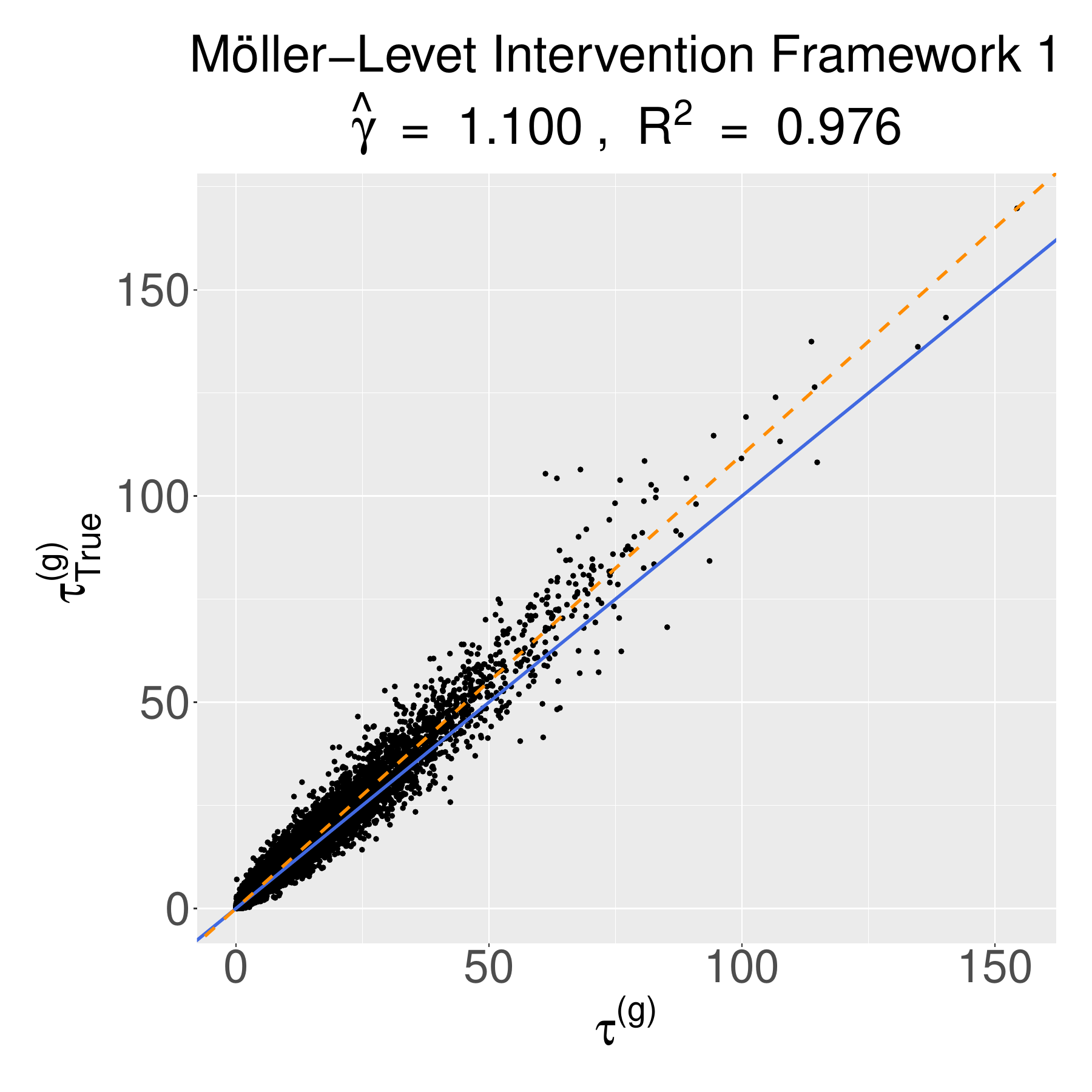}
\includegraphics[scale=0.13]{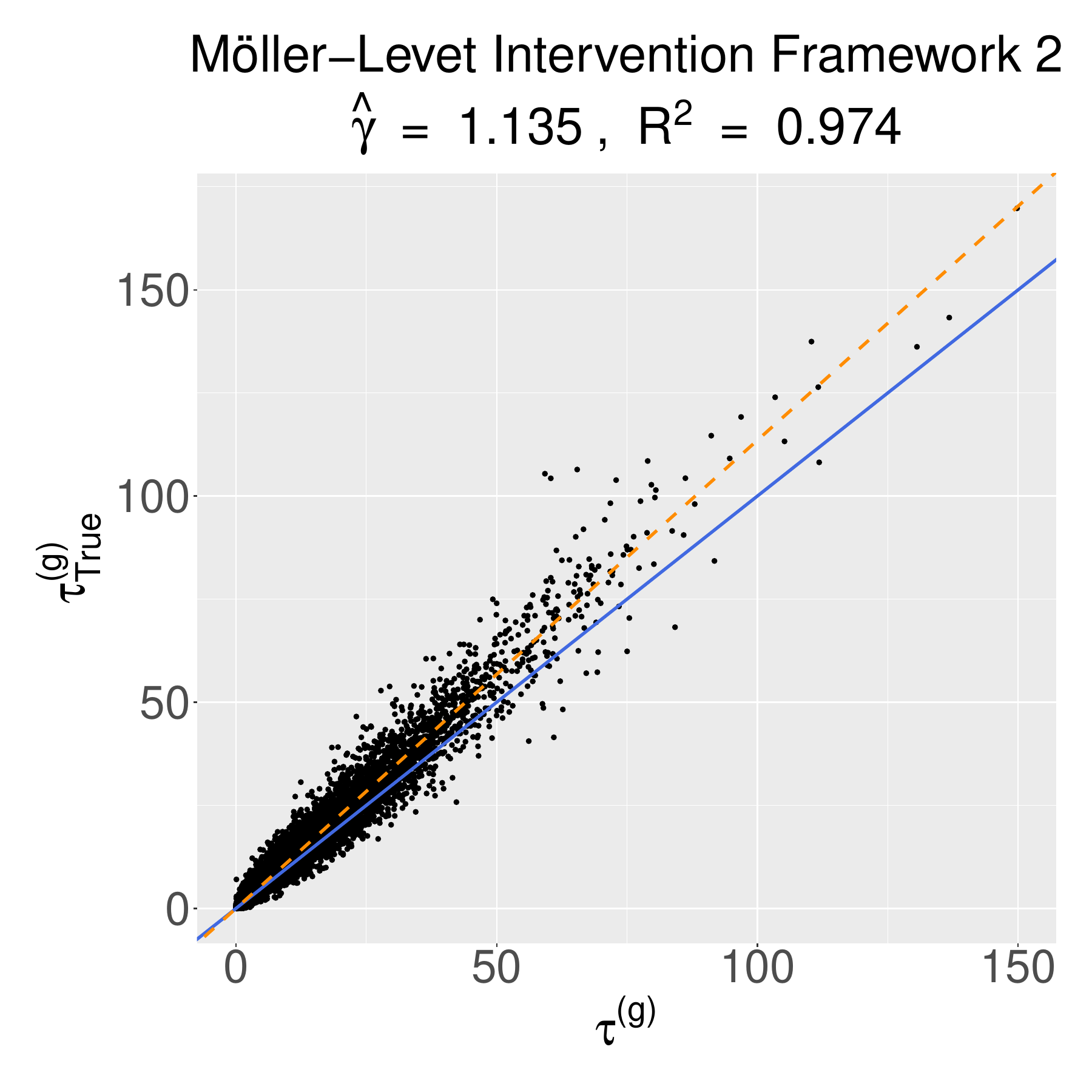} \\
\includegraphics[scale=0.13]{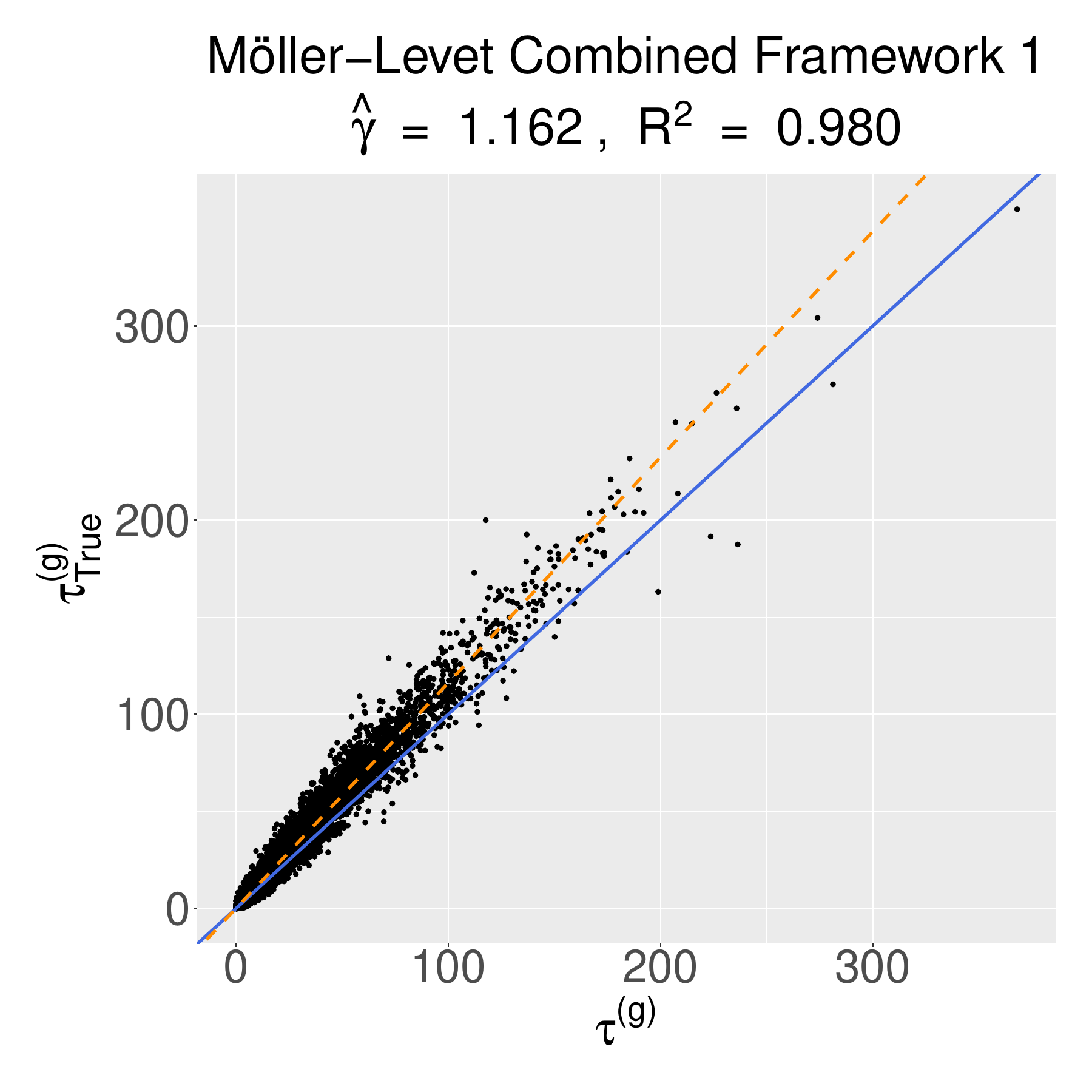}
\includegraphics[scale=0.13]{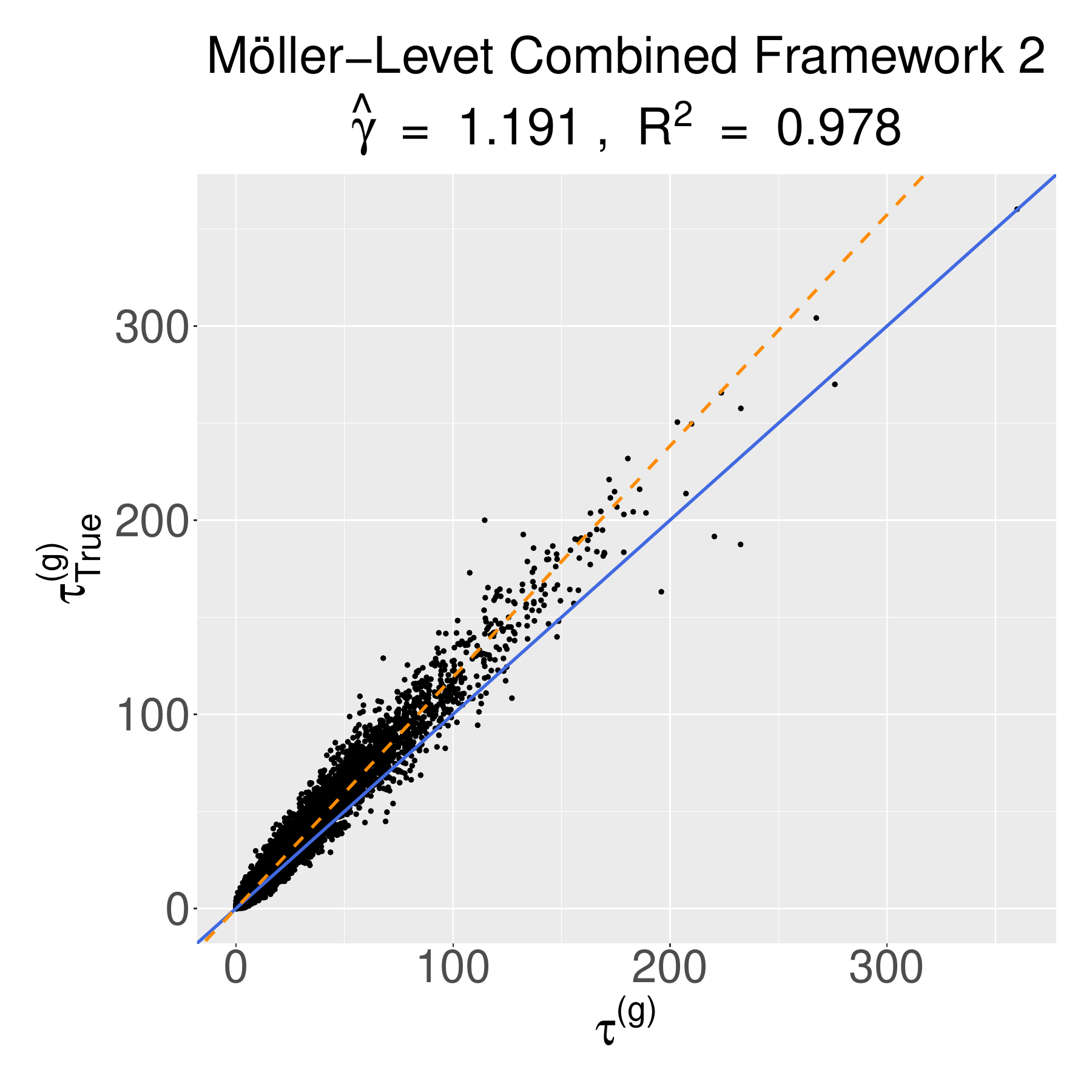}
\caption{Scatter plots of linear model fits using every gene for each Framework (Framework 1 the proposed method from Section \ref{sec:2.3}, Framework 2 linear mixed effects cosinor regression). The y-axis denotes a Wald test statistic obtained from internal circadian time data ($\tau^{(g)}_{\text{True}}$), and the x-axis a Wald test statistic obtained from Zeitgeber time data ($\tau^{(g)}$). The dashed orange line displays the linear model's fit, and the blue line denotes the fit of a linear model when $\hat{\gamma}=1$. Each data point in the scatter plot represents a Wald test statistic obtained for a specific gene.}
    \label{fig:test_w}
\end{figure*}

\clearpage
\newpage

\begin{table*}[!h]
	\caption{Comparison of both frameworks using every gold standard circadian gene available in each sample population. Framework 1 estimates a linear mixed effects cosinor model with the method proposed in Section \ref{sec:2.3} given Zeitgeber time (ZT), while Framework 2 estimates a linear mixed effects cosinor model given ZT. The regression parameter estimate $\hat{\gamma}$ is listed alongside the coefficient of determination ($R^2$), where the latter is in parentheses. Bold values indicate a value of $\hat{\gamma}$ that is closest to one (up to the first three decimal digits), which signifies that the quantities obtained from a framework are closer to the quantities obtained from a linear mixed effects cosinor model estimated with internal circadian time.} \label{tab:app2}
  \centering
 %\resizebox{1.05\textwidth}{!}{
		\begin{tabular}{|l|c|c|c|c|c|c|c|c|c|}
			\hline
   \multirow{1}{*}{Sample Population} & \multirow{1}{*}{Framework}  &  \multicolumn{1}{c|}{$\theta^{(g)}_1$} & \multicolumn{1}{c|}{$\tau^{(g)}$} \\
   \hline
         \multirow{2}{*}{Archer (Control)} & 1 & 0.991 (0.994) & \textbf{1.056 (0.996)} \\ 
        & 2 & \textbf{1.005 (0.992)} & 1.090 (0.995) \\ 
         \hline 
         \multirow{2}{*}{Archer (Intervention)} & 1 & \textbf{1.093 (0.962)} & \textbf{1.055 (0.959)} \\ 
    & 2 & 1.099 (0.963) & 1.081 (0.959) \\
         \hline 
         \multirow{2}{*}{Archer (Combined)}& 1 & \textbf{1.070 (0.971)} & \textbf{1.068 (0.992)} \\
         & 2 & 1.083 (0.966) & 1.097 (0.991) \\
         \hline 
         \multirow{2}{*}{Braun} & 1 & 0.984 (0.996) & 0.943 (0.979) \\
         & 2 & \textbf{0.994 (0.995)} & \textbf{0.955 (0.974)} \\
         \hline 
         \multirow{2}{*}{M\"{o}ller-Levet (Control)} & 1 & \textbf{1.070 (0.988)} & \textbf{0.991 (0.955)} \\ 
         & 2 & 1.095 (0.985) & 1.020 (0.943) \\
         \hline 
         \multirow{2}{*}{M\"{o}ller-Levet (Intervention)} & 1 & \textbf{1.065 (0.991)} & \textbf{1.114 (0.981)} \\ 
         & 2 & 1.083 (0.990) & 1.146 (0.982) \\
         \hline 
         \multirow{2}{*}{M\"{o}ller-Levet (Combined)} & 1 & \textbf{1.085 (0.994)} & \textbf{1.050 (0.985)} \\ 
         & 2 & 1.098 (0.993) & 1.069 (0.985) \\
         \hline 
\end{tabular} %}
\end{table*}

\newpage
\clearpage

\begin{figure*}[!h]
\centering
\includegraphics[scale=0.13]{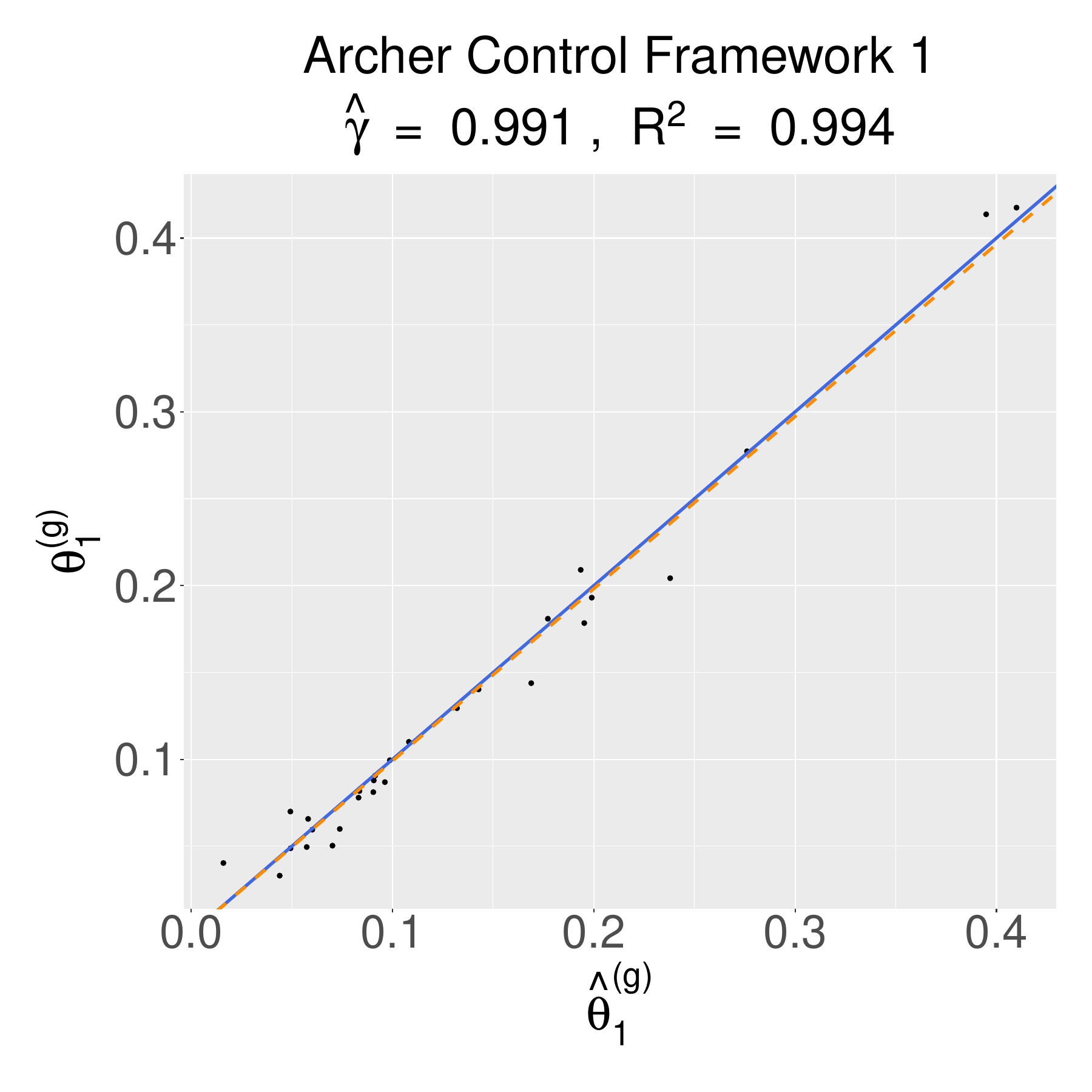}
\includegraphics[scale=0.13]{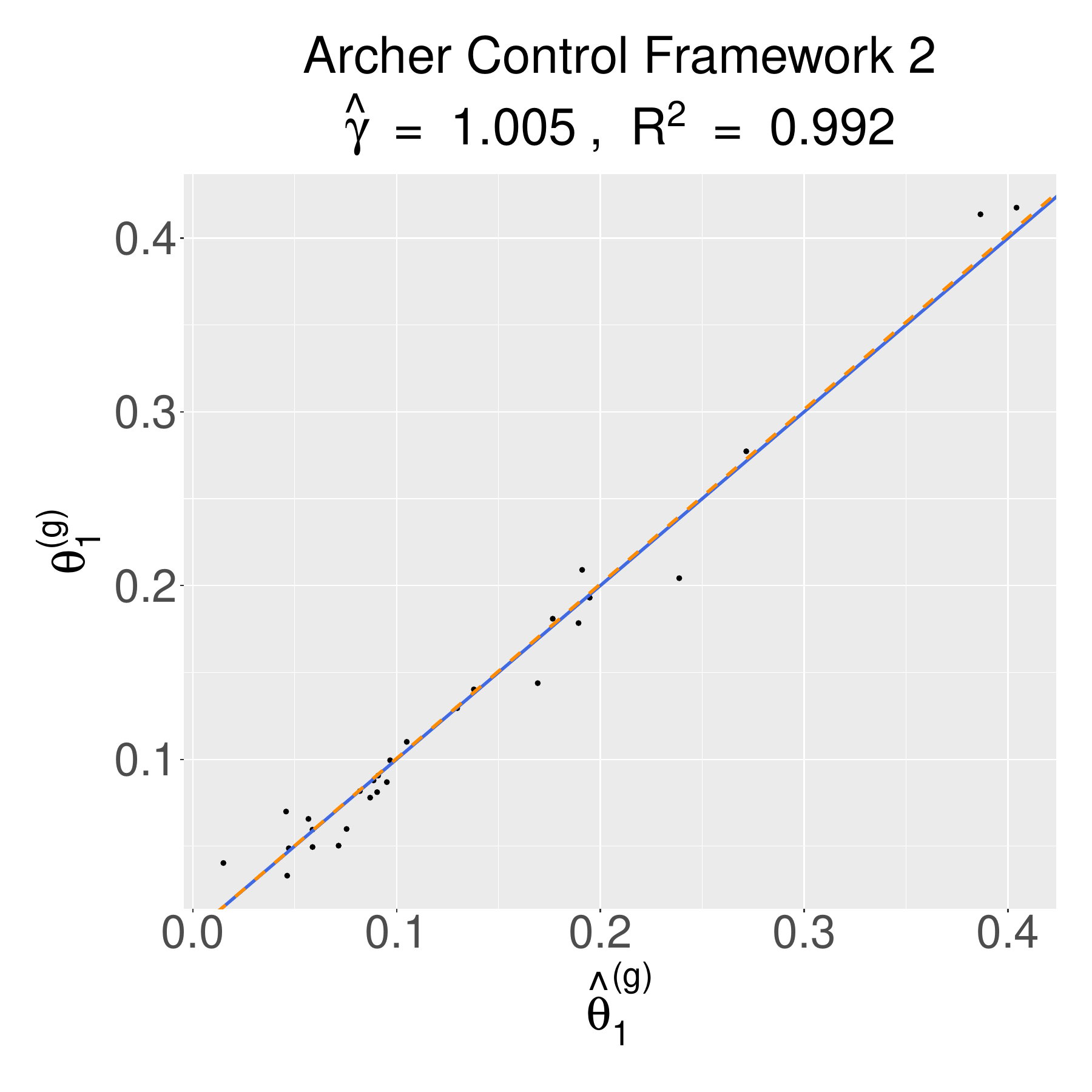} 
\includegraphics[scale=0.13]{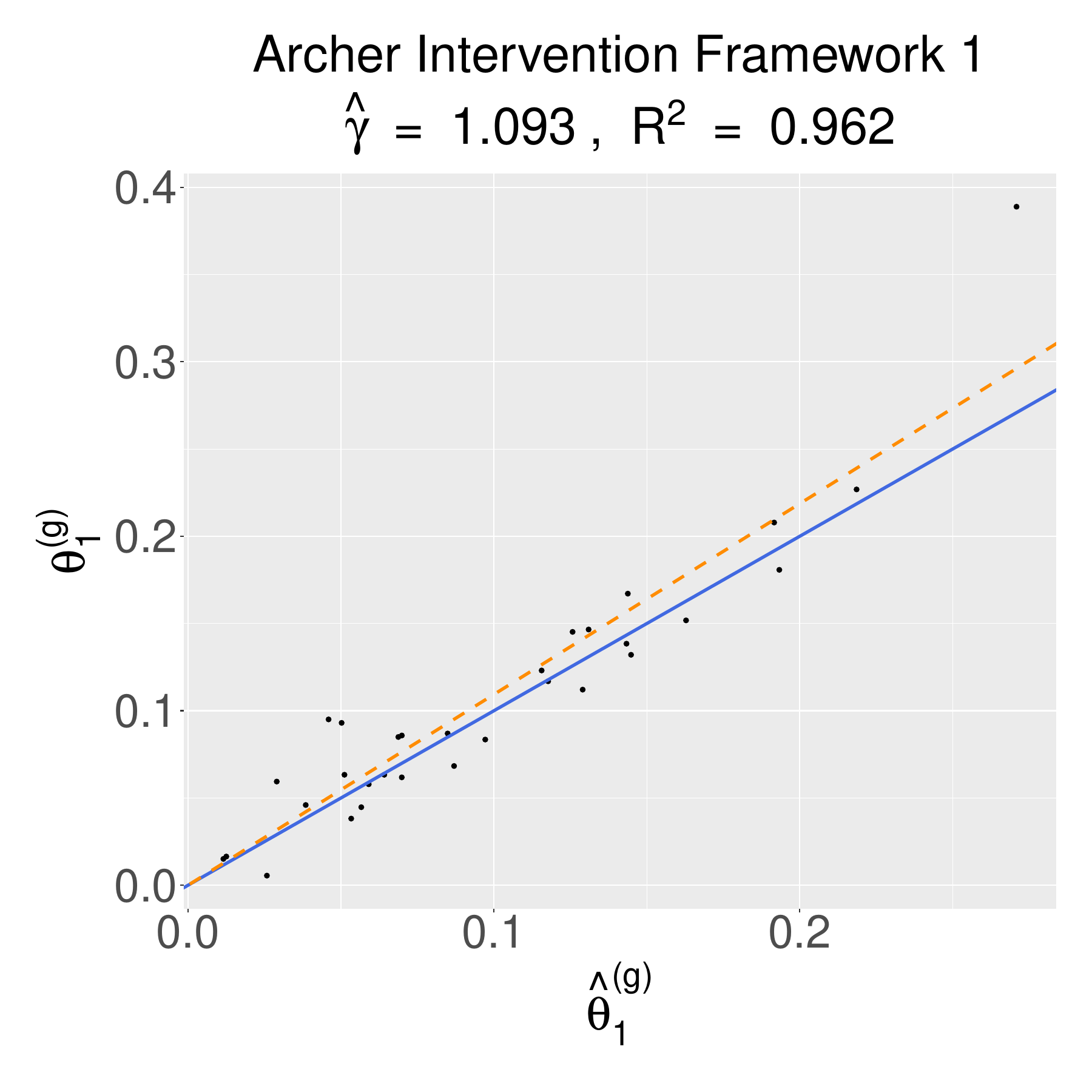} 
\includegraphics[scale=0.13]{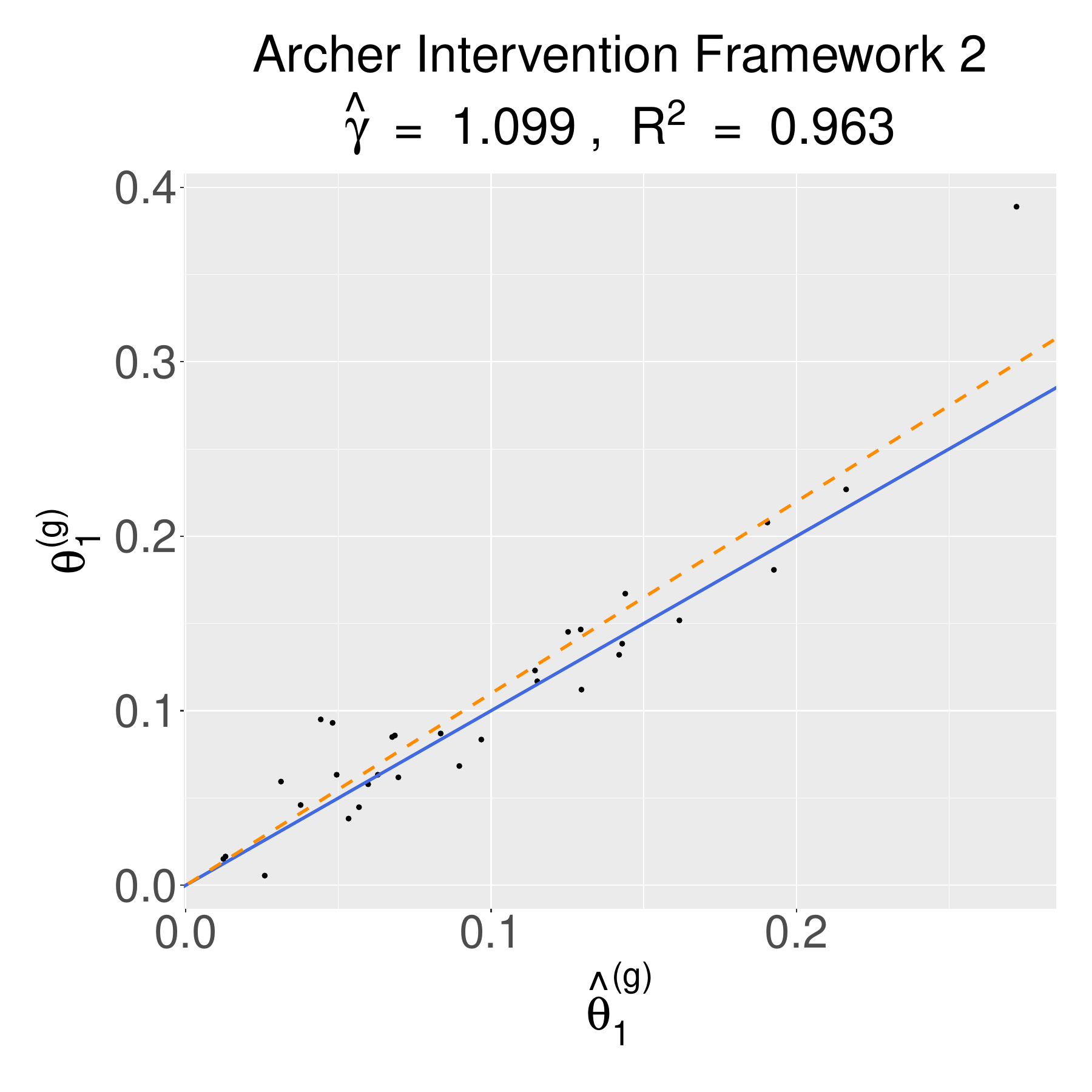}\\
\includegraphics[scale=0.13]{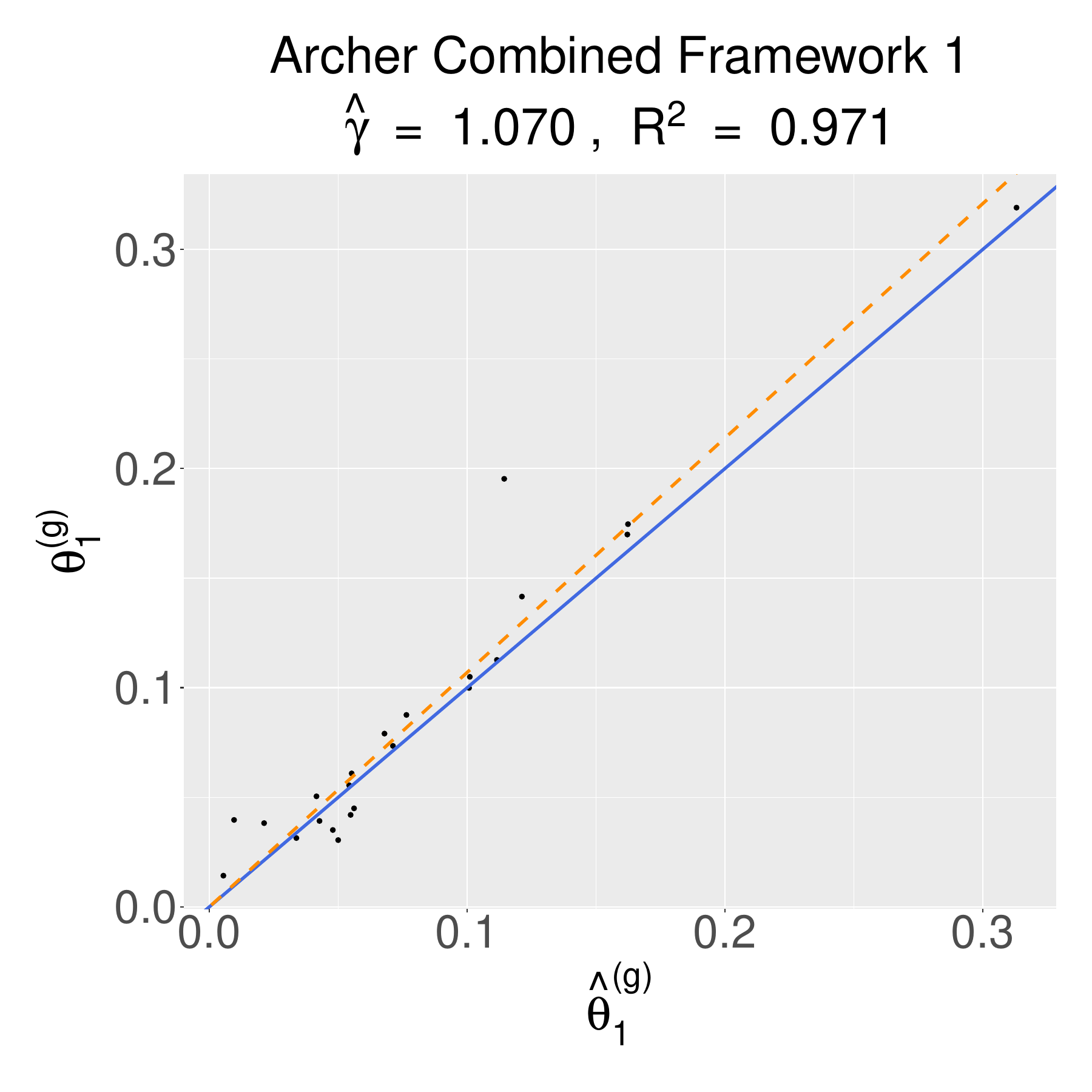}
\includegraphics[scale=0.13]{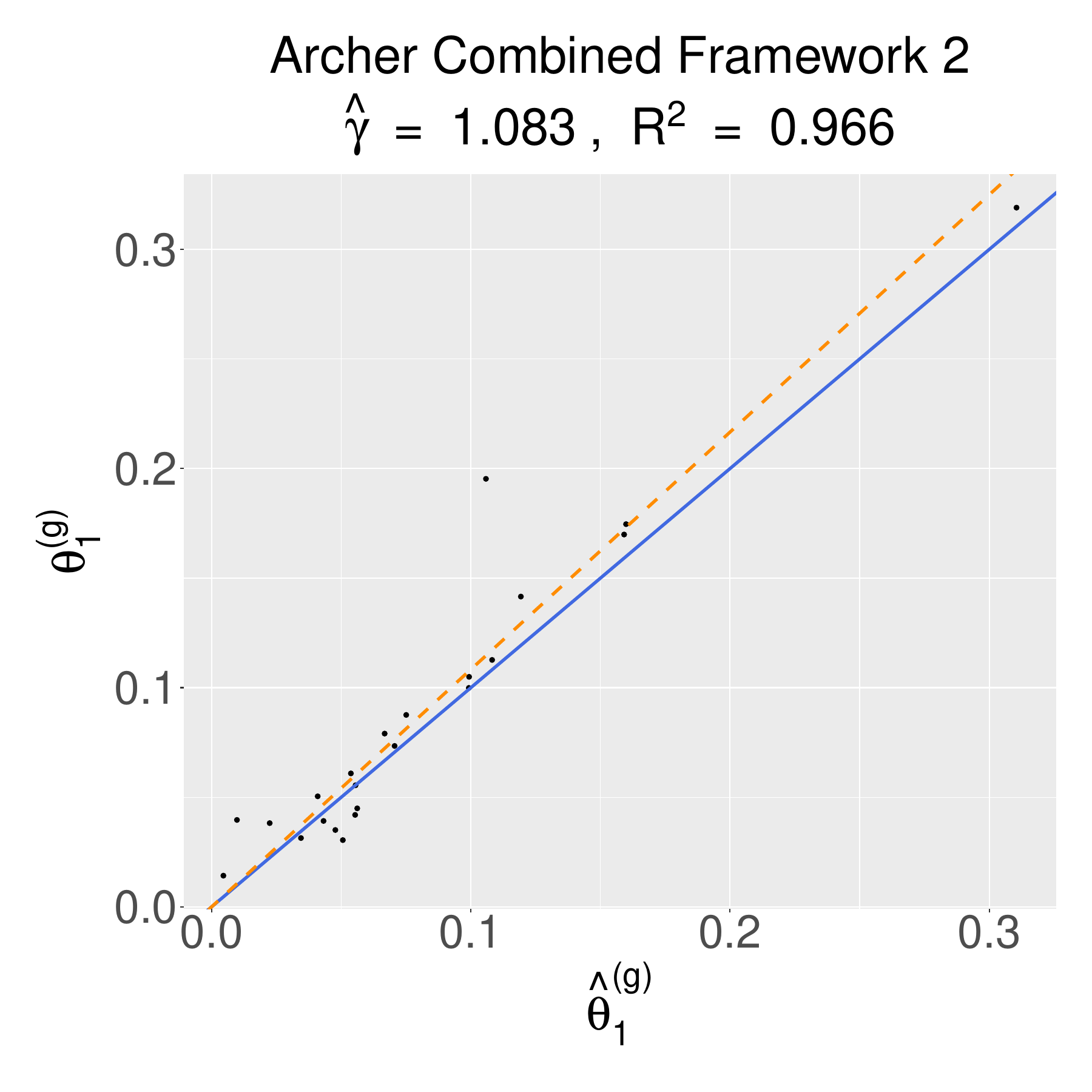} 
\includegraphics[scale=0.13]{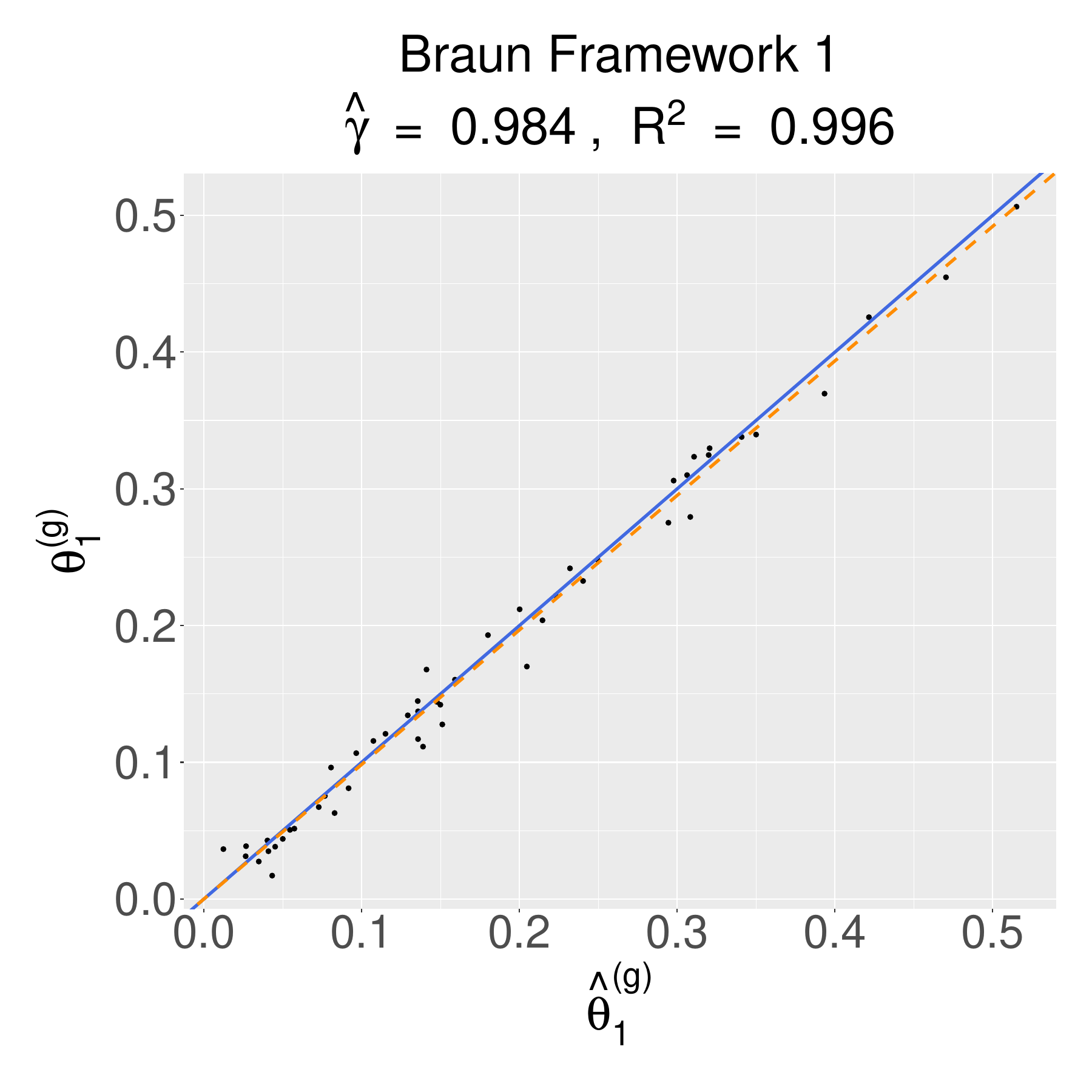}
\includegraphics[scale=0.13]{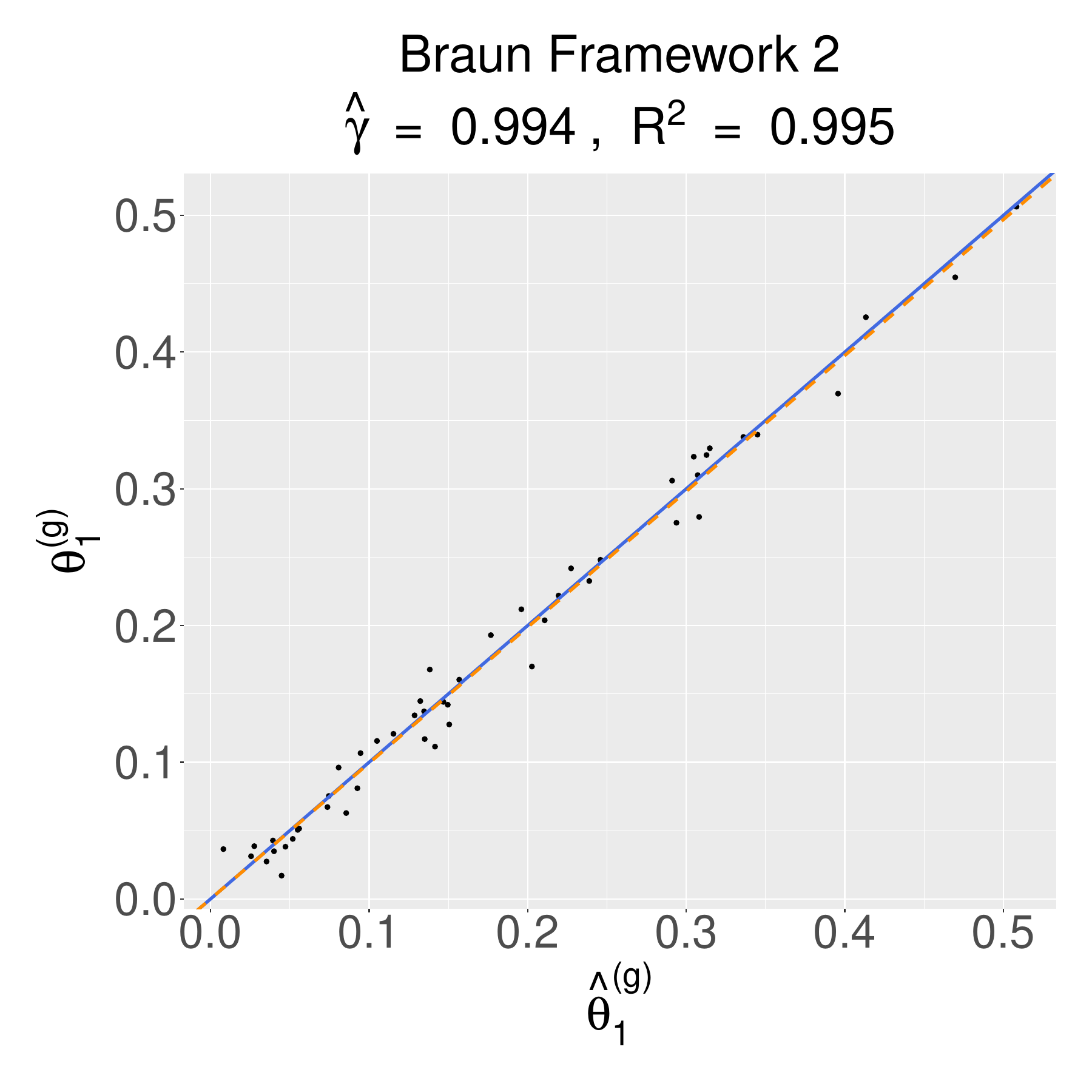} \\
\includegraphics[scale=0.13]{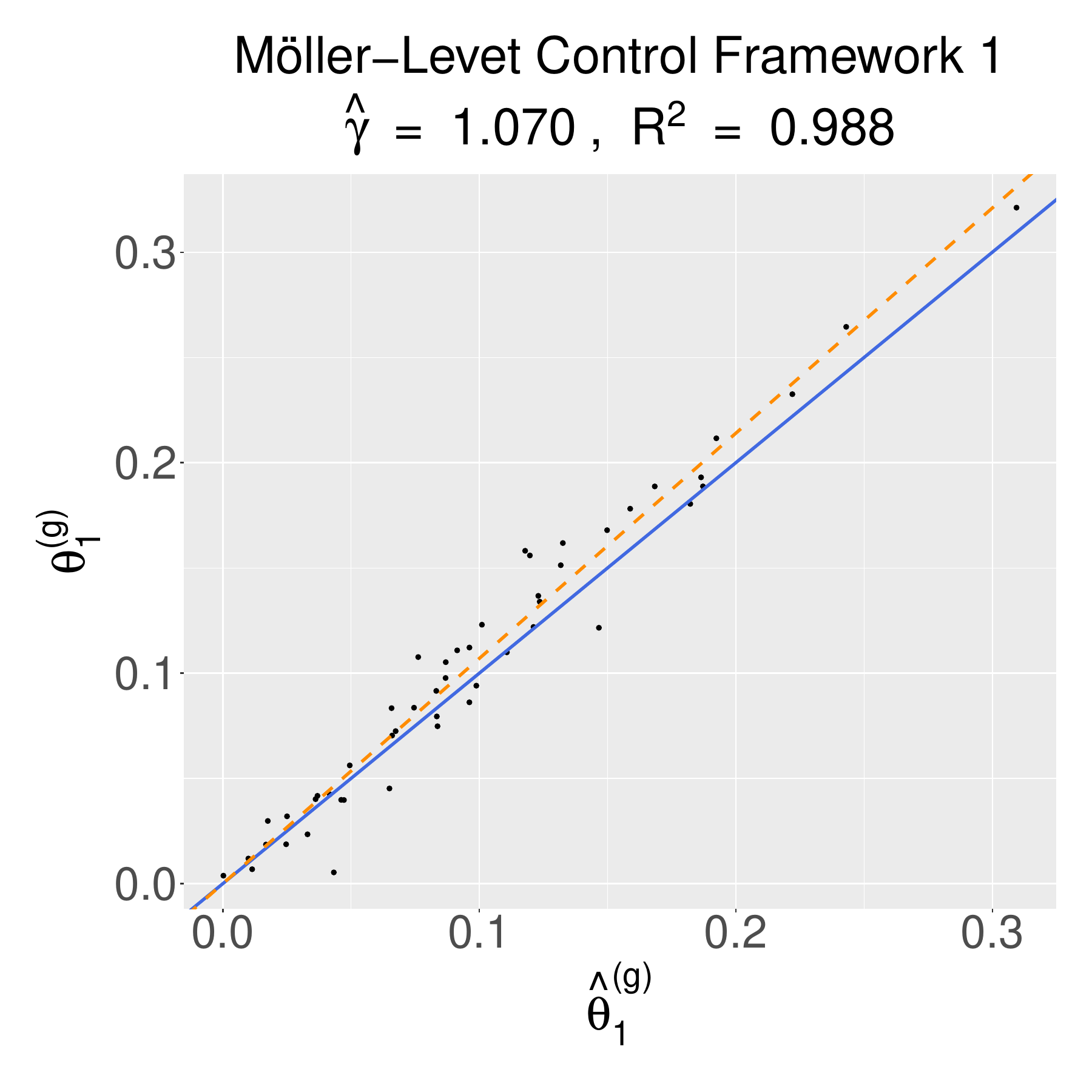} 
\includegraphics[scale=0.13]{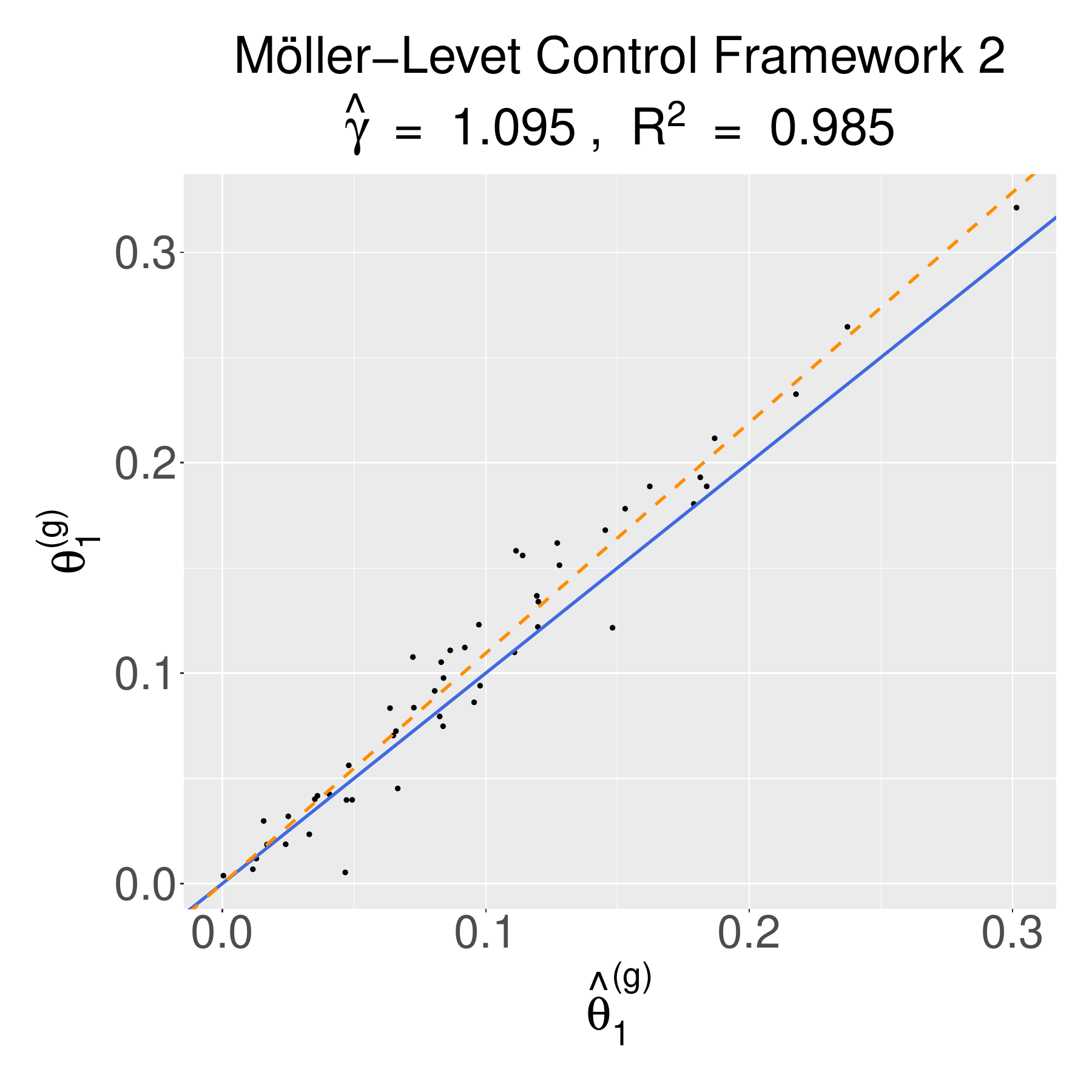}
\includegraphics[scale=0.13]{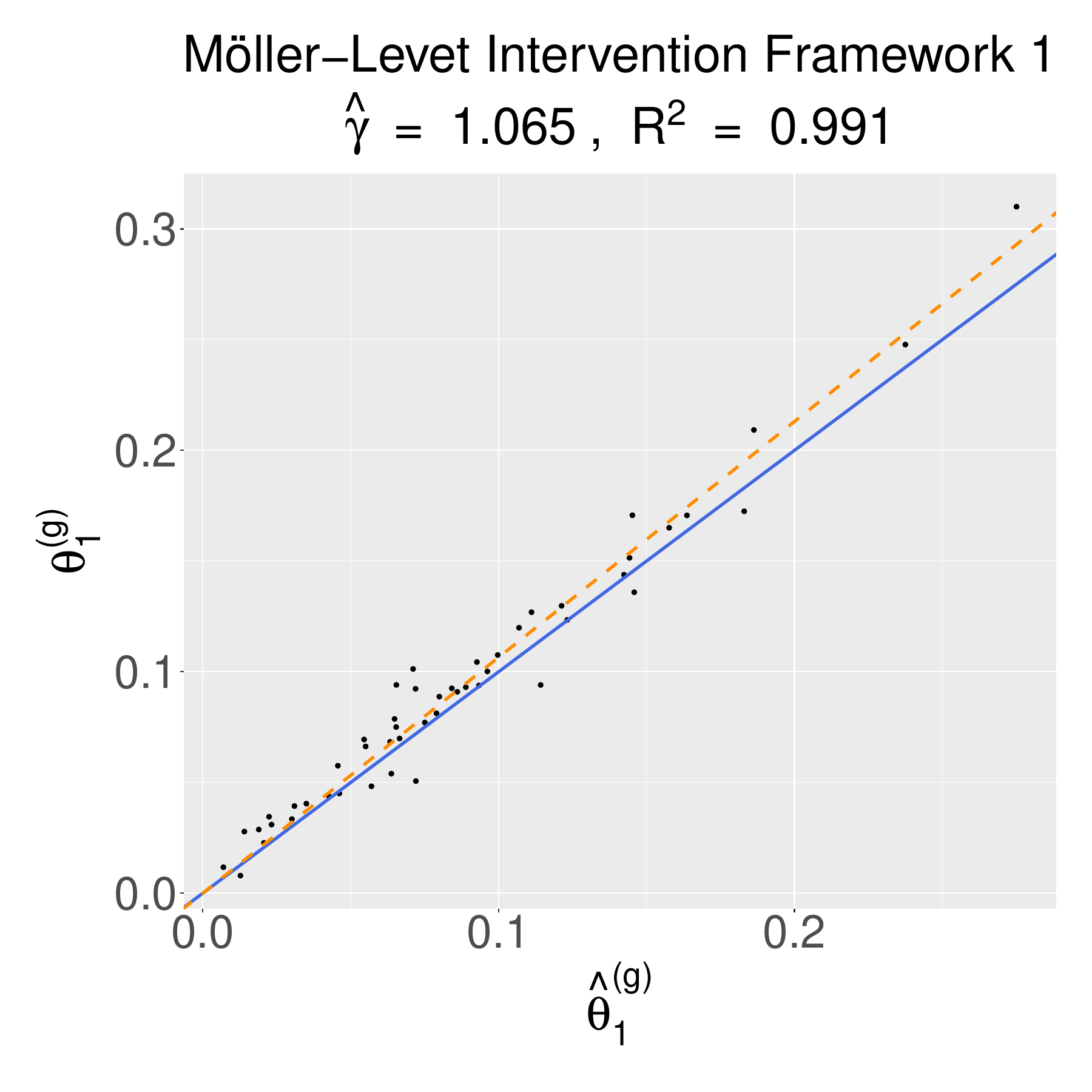}
\includegraphics[scale=0.13]{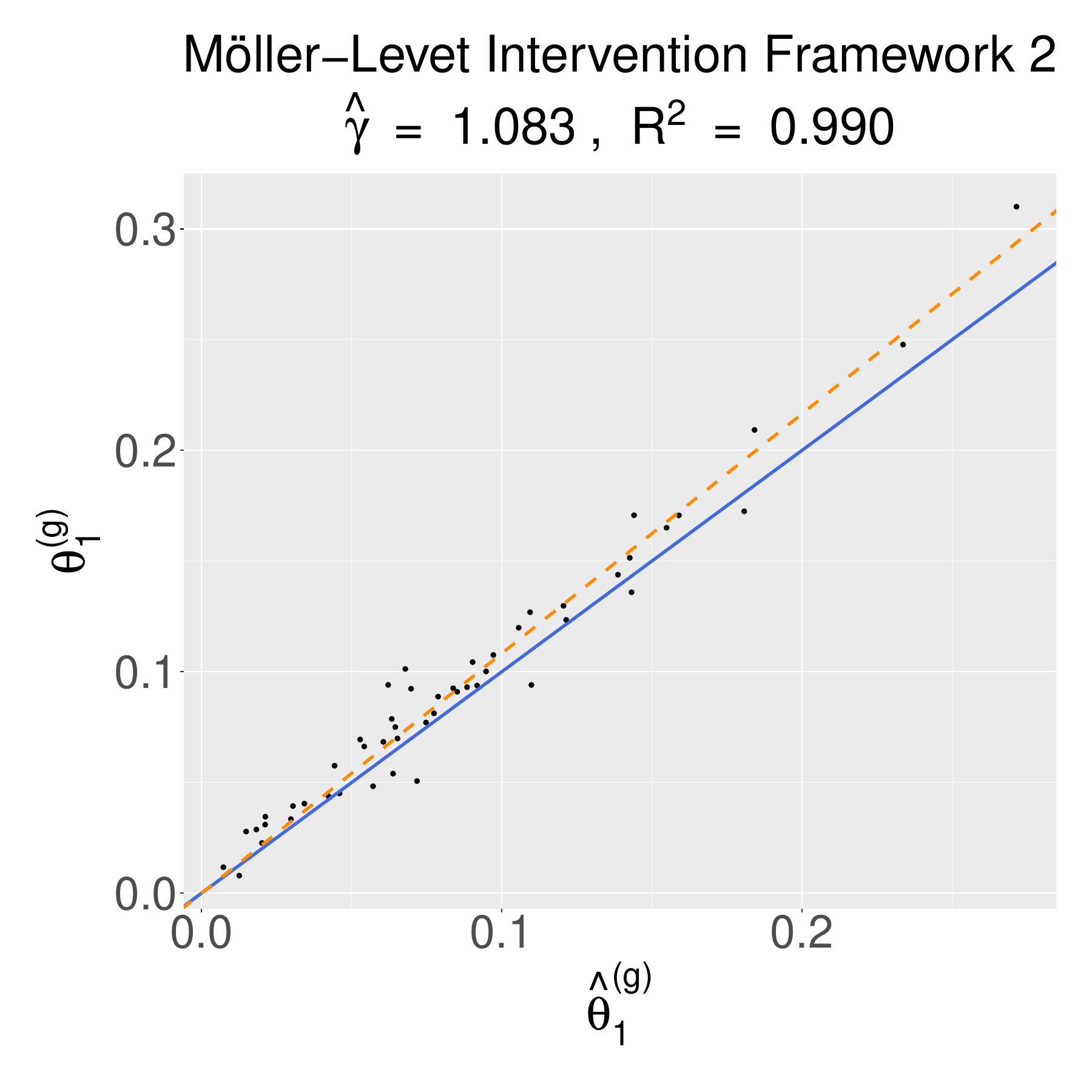} \\
\includegraphics[scale=0.13]{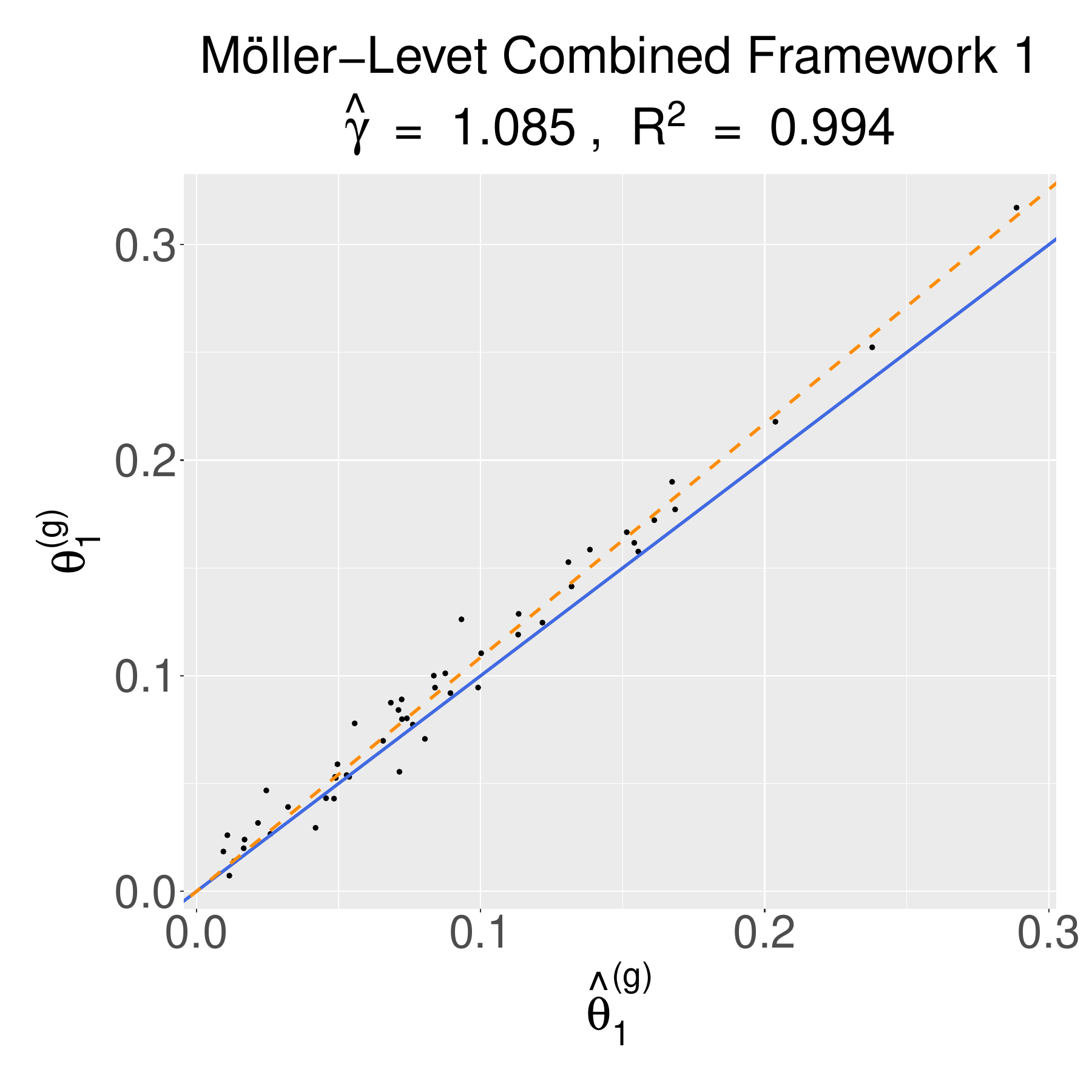}
\includegraphics[scale=0.13]{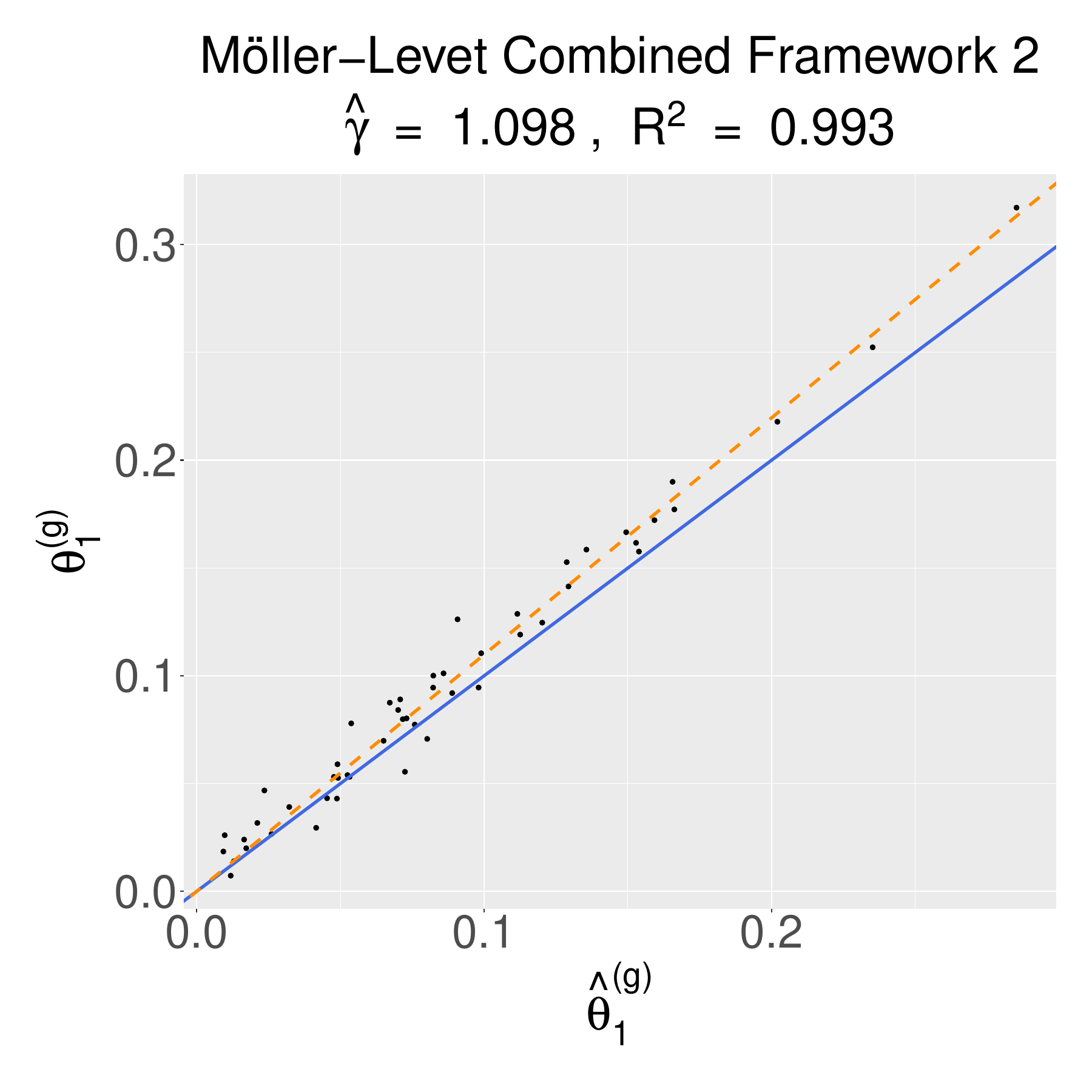}
\caption{Scatter plots of linear model fits using every gold standard circadian gene for each Framework (Framework 1 the proposed method from Section \ref{sec:2.3}, Framework 2 linear mixed effects cosinor regression). The y-axis denotes an amplitude estimate obtained from internal circadian time data ($\theta_1^{(g)}$), and the x-axis an amplitude estimate obtained from Zeitgeber time data ($\hat{\theta}_1^{(g)}$). The dashed orange line displays the linear model's fit, and the blue line denotes the fit of a linear model when $\hat{\gamma}=1$. Each data point in the scatter plot represents an amplitude estimate obtained for a specific gene.}
\label{fig:amp_g}
\end{figure*}

\clearpage
\newpage

\begin{figure*}[!h]
\centering
\includegraphics[scale=0.13]{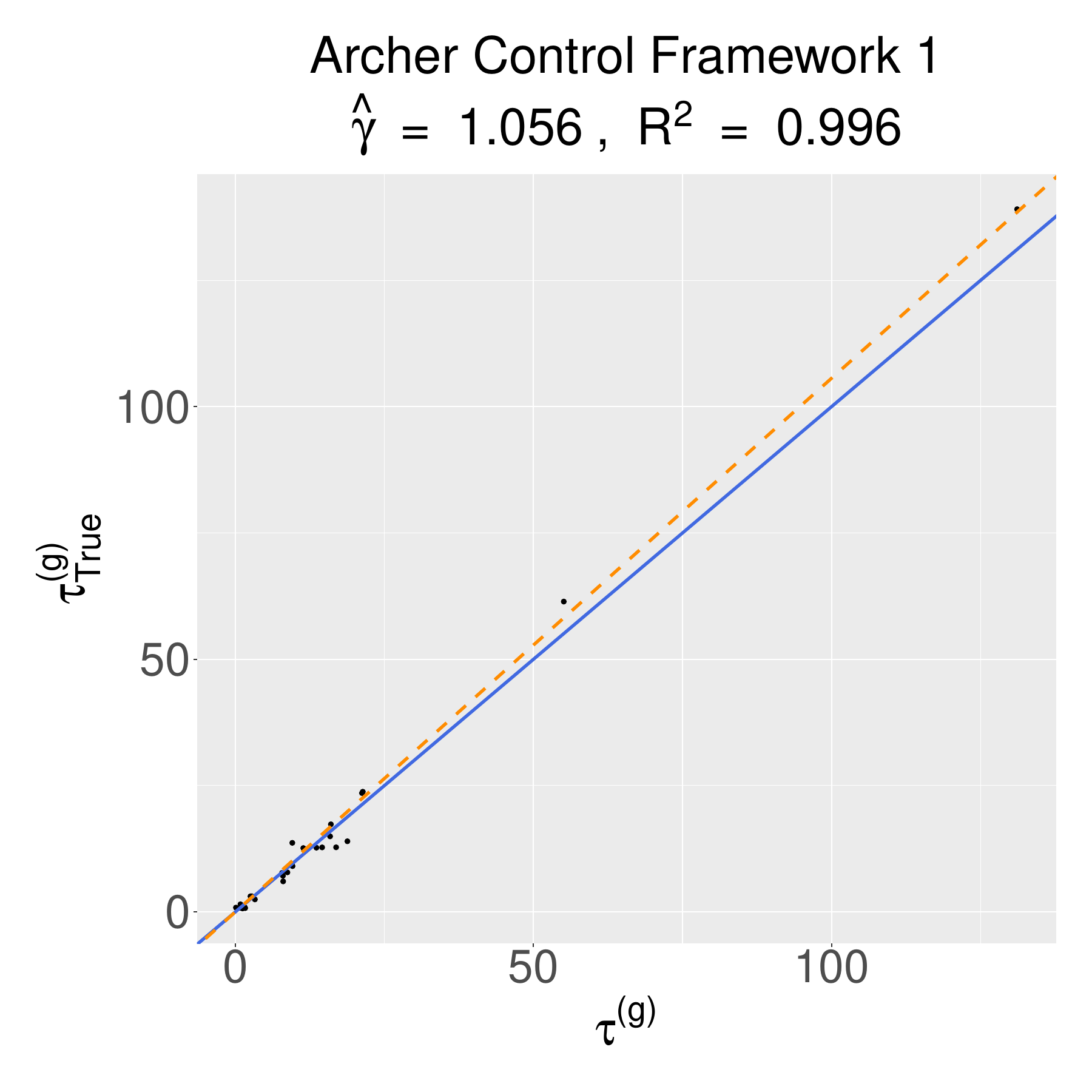}
\includegraphics[scale=0.13]{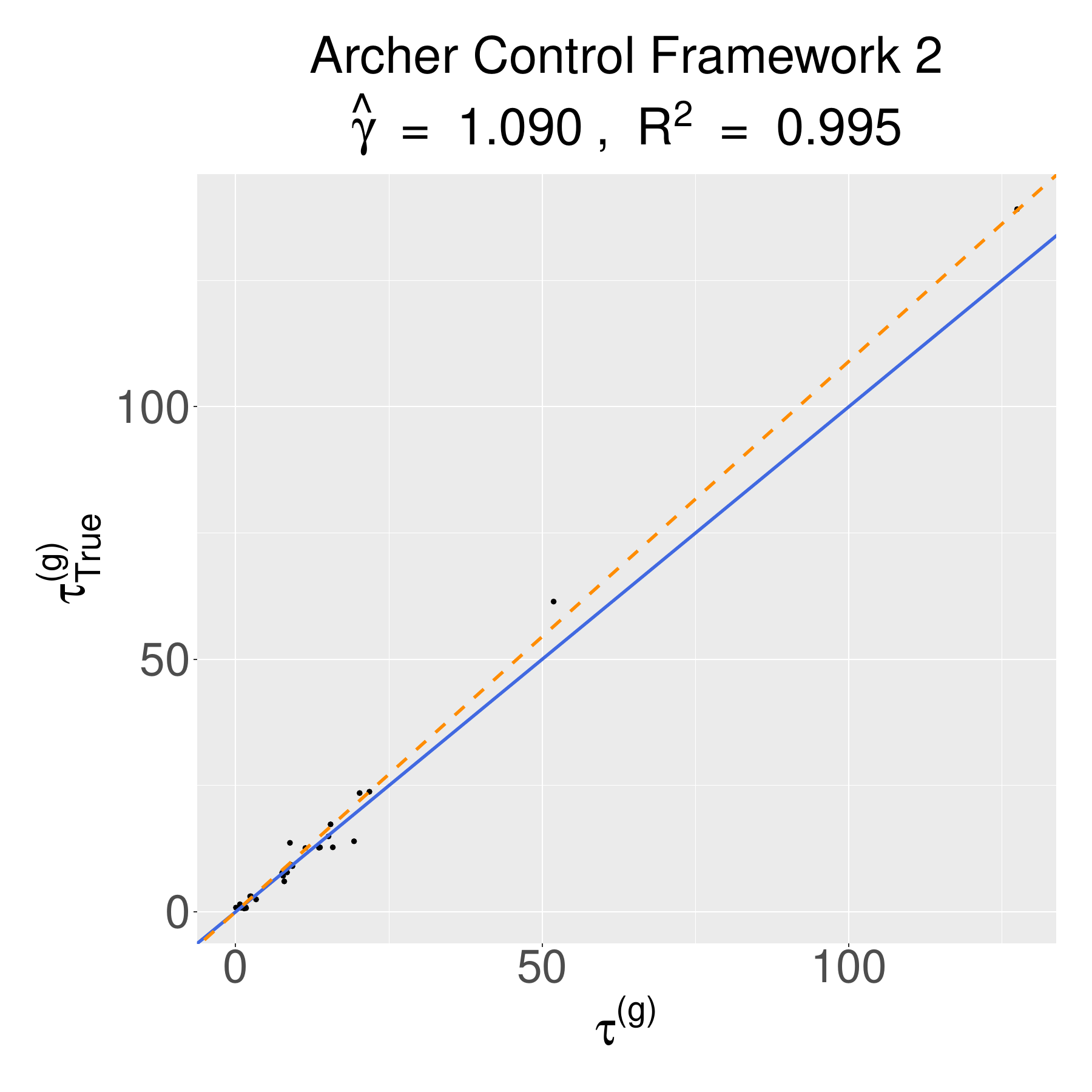} 
\includegraphics[scale=0.13]{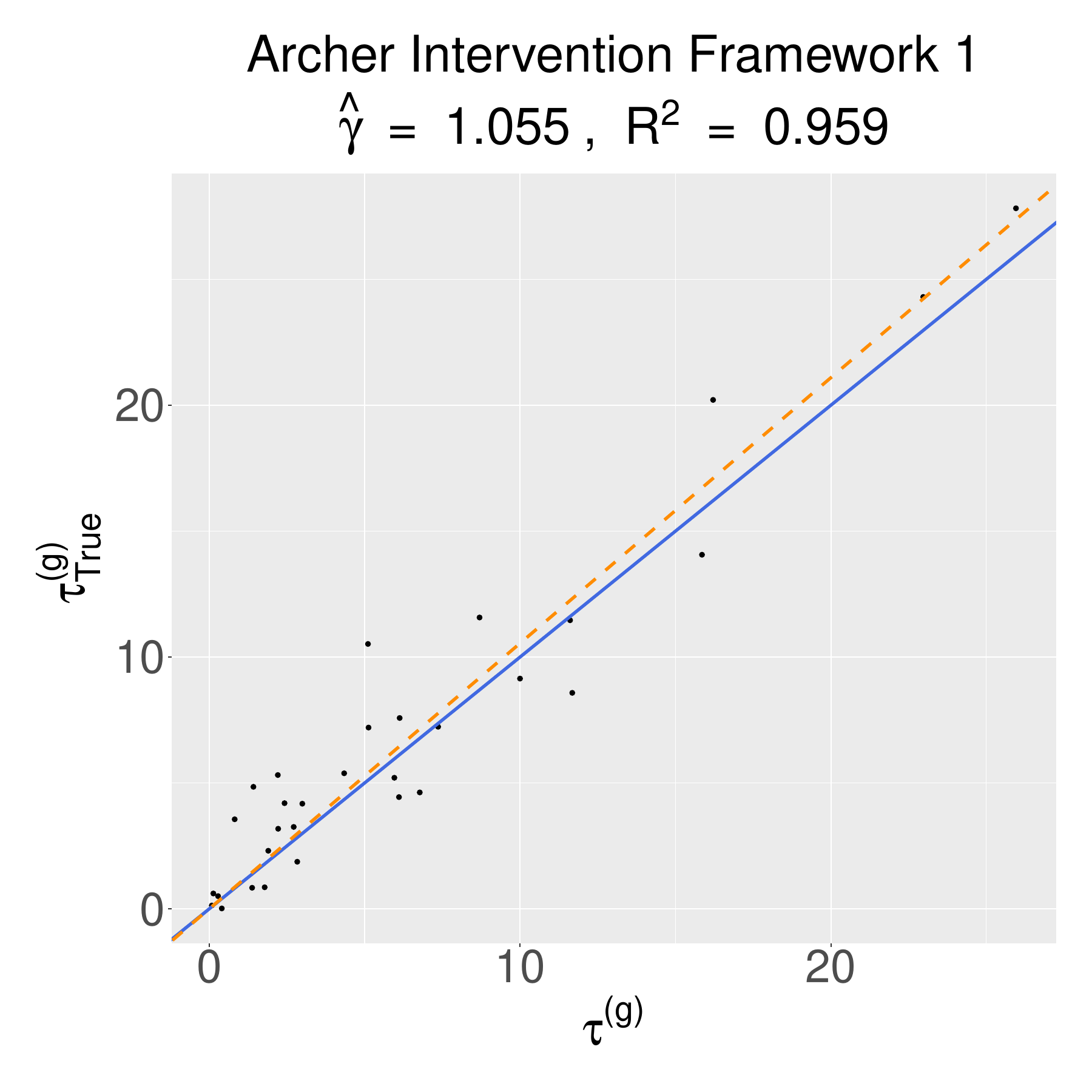} 
\includegraphics[scale=0.13]{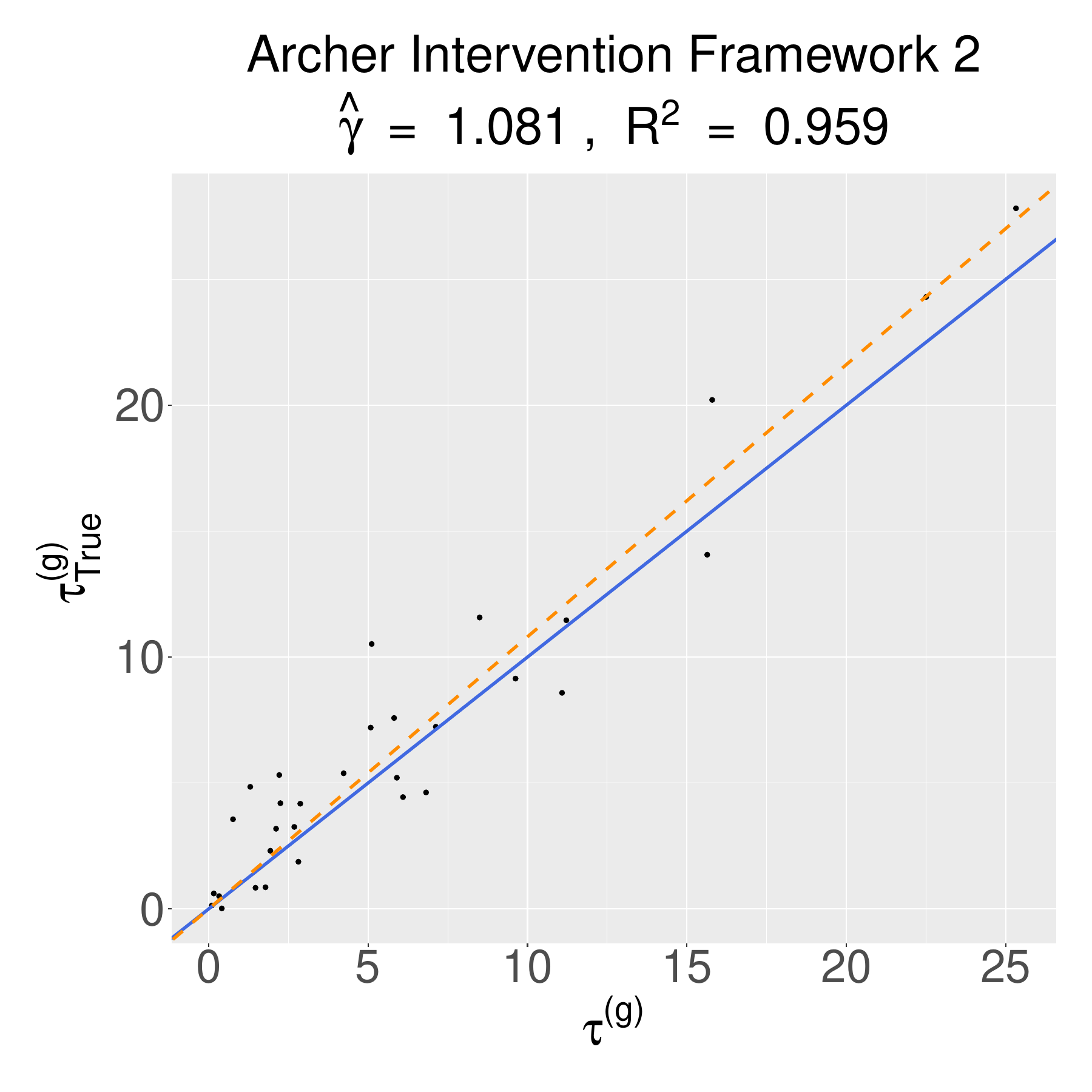}\\
\includegraphics[scale=0.13]{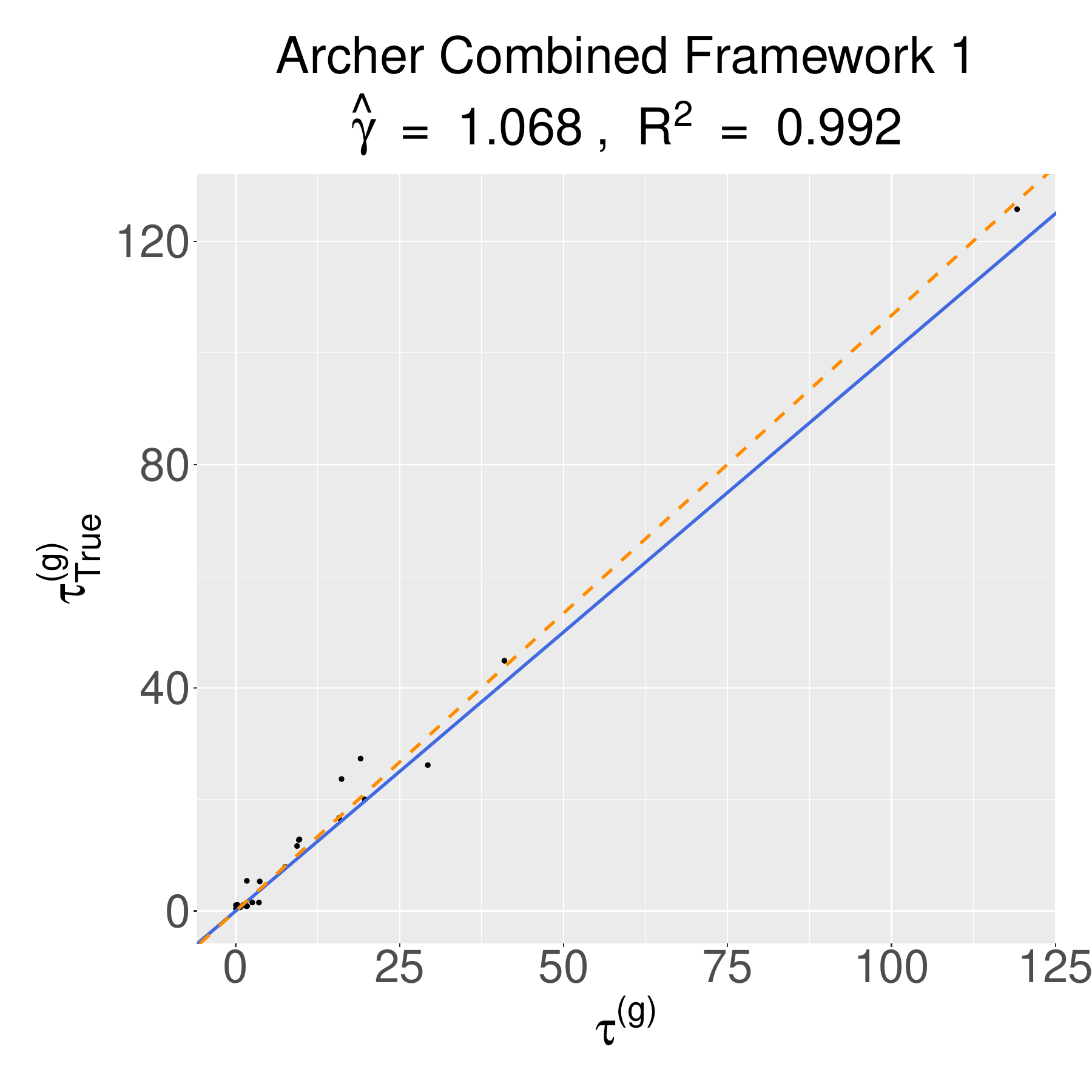}
\includegraphics[scale=0.13]{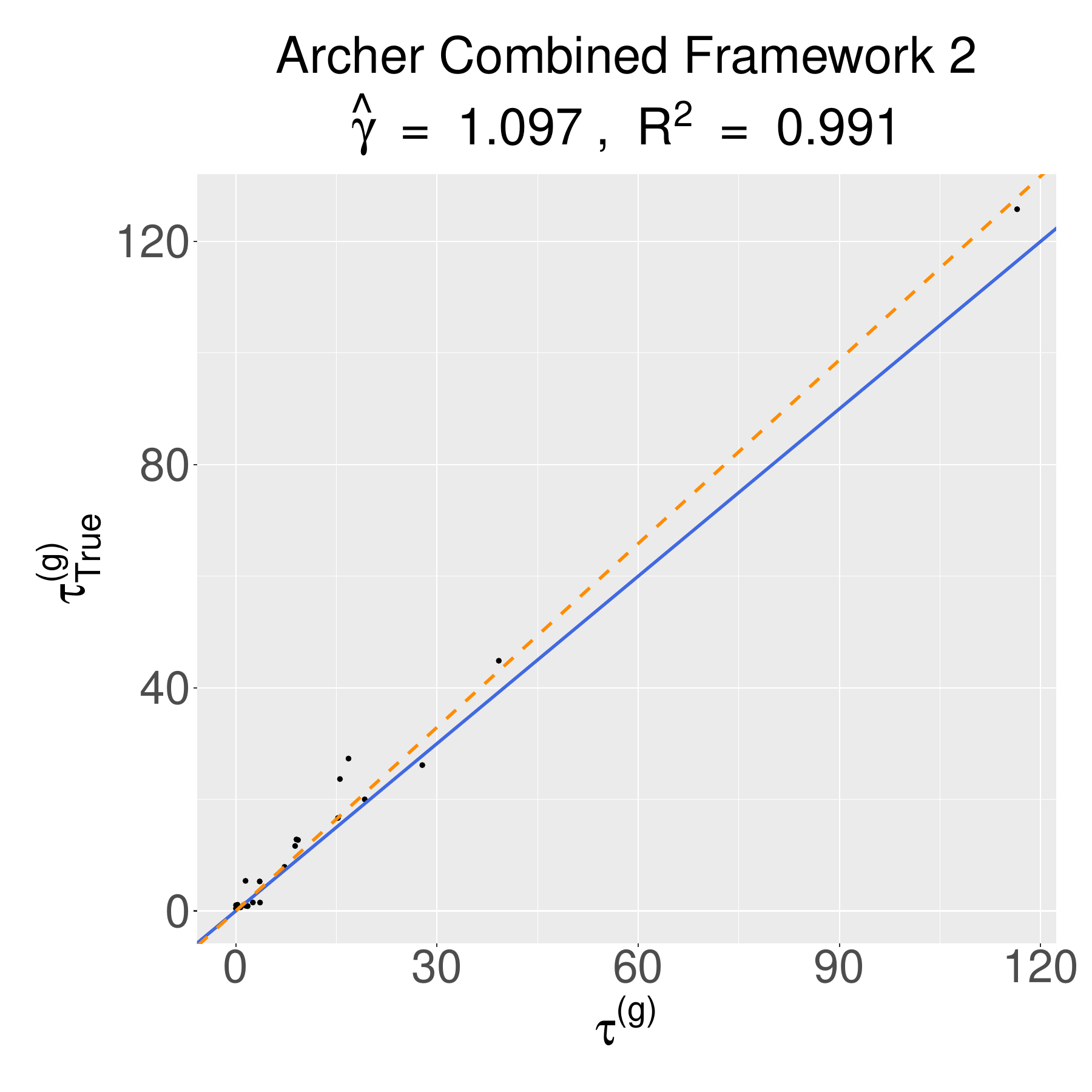} 
\includegraphics[scale=0.13]{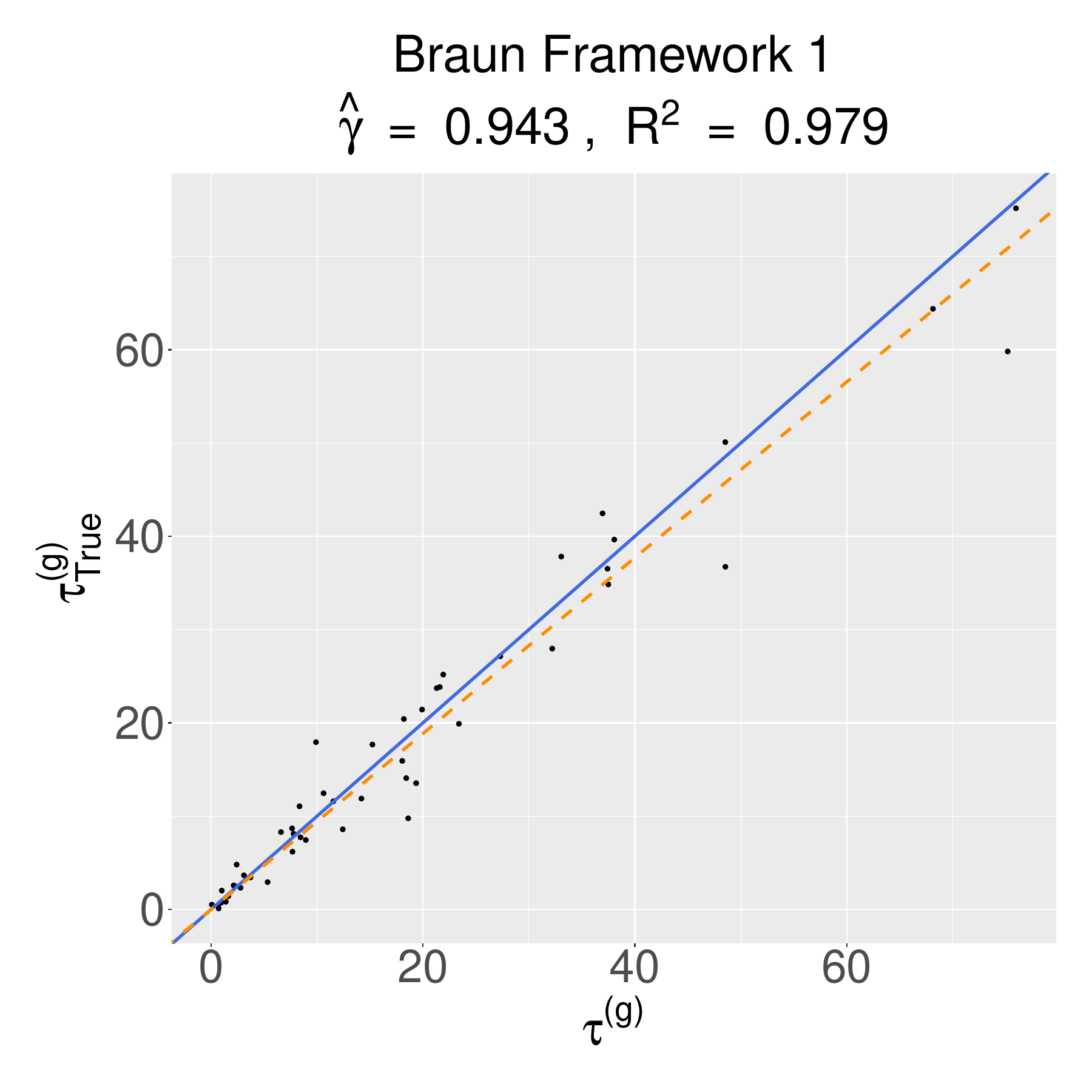}
\includegraphics[scale=0.13]{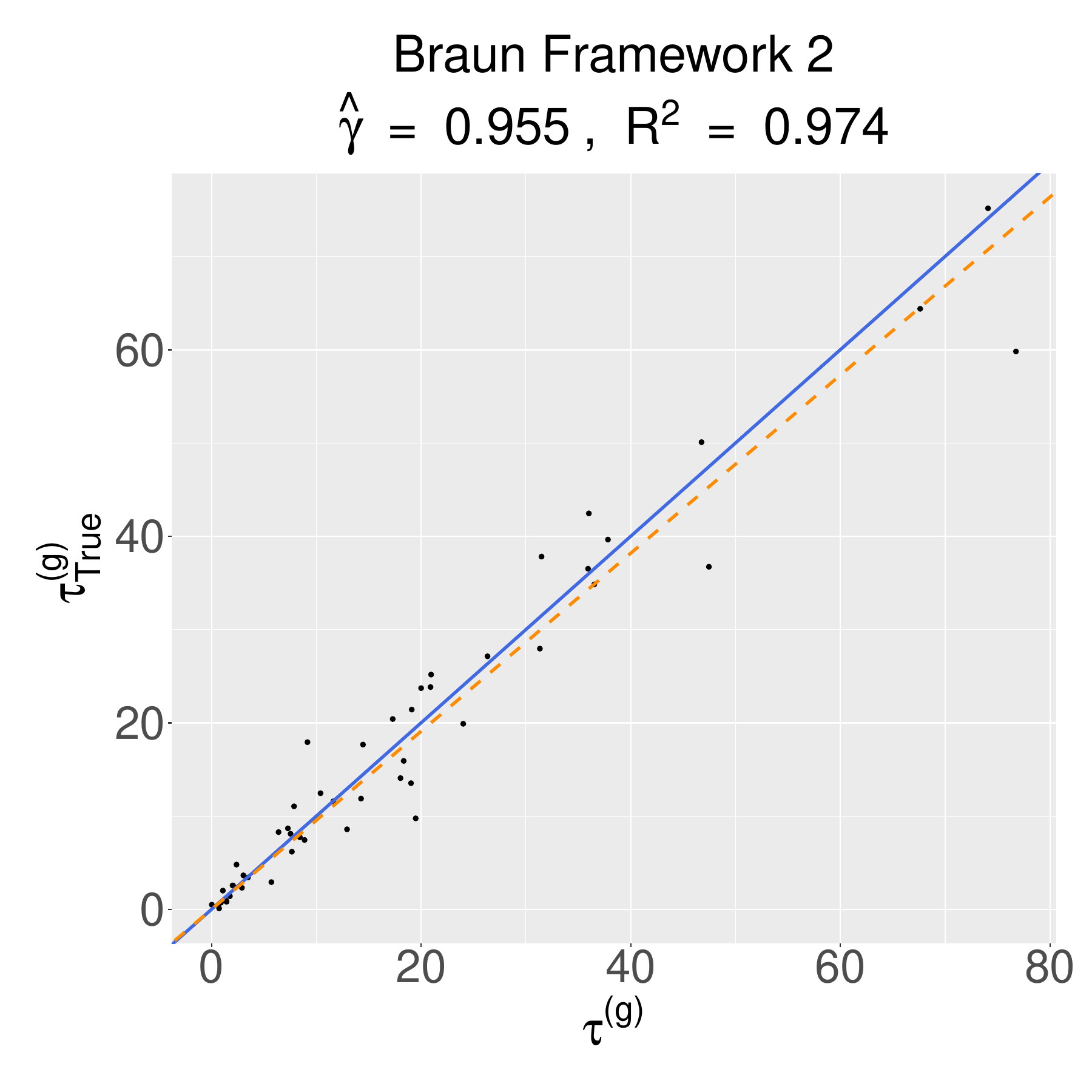} \\
\includegraphics[scale=0.13]{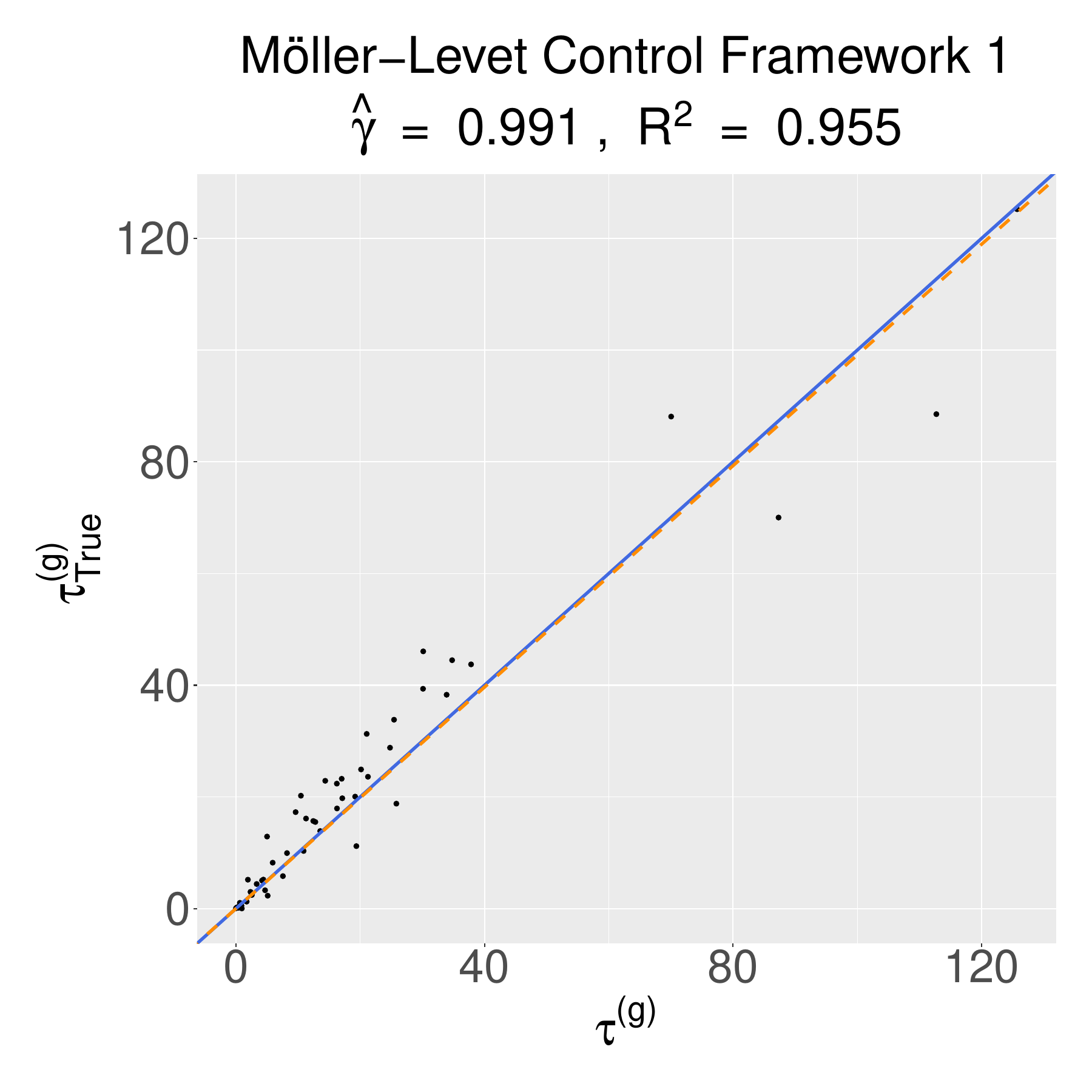} 
\includegraphics[scale=0.13]{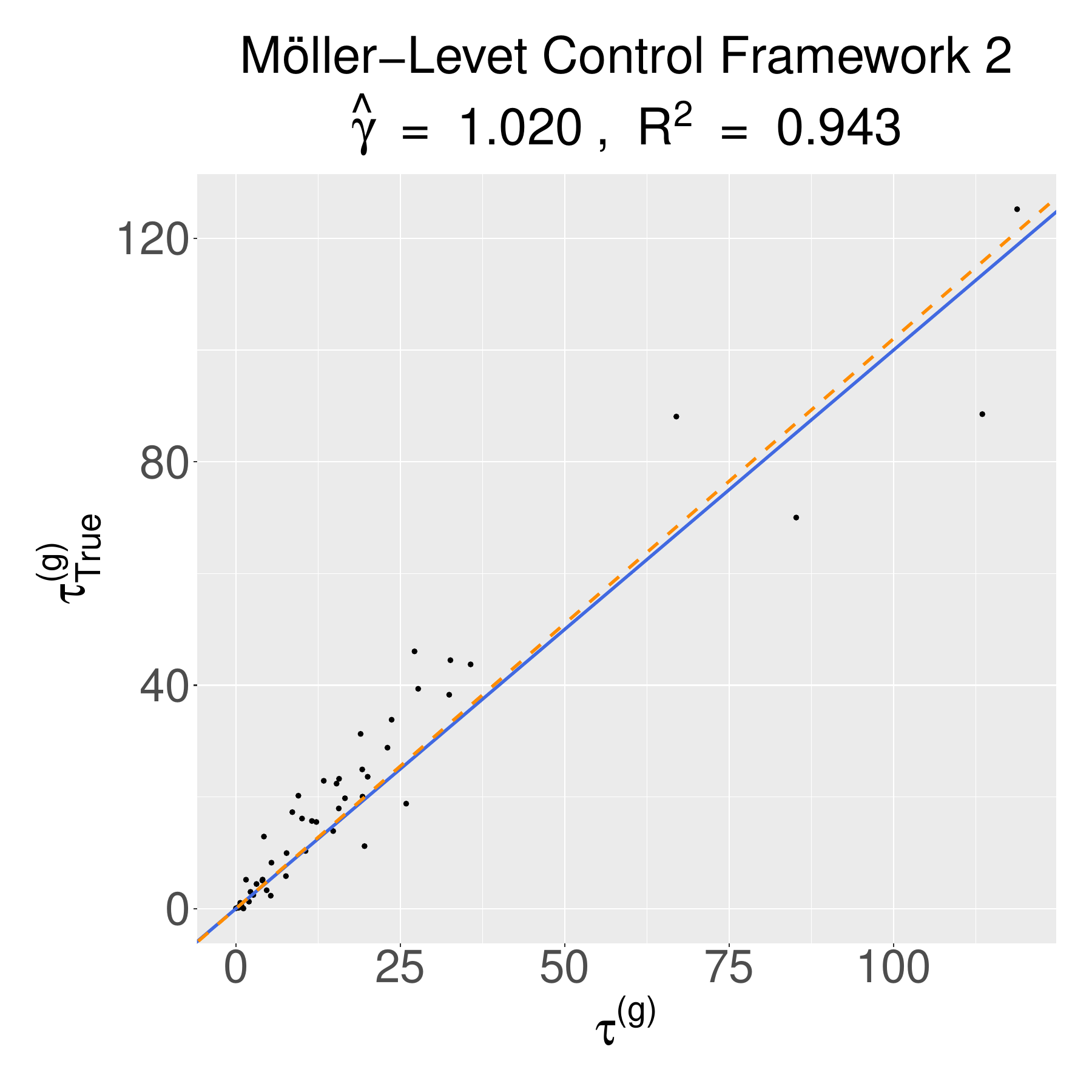}
\includegraphics[scale=0.13]{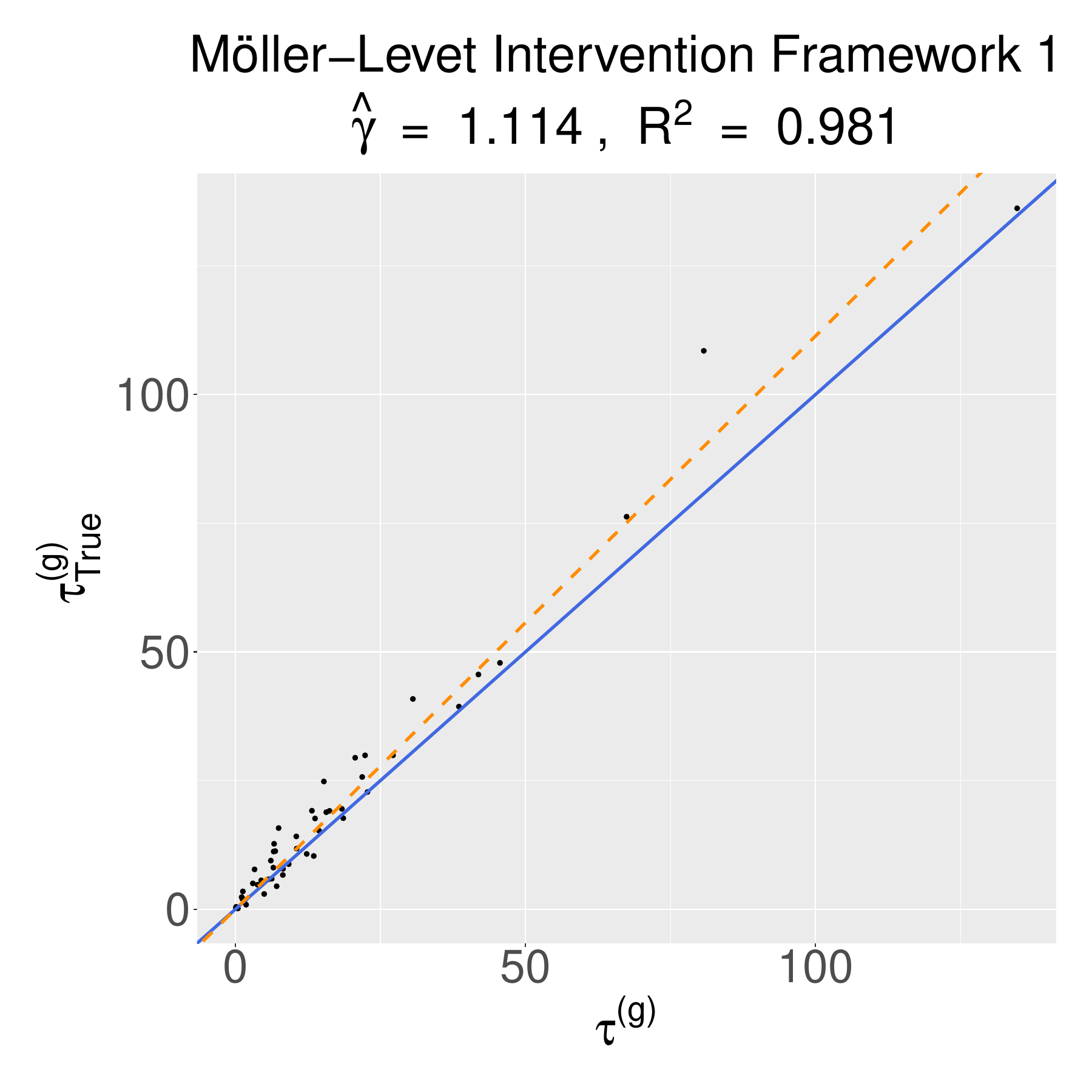}
\includegraphics[scale=0.13]{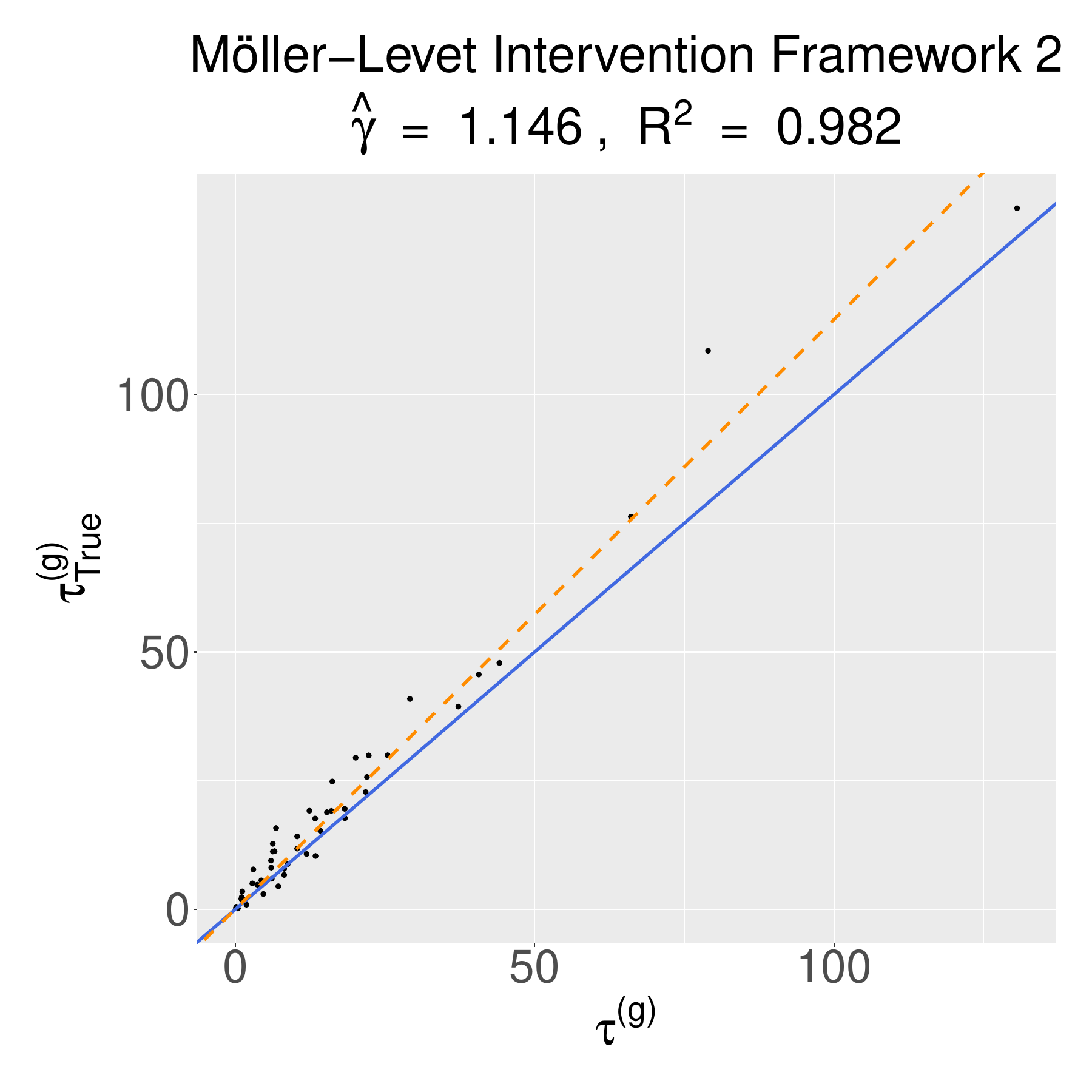} \\
\includegraphics[scale=0.13]{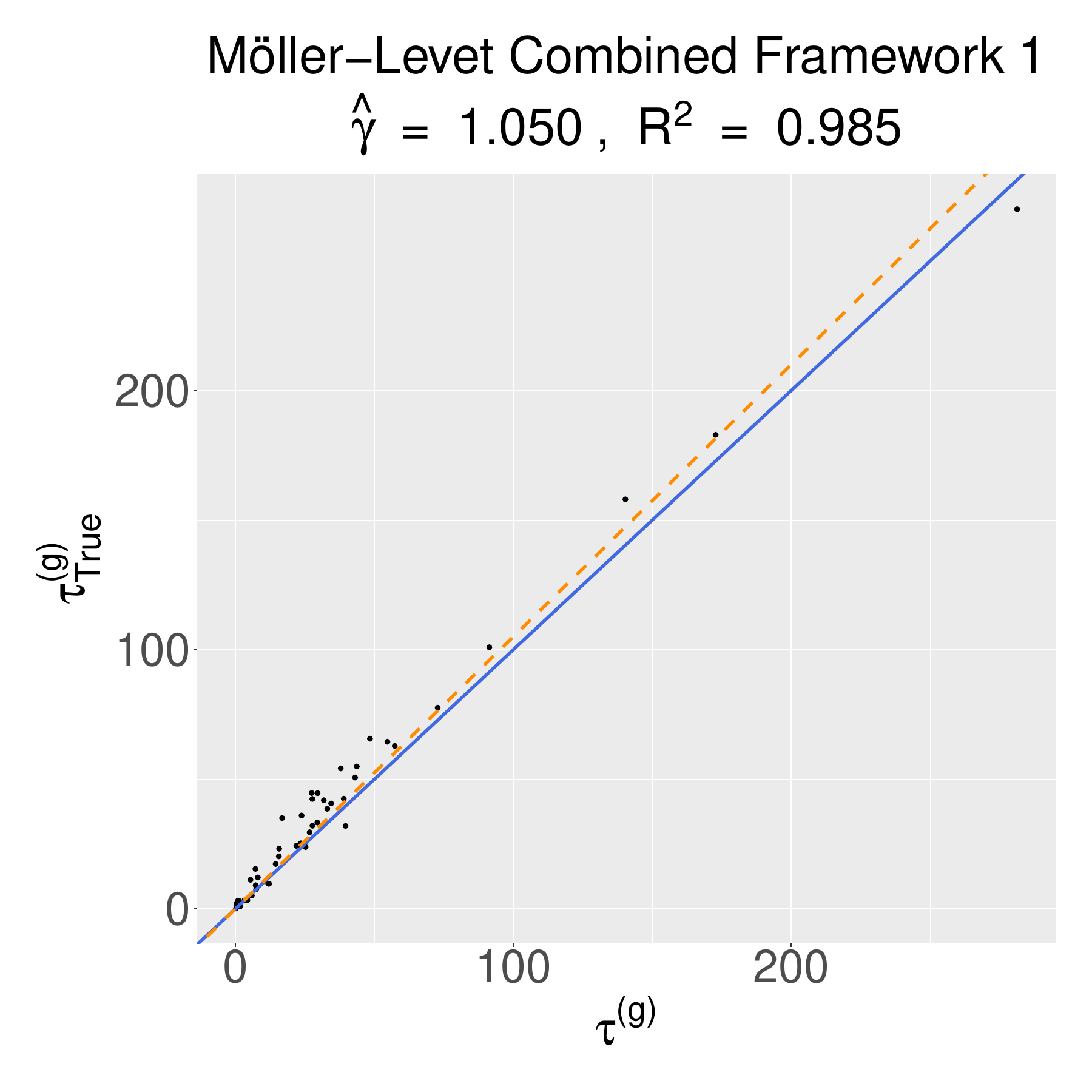}
\includegraphics[scale=0.13]{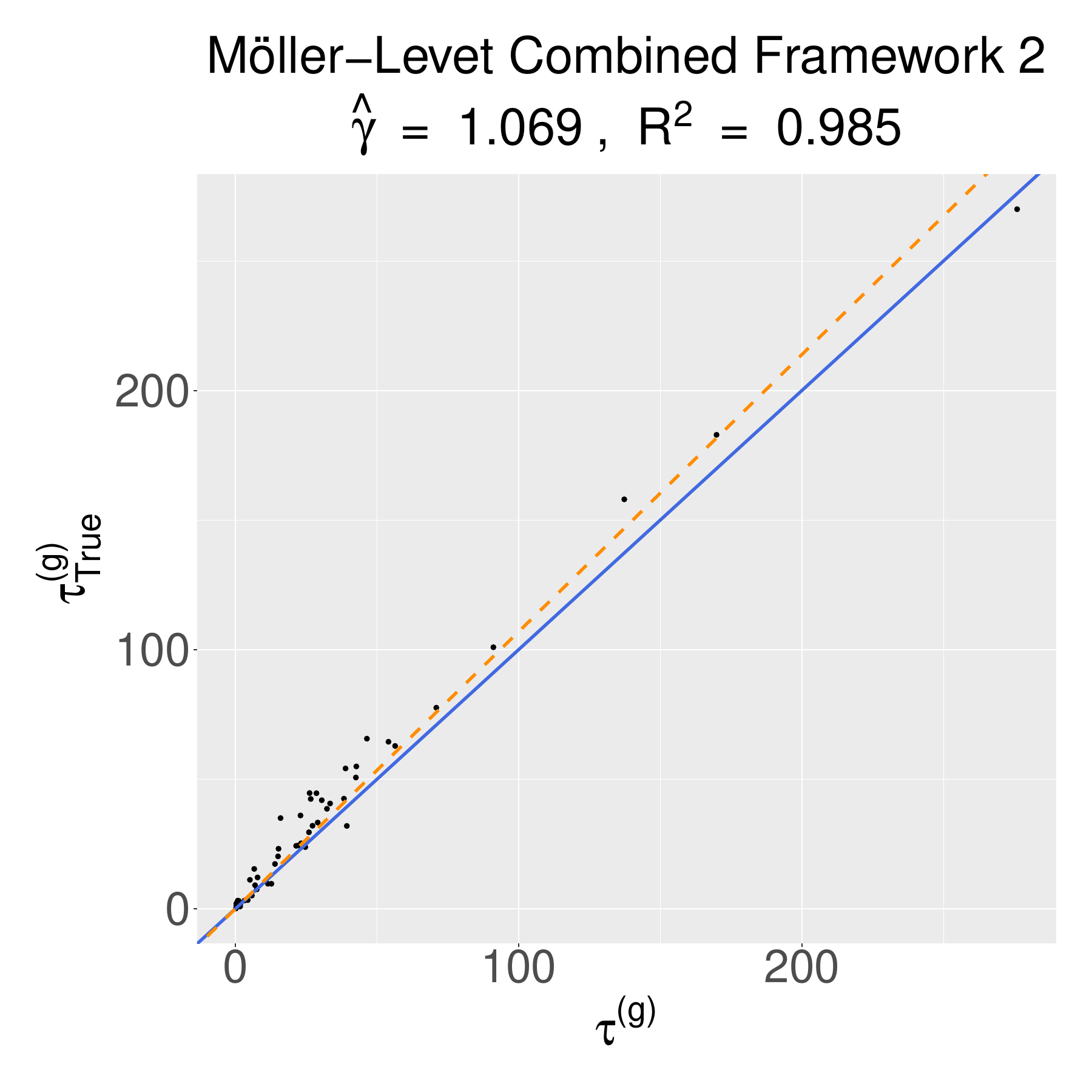}
\caption{Scatter plots of linear model fits using every gold standard circadian gene for each Framework (Framework 1 the proposed method from Section \ref{sec:2.3}, Framework 2 linear mixed effects cosinor regression). The y-axis denotes a Wald test statistic obtained from internal circadian time data ($\tau^{(g)}_{\text{True}}$), and the x-axis a Wald test statistic obtained from Zeitgeber time data ($\tau^{(g)}$). The dashed orange line displays the linear model's fit, and the blue line denotes the fit of a linear model when $\hat{\gamma}=1$. Each data point in the scatter plot represents a Wald test statistic obtained for a specific gene.}
\label{fig:test_g}
\end{figure*}

\clearpage
\newpage

\appendix

\section{Supporting Lemmas} \label{app:A}
\begin{lemma} \label{lem:1}
Suppose $[m_0, c_1, c_2]^T \sim P$ such that the probability density function for $P$, denoted as $\rho(m_0, c_1, c_2)$, is symmetric with a mean of zero. Then $\mathbb{E}\{c_{1}\sin(c_{2})\} = \mathbb{E}\{c_1\cos(c_2)\} = 0$.
\end{lemma}

\begin{proof}
Given that $\rho(m_0, c_1, c_2)$ is symmetric, the corresponding marginal probability density function $\rho(c_1, c_2)$ is also symmetric. For $\mathbb{E}\{c_1\cos(c_2)\}$, it follows that

\begin{align*}
    \mathbb{E}\{c_1\cos(c_2)\} &= \int_{-\infty}^{\infty}\int_{-\infty}^{\infty} \rho(c_1, c_2)c_1\cos(c_2) dc_1 dc_2 \\
    &= \int_{-\infty}^{\infty}\left\{\int_{0}^{\infty} \rho(c_1, c_2)c_1\cos(c_2) dc_1 + \int_{-\infty}^{0} \rho(c_1, c_2)c_1\cos(c_2) dc_1\right\} dc_2 \\
    &= \int_{-\infty}^{\infty}\left\{\int_{0}^{\infty} \rho(c_1, c_2)c_1\cos(c_2) dc_1 - \int_{0}^{\infty} \rho(-c_1, c_2)c_1\cos(c_2) dc_1\right\} dc_2 \\
    &= \int_{-\infty}^{\infty}\left\{\int_{0}^{\infty} \rho(c_1, c_2)c_1\cos(c_2) dc_1 - \int_{0}^{\infty} \rho(c_1, c_2)c_1\cos(c_2) dc_1\right\} dc_2 \\
    &= 0.
\end{align*} 
The derivation is equivalent for $\mathbb{E}\{c_{1}\sin(c_{2})\}$.
\end{proof}

\begin{lemma} \label{lem:2}
Suppose each $n_i = n$, $\Psi^{(g)}_i= \diag(\psi^{(g)}_1, \psi^{(g)}_2, \psi^{(g)}_3)$ is a $3\times 3$ diagonal matrix, and $X_{i,j} = 24(j-1)/n$. Further, define 
\begin{align*}
    \boldsymbol{W_i} &= \begin{bmatrix}1 & \sin\left(\frac{\pi X_{i,1}}{12}\right) & \cos\left(\frac{\pi X_{i,1}}{12}\right) \\
    \vdots & \vdots & \vdots \\
    1 & \sin\left(\frac{\pi X_{i,n}}{12}\right) & \cos\left(\frac{\pi X_{i,n}}{12}\right)
    \end{bmatrix}.
\end{align*} 
Then each element of the matrix $\boldsymbol{V^{(g)}_i}$ can be expressed as
\begin{align*}
    \left\{(\boldsymbol{V^{(g)}_i})^{-1}\right\}_{j,k} &= \left[\left\{(\sigma^{(g)})^2I_{n} + \boldsymbol{W_i}\Psi^{(g)} \boldsymbol{W_i}^T\right\}^{-1}\right]_{j,k} \\
     &= \frac{\mathbbm{1}\{j=k\}}{\sigma^2} \\
     &\quad \quad - \frac{\psi_1}{\sigma^2(N\psi_1+\sigma^2)} \\
     & \quad \quad - \frac{2\psi_2}{\sigma^2(N\psi_2+2\sigma^2)}\sin\left(\frac{\pi X_{i,j}}{12}\right)\sin\left(\frac{\pi X_{i,k}}{12}\right) \\
     & \quad \quad - \frac{2\psi_3}{\sigma^2(N\psi_3+2\sigma^2)}\cos\left(\frac{\pi X_{i,j}}{12}\right)\cos\left(\frac{\pi X_{i,k}}{12}\right),
\end{align*} 
where 
\begin{equation*}
    \mathbbm{1}\{j=k\} =  \begin{cases} 
      1 & j=k, \\
      0 & j\neq k.
   \end{cases}
\end{equation*}
\end{lemma}

\begin{proof}
To simplify presentation, the superscript $(g)$ is omitted. Recall the identity presented in \citet[page 78]{Davidian1995},
\begin{align*}
    (R + Z\Psi Z^T)^{-1} &= R^{-1} - R^{-1}Z(Z^TR^{-1}Z + \Psi^{-1})^{-1}Z^TR^{-1},
\end{align*} 
which yields
\begin{align*}
    \boldsymbol{V_i}^{-1} &= (\sigma^2I_{n} + \boldsymbol{W_i}\Psi \boldsymbol{W_i}^T)^{-1} \\
    &= \frac{1}{\sigma^2}I_{n} - \frac{1}{\sigma^4}\boldsymbol{W_i}\left(\frac{1}{\sigma^2}\boldsymbol{W_i}^T\boldsymbol{W_i}+\Psi^{-1}\right)^{-1}\boldsymbol{W_i}^T.
\end{align*} 
To compute this matrix, first note
\begin{align}
    \frac{1}{\sigma^2}\boldsymbol{W_i}^T\boldsymbol{W_i} &= \frac{1}{\sigma^2}\begin{bmatrix}
    1 & \cdots & 1 \\ 
    \sin\left(\frac{\pi X_{i,1}}{12}\right) & \cdots & \sin\left(\frac{\pi X_{i,n}}{12}\right) \\ 
    \cos\left(\frac{\pi X_{i,1}}{12}\right) & \cdots & \cos\left(\frac{\pi X_{i,n}}{12}\right)   \end{bmatrix} 
    \begin{bmatrix}1 & \sin\left(\frac{\pi X_{i,1}}{12}\right) & \cos\left(\frac{\pi X_{i,1}}{12}\right) \\
    \vdots & \vdots & \vdots \\
    1 & \sin\left(\frac{\pi X_{i,n}}{12}\right) & \cos\left(\frac{\pi X_{i,n}}{12}\right)
    \end{bmatrix} \nonumber \\
    &= \frac{1}{\sigma^2}
    \begin{bmatrix} n & \sum_{j=1}^n\sin(X_{i,j}) & \sum_{j=1}^n\cos(X_{i,j}) \\
    \sum_{j=1}^n\sin(X_{i,j}) & \sum_{j=1}^n\sin^2(X_{i,j}) & \sum_{j=1}^n\sin(X_{i,j})\cos(X_{i,j}) \\
    \sum_{j=1}^n\cos(X_{i,j}) & \sum_{j=1}^n\sin(X_{i,j})\cos(X_{i,j}) & \sum_{j=1}^n\cos^2(X_{i,j})
    \end{bmatrix} \nonumber \\
    &= \frac{n}{\sigma^2}\begin{bmatrix}1 & 0 & 0 \\
    0 & \frac{1}{2} & 0 \\
    0 & 0 & \frac{1}{2}
    \end{bmatrix}, \label{eq:lem_2_1}
\end{align} 
where (\ref{eq:lem_2_1}) is due to orthogonality of distinct terms in a trigonometric basis \citep[Lemma 1.7]{Tsybakov2009}. It follows that
\begin{align}
    \left(\frac{1}{\sigma^2}\boldsymbol{W_i}^T\boldsymbol{W_i}+\Psi^{-1}\right)^{-1} &= \left( \begin{bmatrix}\frac{n}{\sigma^2} & 0 & 0 \\
    0 & \frac{n}{2\sigma^2} & 0 \\
    0 & 0 & \frac{n}{2\sigma^2}
    \end{bmatrix} + \begin{bmatrix} \frac{1}{\psi_1} & 0 & 0 \\
    0 & \frac{1}{\psi_2} & 0 \\
    0 & 0 & \frac{1}{\psi_3}
    \end{bmatrix}\right)^{-1} \nonumber \\
    &= \begin{bmatrix}\frac{\psi_1\sigma^2}{n\psi_1+\sigma^2} & 0 & 0 \\
    0 & \frac{2\psi_2\sigma^2}{n\psi_2+2\sigma^2} & 0 \\
    0 & 0 &  \frac{2\psi_3\sigma^2}{n\psi_3+2\sigma^2} 
    \end{bmatrix}. \label{eq:comp_1}
\end{align} 
The identity in (\ref{eq:comp_1}) can be utilized to obtain
\begin{align*}
\boldsymbol{W_i}\left(\frac{1}{\sigma^2}\boldsymbol{W_i}^T\boldsymbol{W_i}+\Psi^{-1}\right)^{-1}\boldsymbol{W_i}^T 
&= \begin{bmatrix}1 & \sin\left(\frac{\pi X_{i,1}}{12}\right) & \cos\left(\frac{\pi X_{i,1}}{12}\right) \\
    \vdots & \vdots & \vdots \\
    1 & \sin\left(\frac{\pi X_{i,n}}{12}\right) & \cos\left(\frac{\pi X_{i,n}}{12}\right) \end{bmatrix}\begin{bmatrix}\frac{\psi_1\sigma^2}{n\psi_1+\sigma^2} & 0 & 0 \\
    0 & \frac{2\psi_2\sigma^2}{n\psi_2+2\sigma^2} & 0 \\
    0 & 0 &  \frac{2\psi_3\sigma^2}{n\psi_3+2\sigma^2}
    \end{bmatrix} \\
    & \quad \quad \times \begin{bmatrix}
    1 & \cdots & 1 \\ 
    \sin\left(\frac{\pi X_{i,1}}{12}\right) & \cdots & \sin\left(\frac{\pi X_{i,n}}{12}\right) \\ 
    \cos\left(\frac{\pi X_{i,1}}{12}\right) & \cdots & \cos\left(\frac{\pi X_{i,n}}{12}\right)   \end{bmatrix} \\
    &= \begin{bmatrix}\frac{\psi_1\sigma^2}{n\psi_1+\sigma^2} & \frac{2\psi_2\sigma^2}{n\psi_2+2\sigma^2}\sin\left(\frac{\pi X_{i,1}}{12}\right) & \frac{2\psi_3\sigma^2}{n\psi_3+2\sigma^2}\cos\left(\frac{\pi X_{i,1}}{12}\right) \\
    \vdots & \vdots & \vdots \\
    \frac{\psi_1\sigma^2}{n\psi_1+\sigma^2} & \frac{2\psi_2\sigma^2}{n\psi_2+2\sigma^2}\sin\left(\frac{\pi X_{i,n}}{12}\right) & \frac{2\psi_3\sigma^2}{n\psi_3+2\sigma^2}\cos\left(\frac{\pi X_{i,n}}{12}\right) \end{bmatrix} \\
    & \quad \quad \times \begin{bmatrix}
    1 & \cdots & 1 \\ 
    \sin\left(\frac{\pi X_{i,1}}{12}\right) & \cdots & \sin\left(\frac{\pi X_{i,n}}{12}\right) \\ 
    \cos\left(\frac{\pi X_{i,1}}{12}\right) & \cdots & \cos\left(\frac{\pi X_{i,n}}{12}\right)   \end{bmatrix},
\end{align*} 
which implies
\begin{align*}
\left\{\boldsymbol{W_i}\left(\frac{1}{\sigma^2}\boldsymbol{W_i}^T\boldsymbol{W_i}+\Psi^{-1}\right)^{-1}\boldsymbol{W_i}^T \right\}_{j,k} &= \frac{\psi_1\sigma^2}{n\psi_1+\sigma^2} \\
& \quad \quad + \frac{2\psi_2\sigma^2}{n\psi_2+2\sigma^2}\sin\left(\frac{\pi X_{i,j}}{12}\right)\sin\left(\frac{\pi X_{i,k}}{12}\right) \\
& \quad \quad + \frac{2\psi_3\sigma^2}{n\psi_3+2\sigma^2}\cos\left(\frac{\pi X_{i,j}}{12}\right)\cos\left(\frac{\pi X_{i,k}}{12}\right).
\end{align*} 
To conclude,
\begin{align*}
 \left(\boldsymbol{V_i}^{-1}\right)_{j,k} &= \left\{(\sigma^2I_{n} + \boldsymbol{W_i}\Psi \boldsymbol{W_i}^T)^{-1}\right\}_{j,k} \\
    &= \left\{\frac{1}{\sigma^2}I_{n} - \frac{1}{\sigma^4}\boldsymbol{W_i}\left(\frac{1}{\sigma^2}\boldsymbol{W_i}^T\boldsymbol{W_i}+\Psi^{-1}\right)^{-1}\boldsymbol{W_i}^T\right\}_{j,k} \\
    &= \frac{\mathbbm{1}\{j=k\}}{\sigma^2} \\
    &\quad \quad - \frac{\psi_1}{\sigma^2(n\psi_1+\sigma^2)} \\
    & \quad \quad - \frac{2\psi_2}{\sigma^2(n\psi_2+2\sigma^2)}\sin\left(\frac{\pi X_{i,j}}{12}\right)\sin\left(\frac{\pi X_{i,k}}{12}\right) \\
    &\quad \quad - \frac{2\psi_3}{\sigma^2(n\psi_3+2\sigma^2)}\cos\left(\frac{\pi X_{i,j}}{12}\right)\cos\left(\frac{\pi X_{i,k}}{12}\right).
\end{align*} 
\end{proof}
\begin{lemma} \label{lem:3}
Suppose $[m^{(g)}_0, c^{(g)}_1, c^{(g)}_2]^T \sim P^{(g)}$ such that the probability density function of $P^{(g)}$, denoted as $\rho(m^{(g)}_0, c^{(g)}_1, c^{(g)}_2)$, is symmetric with a mean of zero. Then
\begin{align*}
    \mathbb{E}(Y^{(g)} \mid X=X_{i,j}) &= \mu^{(g)}_0 + \phi_{c^{(g)}_2}(1)\theta^{(g)}_1\left\{- \sin(\theta^{(g)}_2)\sin\left(\frac{\pi X_{i,j}}{12} \right) + \cos(\theta^{(g)}_2)\cos\left(\frac{\pi X_{i,j}}{12} \right)\right\}.
\end{align*} 
\end{lemma}
\begin{proof}
The superscript $(g)$ is omitted. The result follows the derivation of Theorem 1 in \cite{Gorczycaa2024}, with
\begin{align}
\mathbb{E}(Y_{i,j} \mid X=X_{i,j}) &= \mathbb{E}\left\{\mu_0 + m_{i,0} +  (\theta_1+c_{i,1})\cos\left(\frac{\pi X}{12} + \theta_2 + c_{i,2}\right)+\epsilon\mid X=X_{i,j} \right\} \nonumber \\
&= \mathbb{E}\left\{\mu_0 +  (\theta_1+c_{i,1})\cos\left(\frac{\pi X}{12} + \theta_2 + c_2\right) \mid X=X_{i,j} \right\} \nonumber \\
&= \mu_0 + \theta_1\mathbb{E}\left\{\cos\left(\frac{\pi X}{12} + \theta_2 + c_2\right) \mid X=X_{i,j} \right\} \nonumber \\
& \quad \quad + \mathbb{E}\left\{c_1\cos\left(\frac{\pi X}{12} + \theta_2 + c_2\right) \mid X=X_{i,j} \right\} \nonumber  \\
&= \mu_0 + \theta_1\mathbb{E}\left\{\cos\left(\frac{\pi X}{12} + \theta_2 + c_2\right) \mid X=X_{i,j}\right\}. \label{eq:lem_1_app}
\end{align} 
Here, (\ref{eq:lem_1_app}) is due to Lemma \ref{lem:1}. Now, note that
\begin{align}
    \mathbb{E}\left\{\cos\left(\frac{\pi X}{12} + \theta_2 + c_2\right) \mid X=X_{i,j}\right\} &= \mathbb{E}\left\{\cos\left(c_2\right)\cos\left(\frac{\pi X}{12} + \theta_2\right) \mid X=X_{i,j}\right\} \nonumber \\
    & \quad \quad - \mathbb{E}\left\{\sin\left(c_2\right)\sin\left(\frac{\pi X}{12} + \theta_2\right) \mid X=X_{i,j}\right\} \label{eq:app_3_1} \\
    &= \mathbb{E}\left\{\cos\left(c_2\right)\right\}\mathbb{E}\left\{\cos\left(\frac{\pi X}{12} + \theta_2\right) \mid X=X_{i,j}\right\} \nonumber \\
    & \quad \quad - \mathbb{E}\left\{\sin\left(c_2\right)\right\}\mathbb{E}\left\{\sin\left(\frac{\pi X}{12} + \theta_2\right) \mid X=X_{i,j}\right\} \label{eq:ind_c} \\
    &= \mathbb{E}\left\{\cos\left(c_2\right)\right\}\mathbb{E}\left\{\cos\left(\frac{\pi X}{12} + \theta_2\right) \mid X=X_{i,j}\right\} \label{eq:eval_sin} \\
    &= \phi_{c^{ }_2}(1)\mathbb{E}\left\{\cos\left(\frac{\pi X}{12} + \theta_2\right) \mid X=X_{i,j}\right\}, \label{eq:Four_ser}
\end{align} 
where (\ref{eq:app_3_1}) is by application of the identities in (\ref{eq:alt_to_orig}); (\ref{eq:ind_c}) is due to the fact that $X$ and $c_2$ are independent under the assumptions of a linear mixed effects model; (\ref{eq:eval_sin}) is due to the assumption that the marginal density of $c_2$ is symmetric around zero, which implies $\mathbb{E}\left\{\sin\left(c_2\right)\right\}=0$; and (\ref{eq:Four_ser}) is by definition of the characteristic function for a symmetric probability distribution, or
\begin{align*}
    \phi_{c_2}(t) = \int_{-\infty}^{\infty} \rho(c_2)\exp(itc_2)dc_2 = \int_{-\infty}^{\infty} \rho(c_2)\{\cos(tc_2)+i\sin(tc_2)\}dc_2 = \mathbb{E}\{\cos(tc_2)\}
\end{align*} 
given an argument $t$. We can conclude
\begin{align}
\mathbb{E}(Y_{i,j} \mid X=X_{i,j}) &= \mu_0 + \phi_{c^{ }_2}(1)\theta_1\cos\left(\frac{\pi X_{i,j}}{12} + \theta_2\right) \nonumber \\
&= \mu_0 + \phi_{c^{ }_2}(1)\theta_1\left\{- \sin(\theta_2)\sin\left(\frac{\pi X_{i,j}}{12} \right) + \cos(\theta_2)\cos\left(\frac{\pi X_{i,j}}{12} \right)\right\}, \label{eq:id_app}
\end{align} 
where (\ref{eq:id_app}) is by application of the identities in (\ref{eq:alt_to_orig}).
\end{proof}

\section{Derivation of Proposition \ref{prop:1}} \label{app:B}
\begin{proof}
The superscript $(g)$ is again omitted to simplify presentation. The expected parameter estimates are first obtained. Given that $\boldsymbol{V_i}$ is known, estimation of $\beta$ is analogous to solving the normal equation for generalized least squares estimation, or identifying the quantity $\hat{\beta}$ that satisfies the equality
\begin{align*}
    \left(\sum_{i=1}^M\boldsymbol{W_i}^T\boldsymbol{V_i}^{-1}\boldsymbol{W_i}\right)\hat{\beta} - \left(\sum_{i=1}^M\boldsymbol{W_i}^T\boldsymbol{V_i}^{-1}\boldsymbol{Y_i}\right) &= 0,
\end{align*} 
which under expectation is expressed as
\begin{align*}
    \left(\sum_{i=1}^M\boldsymbol{W_i}^T\boldsymbol{V_i}^{-1}\boldsymbol{W_i}\right)\mathbb{E}(\hat{\beta}) - \left\{\sum_{i=1}^M\boldsymbol{W_i}^T\boldsymbol{V_i}^{-1}\mathbb{E}(\boldsymbol{Y_i} \mid \boldsymbol{X_i})\right\} &= 0
\end{align*} 
given $\boldsymbol{V_i}$, $\boldsymbol{W_i}$, and $\boldsymbol{X_i}$ are non-random quantities. First, note that application of Lemma \ref{lem:2} yields
\begin{align*}
    \left(\boldsymbol{W_i}^T\boldsymbol{V_i}^{-1}\right)_{1,k} &= \frac{1}{\sigma^2} - \sum_{k=1}^n\left\{\frac{\psi_1}{\sigma^2(n\psi_1+\sigma^2)} \right. \\
    & \left. \quad \quad - \frac{2\psi_2}{\sigma^2(n\psi_2+2\sigma^2)}\sin\left(\frac{\pi X_{i,j}}{12}\right)\sin\left(\frac{\pi X_{i,k}}{12}\right) \right. \\
    &\left. \quad \quad - \frac{2\psi_3}{\sigma^2(n\psi_3+2\sigma^2)}\cos\left(\frac{\pi X_{i,j}}{12}\right)\cos\left(\frac{\pi X_{i,k}}{12}\right)\right\} \\
    &= \frac{1}{\sigma^2} - \frac{n\psi_1}{\sigma^2(n\psi_1+\sigma^2)}
\end{align*} 
for elements in the first row of $\boldsymbol{W_i}^T\boldsymbol{V_i}^{-1}$,
\begin{align*}
    \left(\boldsymbol{W_i}^T\boldsymbol{V_i}^{-1}\right)_{2,k} &= \left[\frac{1}{\sigma^2} - \sum_{k=1}^n \left\{\frac{\psi_1}{\sigma^2(n\psi_1+\sigma^2)} \right.  \right. \\
    & \left. \left. \quad \quad + \frac{2\psi_2}{\sigma^2(n\psi_2+2\sigma^2)}\sin\left(\frac{\pi X_{i,j}}{12}\right)\sin\left(\frac{\pi X_{i,k}}{12}\right) \right. \right. \\
    & \left. \left. \quad \quad + \frac{2\psi_3}{\sigma^2(n\psi_3+2\sigma^2)}\cos\left(\frac{\pi X_{i,j}}{12}\right)\cos\left(\frac{\pi X_{i,k}}{12}\right)\right\}\right]\sin\left( \frac{\pi X_{i,k}}{12}\right) \\
    &= \left\{\frac{1}{\sigma^2} - \frac{n\psi_2}{\sigma^2(n\psi_2+2\sigma^2)}\right\}\sin\left(\frac{\pi X_{i,j}}{12}\right)
\end{align*} 
for elements in the second row of $\boldsymbol{W_i}^T\boldsymbol{V_i}^{-1}$, and
\begin{align*}
    \left(\boldsymbol{W_i}^T\boldsymbol{V_i}^{-1}\right)_{3,k} &= \left[\frac{1}{\sigma^2} -\sum_{k=1}^n\left\{ \frac{\psi_1}{\sigma^2(n\psi_1+\sigma^2)} \right. \right. \\
    & \left. \left. \quad \quad + \frac{2\psi_2}{\sigma^2(n\psi_2+2\sigma^2)}\sin\left(\frac{\pi X_{i,j}}{12}\right)\sin\left(\frac{\pi X_{i,k}}{12}\right) \right. \right. \\
    & \left. \left. \quad \quad + \frac{2\psi_3}{\sigma^2(n\psi_3+2\sigma^2)}\cos\left(\frac{\pi X_{i,j}}{12}\right)\cos\left(\frac{\pi X_{i,k}}{12}\right)\right\}\right]\cos\left( \frac{\pi X_{i,k}}{12}\right) \\
    &= \left\{\frac{1}{\sigma^2} - \frac{n\psi_3}{\sigma^2(n\psi_3+2\sigma^2)}\right\}\cos\left(\frac{\pi X_{i,j}}{12}\right)
\end{align*} 
for elements in the third row of $\boldsymbol{W_i}^T\boldsymbol{V_i}^{-1}$. From these three expressions, it follows that
\begin{align}
\boldsymbol{W_i}^T\boldsymbol{V_i}^{-1}\boldsymbol{W_i} &= \begin{bmatrix} \frac{1}{\sigma^2} - \frac{n\psi_1}{\sigma^2(n\psi_1+\sigma^2)} & \cdots & \frac{1}{\sigma^2} - \frac{n\psi_1}{\sigma^2(n\psi_1+\sigma^2)} \\
 \left\{\frac{1}{\sigma^2} - \frac{\psi_2}{\sigma^2(n\psi_2+2\sigma^2)}\right\}\sin\left(\frac{\pi X_{i,1}}{12}\right) & \cdots & \left\{\frac{1}{\sigma^2} - \frac{\psi_2}{\sigma^2(n\psi_2+2\sigma^2)}\right\}\sin\left(\frac{\pi X_{i,n}}{12}\right) \\
 \left\{\frac{1}{\sigma^2} - \frac{\psi_3}{\sigma^2(n\psi_3+2\sigma^2)}\right\}\cos\left(\frac{\pi X_{i,1}}{12}\right) & \cdots & \left\{\frac{1}{\sigma^2} - \frac{\psi_3}{\sigma^2(n\psi_3+2\sigma^2)}\right\}\cos\left(\frac{\pi X_{i,n}}{12}\right)
\end{bmatrix} \nonumber \\
& \quad \quad \times 
\begin{bmatrix}1 & \sin\left(\frac{\pi X_{i,1}}{12}\right) & \cos\left(\frac{\pi X_{i,1}}{12}\right) \\
    \vdots & \vdots & \vdots \\
    1 & \sin\left(\frac{\pi X_{i,n}}{12}\right) & \cos\left(\frac{\pi X_{i,n}}{12}\right) 
    \end{bmatrix} \nonumber \\
    &= \begin{bmatrix}\frac{n}{\sigma^2} - \frac{n^2\psi_1}{\sigma^2(n\psi_1+\sigma^2)} & 0 & 0 \\
    0 & \frac{n}{2}\left\{\frac{1}{\sigma^2} - \frac{n\psi_2}{\sigma^2(n\psi_2+2\sigma^2)}\right\} & 0 \\
    0 & 0 & \frac{n}{2}\left\{\frac{1}{\sigma^2} - \frac{n\psi_3}{\sigma^2(n\psi_3+2\sigma^2)}\right\}  \end{bmatrix}, \label{ortho_2}
\end{align} 
where (\ref{ortho_2}) is due to orthogonality of distinct terms in a trigonometric basis noted in Lemma \ref{lem:2} \citep[Lemma 1.7]{Tsybakov2009}. Now, by application of Lemma \ref{lem:3}, we find
\begin{align}
\boldsymbol{W_i}^T\boldsymbol{V_i}^{-1}\mathbb{E}(\boldsymbol{Y_i}\mid \boldsymbol{X_i}) 
&= \begin{bmatrix} \frac{1}{\sigma^2} - \frac{n\psi_1}{\sigma^2(n\psi_1+\sigma^2)} & \cdots & \frac{1}{\sigma^2} - \frac{n\psi_1}{\sigma^2(n\psi_1+\sigma^2)} \\
 \left\{\frac{1}{\sigma^2} - \frac{n\psi_2}{\sigma^2(n\psi_2+2\sigma^2)}\right\}\sin\left(\frac{\pi X_{i,1}}{12}\right) & \cdots & \left\{\frac{1}{\sigma^2} - \frac{n\psi_2}{\sigma^2(n\psi_2+2\sigma^2)}\right\}\sin\left(\frac{\pi X_{i,n}}{12}\right)\\
 \left\{\frac{1}{\sigma^2} - \frac{n\psi_3}{\sigma^2(n\psi_3+2\sigma^2)}\right\}\cos\left(\frac{\pi X_{i,1}}{12}\right) & \cdots & \left\{\frac{1}{\sigma^2} - \frac{n\psi_3}{\sigma^2(n\psi_3+2\sigma^2)}\right\}\cos\left(\frac{\pi X_{i,n}}{12}\right)
\end{bmatrix} \nonumber \\
&\quad \quad \times 
\begin{bmatrix}
\mu_0 + \phi_{c^{ }_2}(1)\theta_1\left\{- \sin(\theta_2)\sin\left(\frac{\pi X_{i,1}}{12} \right) + \cos(\theta_2)\cos\left(\frac{\pi X_{i,1}}{12} \right)\right\} \\
\vdots \\
\mu_0 + \phi_{c^{ }_2}(1)\theta_1\left\{- \sin(\theta_2)\sin\left(\frac{\pi X_{i,n}}{12} \right) + \cos(\theta_2)\cos\left(\frac{\pi X_{i,n}}{12} \right)\right\}
\end{bmatrix}, \nonumber 
\end{align} 
which yields
\begin{align*}
    \left\{\boldsymbol{W_i}^T\boldsymbol{V_i}^{-1}\mathbb{E}(\boldsymbol{Y_i}\mid \boldsymbol{X_i})\right\}_{1,1} &= \left\{\frac{1}{\sigma^2} - \frac{n\psi_1}{\sigma^2(n\psi_1+\sigma^2)}\right\} \\
    & \quad \quad \times \sum_{j=1}^n\left[ \mu_0  + \phi_{c^{ }_2}(1)\theta_1\left\{- \sin(\theta_2)\sin\left(\frac{\pi X_{i,j}}{12} \right) \right. \right.  \\
    & \quad \quad \quad \quad \left. \left. + \cos(\theta_2)\cos\left(\frac{\pi X_{i,j}}{12} \right)\right\}\right] \\
    &= \mu_0\left\{\frac{n}{\sigma^2} - \frac{n^2\psi_1}{\sigma^2(n\psi_1+\sigma^2)}\right\}
\end{align*} 
for the first element of $\boldsymbol{W_i}^T\boldsymbol{V_i}^{-1}\mathbb{E}(\boldsymbol{Y_i}\mid \boldsymbol{X_i})$,
\begin{align*}
    \left\{\boldsymbol{W_i}^T\boldsymbol{V_i}^{-1}\mathbb{E}(\boldsymbol{Y_i}\mid \boldsymbol{X_i})\right\}_{2,1} &= \left\{\frac{1}{\sigma^2} - \frac{n\psi_2}{\sigma^2(n\psi_2+2\sigma^2)}\right\} \\
    & \quad \quad \times \sum_{j=1}^n\sin\left(\frac{\pi X_{i,j}}{12}\right)\left[ \mu_0 + \phi_{c^{ }_2}(1)\theta_1\left\{- \sin(\theta_2)\sin\left(\frac{\pi X_{i,j}}{12} \right) \right. \right. \\
    &\left. \left. \quad \quad \quad \quad + \cos(\theta_2)\cos\left(\frac{\pi X_{i,j}}{12} \right)\right\} \right] \\
    &= -\frac{n\phi_{c^{ }_2}(1)\theta_1\sin(\theta_2)}{2}\left\{\frac{1}{\sigma^2} - \frac{n\psi_2}{\sigma^2(n\psi_2+2\sigma^2)}\right\}
\end{align*} 
for the second element of $\boldsymbol{W_i}^T\boldsymbol{V_i}^{-1}\mathbb{E}(\boldsymbol{Y_i}\mid \boldsymbol{X_i})$, and
\begin{align*}
    \left\{\boldsymbol{W_i}^T\boldsymbol{V_i}^{-1}\mathbb{E}(\boldsymbol{Y_i}\mid \boldsymbol{X_i})\right\}_{3,1} &= \left\{\frac{1}{\sigma^2} - \frac{\psi_3}{\sigma^2(n\psi_3+2\sigma^2)}\right\} \\
    & \quad \quad \times \sum_{j=1}^n\cos\left(\frac{\pi X_{i,j}}{12}\right)\left[ \mu_0 + \phi_{c^{ }_2}(1)\theta_1\left\{- \sin(\theta_2)\sin\left(\frac{\pi X_{i,j}}{12} \right) \right. \right. \\
    &\left. \left. \quad \quad \quad \quad + \cos(\theta_2)\cos\left(\frac{\pi X_{i,j}}{12} \right)\right\} \right] \\
    &= \frac{n\phi_{c^{ }_2}(1)\theta_1\cos(\theta_2)}{2}\left\{\frac{1}{\sigma^2} - \frac{n\psi_3}{\sigma^2(n\psi_3+2\sigma^2)}\right\}
\end{align*} 
for the third element of $\boldsymbol{W_i}^T\boldsymbol{V_i}^{-1}\mathbb{E}(\boldsymbol{Y_i}\mid \boldsymbol{X_i})$. The expected fixed parameter estimates can subsequently be expressed as
\begin{align}
    \mathbb{E}(\hat{\beta}) &= \begin{bmatrix}
    \mathbb{E}(\hat{\mu}_0) \\
    \mathbb{E}(\hat{\beta}_1) \\
    \mathbb{E}(\hat{\beta}_2)
    \end{bmatrix} \nonumber \\
    &= \left(\sum_{i=1}^M\boldsymbol{W_i}^T\boldsymbol{V_i}^{-1}\boldsymbol{W_i}\right)^{-1} \left(\sum_{i=1}^M\boldsymbol{W_i}^T\boldsymbol{V_i}^{-1}\mathbb{E}(\boldsymbol{Y_i} \mid \boldsymbol{X_i})\right) \nonumber \\
    &= \left(M\begin{bmatrix}\frac{n}{\sigma^2} - \frac{n^2\psi_1}{\sigma^2(n\psi_1+\sigma^2)} & 0 & 0 \\
    0 & \frac{n}{2}\left\{\frac{1}{\sigma^2} - \frac{n\psi_2}{\sigma^2(n\psi_2+2\sigma^2)}\right\} & 0 \\
    0 & 0 & \frac{n}{2}\left\{\frac{1}{\sigma^2} - \frac{n\psi_3}{\sigma^2(n\psi_3+2\sigma^2)}\right\}  \end{bmatrix}\right)^{-1} 
\nonumber \\
& \quad \quad \times M\begin{bmatrix} \mu_0\left\{\frac{n}{\sigma^2} - \frac{n^2\psi_1}{\sigma^2(n\psi_1+\sigma^2)}\right\}   \\
  -\frac{n\phi_{c^{ }_2}(1)\theta_1\sin(\theta_2)}{2}\left\{\frac{1}{\sigma^2} - \frac{n\psi_2}{\sigma^2(n\psi_2+2\sigma^2)}\right\} \\
 \frac{n\phi_{c^{ }_2}(1)\theta_1\cos(\theta_2)}{2}\left\{\frac{1}{\sigma^2} - \frac{n\psi_3}{\sigma^2(n\psi_3+2\sigma^2)}\right\}
\end{bmatrix} \nonumber \\
    &= \begin{bmatrix}
    \mu_0 \\
    -\phi_{c^{ }_2}(1)\theta_1\sin(\theta_2) \\
    \phi_{c^{ }_2}(1)\theta_1\cos(\theta_2)
    \end{bmatrix} \nonumber \\
    &= \begin{bmatrix}
    \mu_0 \\
    \phi_{c^{ }_2}(1)\beta_1 \\
    \phi_{c^{ }_2}(1)\beta_2
    \end{bmatrix}, \label{eq:eq_3_res}
\end{align} 
where (\ref{eq:eq_3_res}) is due to the identities in (\ref{eq:alt_to_orig}).

Consideration is now given towards computing the $\tau$ in (\ref{eq:wald}) given these parameter estimates. First, note that
\begin{align*}
L(\hat{\beta} - \hat{\beta}_{\text{Null}}) &= \begin{bmatrix}
        0 & 1 & 0 \\
        0 & 0 & 1
    \end{bmatrix}\left(\begin{bmatrix}
    \mu_0 \\
    \phi_{c^{ }_2}(1)\beta_1 \\
    \phi_{c^{ }_2}(1)\beta_2
    \end{bmatrix} - \begin{bmatrix}
    0 \\
    0 \\
    0
    \end{bmatrix} \right) \\
    &= \begin{bmatrix}
    \phi_{c^{ }_2}(1)\beta_1 \\
    \phi_{c^{ }_2}(1)\beta_2
    \end{bmatrix}.
\end{align*} 
Further,
\begin{align*}
    L\left\{\sum_{i=1}^M\boldsymbol{W_i}^T\boldsymbol{V_i}^{-1}\boldsymbol{W_i}\right\}^{-1}L^T &= M^{-1}\begin{bmatrix}
        0 & 1 & 0 \\
        0 & 0 & 1
    \end{bmatrix} \\
    & \quad \quad \times \begin{bmatrix}\frac{n}{\sigma^2} - \frac{n^2\psi_1}{\sigma^2(n\psi_1+\sigma^2)} & 0 & 0 \\
    0 & \frac{n}{2}\left\{\frac{1}{\sigma^2} - \frac{n\psi_2}{\sigma^2(n\psi_2+2\sigma^2)}\right\} & 0 \\
    0 & 0 & \frac{n}{2}\left\{\frac{1}{\sigma^2} - \frac{n\psi_3}{\sigma^2(n\psi_3+2\sigma^2)}\right\}  \end{bmatrix}^{-1} \\
    & \quad \quad \times \begin{bmatrix}
        0 & 0 \\
        1 & 0 \\
        0 & 1
    \end{bmatrix} \\
    &= M^{-1}\begin{bmatrix} \frac{2}{n}\left\{\frac{1}{\sigma^2} - \frac{n\psi_2}{\sigma^2(n\psi_2+2\sigma^2)}\right\}^{-1} & 0 \\
     0 & \frac{2}{n}\left\{\frac{1}{\sigma^2} - \frac{n\psi_3}{\sigma^2(n\psi_3+2\sigma^2)}\right\}^{-1}  \end{bmatrix},
\end{align*} 
which implies that
\begin{align*}
\left[L\left\{\sum_{i=1}^M\boldsymbol{W_i}^T\boldsymbol{V_i}^{-1}\boldsymbol{W_i}\right\}^{-1}L^T\right]^{-1} &= M\begin{bmatrix} \frac{n}{2}\left\{\frac{1}{\sigma^2} - \frac{n\psi_2}{\sigma^2(n\psi_2+2\sigma^2)}\right\} & 0 \\
     0 & \frac{n}{2}\left\{\frac{1}{\sigma^2} - \frac{n\psi_3}{\sigma^2(n\psi_3+2\sigma^2)}\right\}  \end{bmatrix}.
\end{align*} 
To conclude, the Wald test statistic can be expressed as
\begin{align*}
    \tau &= M\begin{bmatrix}
    \phi_{c^{ }_2}(1)\beta_1 &
    \phi_{c^{ }_2}(1)\beta_2
    \end{bmatrix}\begin{bmatrix} \frac{n}{2}\left\{\frac{1}{\sigma^2} - \frac{n\psi_2}{\sigma^2(n\psi_2+2\sigma^2)}\right\} & 0 \\
     0 & \frac{n}{2}\left\{\frac{1}{\sigma^2} - \frac{n\psi_3}{\sigma^2(n\psi_3+2\sigma^2)}\right\}  \end{bmatrix}\begin{bmatrix}
    \phi_{c^{ }_2}(1)\beta_1 \\
    \phi_{c^{ }_2}(1)\beta_2
    \end{bmatrix} \\
    &= \frac{Mn\phi^2_{c_2}(1)}{2}\left[\left\{\frac{1}{\sigma^2} - \frac{n\psi_2}{\sigma^2(n\psi_2+2\sigma^2)}\right\}\beta_1^2 + \left\{\frac{1}{\sigma^2} - \frac{n\psi_3}{\sigma^2(n\psi_3+2\sigma^2)}\right\}\beta_2^2 \right].
\end{align*} 

\end{proof}

\section{Derivation of Proposition \ref{prop:2}} \label{app:C}
\begin{proof}
Omitting the superscript $(g)$, the response $Y_{i,j}$ can be expressed as
\begin{align}
Y_{i,j} &= (\mu_0 + m_{i, 0}) + \theta_1\cos(X_{i,j}+\theta_{2}) + c_{i,1}\cos(X_{i,j}+c_{i,2}) + \epsilon_{i,j} \nonumber \\
&= (\mu_0 + m_{i, 0}) -\theta_1\sin(\theta_2)\sin(X_{i,j}) + \theta_1\cos(\theta_2)\cos(X_{i,j}) \nonumber \\
& \quad \quad - c_{i,1}\sin(c_{i,2})\sin(X_{i,j}) + c_{i,1}\cos(c_{i,2})\cos(X_{i,j}) + \epsilon_{i,j} \nonumber \\
&= (\mu_0 + m_{i, 0}) -\left\{\theta_1\sin(\theta_2) + c_{i,1}\sin(c_{i,2})\right\}\sin(X_{i,j}) \nonumber \\
& \quad \quad + \left\{\theta_1\cos(\theta_2) + c_{i,1}\cos(c_{i,2})\right\}\cos(X_{i,j}) + \epsilon_{i,j} \nonumber \\ 
&= (\mu_0 + m_{i, 0}) + \sqrt{\left\{\theta_1\sin(\theta_2) + c_{i,1}\sin(c_{i,2})\right\}^2 + \left\{\theta_1\cos(\theta_2) + c_{i,1}\cos(c_{i,2})\right\}^2} \nonumber \\
    & \quad \quad \times \cos\left[\frac{\pi X_{i,j}}{12} + \atan\{-\theta_1\sin(\theta_2) - c_{i,1}\sin(c_{i,2}), \theta_1\cos(\theta_2) + c_{i,1}\cos(c_{i,2})\}\right] \nonumber \\
    & \quad \quad + \epsilon_{i,j} \label{eq:res_p2},
\end{align} 
where (\ref{eq:res_p2}) is due to the identities presented in (\ref{eq:alt_to_orig}).

\end{proof}

\section{Derivation of Cosinor Model Phase-Shift Variance} \label{app:D}
%\begin{proof}
Define $h(\mu_0, \beta_1, \beta_2) = \mathrm{atan2}(-\beta_1, \beta_2) = \theta_2$, which implies
\begin{align*}
    \nabla h(\mu_0, \beta_1, \beta_2) &= \begin{bmatrix} 0 & \frac{-\beta_2}{\sqrt{\beta_1^2+\beta_2^2}} & \frac{\beta_1}{\sqrt{\beta_1^2+\beta_2^2}}
    \end{bmatrix}.
\end{align*} 
Further, let $\hat{\Sigma}$ denote the estimated covariance matrix for the parameter vector $[\hat{\mu}_0, \hat{\beta}_1, \hat{\beta}_2]^T$. By application of the Delta method \citep[Theorem 1.3]{Boos2013},
\begin{align*}
    \mathrm{Var}(\hat{\theta}_2) &= \begin{bmatrix} 0 & \frac{-\hat{\beta}_2}{\sqrt{\hat{\beta}_1^2+\hat{\beta}_2^2}} & \frac{\hat{\beta}_1}{\sqrt{\hat{\beta}_1^2+\hat{\beta}_2^2}}
    \end{bmatrix} \begin{bmatrix}
    \hat{\Sigma}_{1,1} & \hat{\Sigma}_{1,2} & \hat{\Sigma}_{1,3} \\
    \hat{\Sigma}_{1,2} & \hat{\Sigma}_{2,2} & \hat{\Sigma}_{2,3} \\
    \hat{\Sigma}_{1,3} & \hat{\Sigma}_{2,3} & \hat{\Sigma}_{3,3}
    \end{bmatrix} \begin{bmatrix} 0 \\ \frac{-\hat{\beta}_2}{\sqrt{\hat{\beta}_1^2+\hat{\beta}_2^2}} \\ \frac{\hat{\beta}_1}{\sqrt{\hat{\beta}_1^2+\hat{\beta}_2^2}}
    \end{bmatrix}  \\
    &= \frac{\hat{\Sigma}_{2,2}\hat{\beta}_2^2+\hat{\Sigma}_{3,3}\hat{\beta}_1^2 - 2\hat{\Sigma}_{2,3}\hat{\beta}_1\hat{\beta}_2}{\hat{\beta}_1^2+\hat{\beta}_2^2}.
\end{align*} 
%\end{proof}

\section{Overview of Simulation Study Design from Section \ref{sec:3}} \label{app:E}
We present the following six simulation settings, where the names in parentheses are from the BioCycle$_{\text{Form}}$ software:
\begin{description}
    \item[Setting 1 (Cosine).] $\mu^{(g)}_0=6$, $\theta^{(g)}_1 = 0.5$, $\theta^{(g)}_2 = 0$, $m^{(g)}_{i,0} \sim \mathrm{N}(0, 1)$, $c^{(g)}_{i,1} \sim \mathrm{TN}(0, 0.5, -0.3, 0.3)$, $c^{(g)}_{i,2} \sim \mathrm{TN}(0, \pi^2/36, -\pi, \pi)$, $M = 10$, $n_i = 12$ (for all $i$), $X_{i,j}=2j$, $\epsilon^{(g)}_{i,j} \sim \text{N}(0, 0.25)$, and $$Y^{(g)}_{i,j} = (\mu^{(g)}_0+m^{(g)}_{i,0}) + (\theta^{(g)}_1+c^{(g)}_{i,1})\cos\left(\frac{\pi X_{i,j}}{12} + \theta^{(g)}_2 + c^{(g)}_{i,2} \right)+\epsilon^{(g)}_{i,j}.$$  
    \item[Setting 2 (Cosine + Outlier).] $\mu^{(g)}_0=6$, $\theta^{(g)}_1 = 0.5$, $\theta^{(g)}_2 = \pi/6$, $m^{(g)}_{i,0} \sim \mathrm{N}(0, 1)$, $c^{(g)}_{i,1} \sim \mathrm{TN}(0, 0.5, -0.3, 0.3)$, $c^{(g)}_{i,2} \sim \mathrm{TN}(0, \pi^2/36, -\pi, \pi)$, $M = 10$, $n_i = 8$ (for all $i$), $X_{i,j}=3j$, $\epsilon^{(g)}_{i,j} \sim \text{N}(0, 0.25)$, $P_{i,j} = \mathrm{Unif}(0, 1)$, and $$Y^{(g)}_{i,j} = g(P_{i,j})\left\{(\mu^{(g)}_0+m^{(g)}_{i,0}) + (\theta^{(g)}_1+c^{(g)}_{i,1})\cos\left(\frac{\pi X_{i,j}}{12} + \theta^{(g)}_2 + c^{(g)}_{i,2} \right)+\epsilon^{(g)}_{i,j}\right\},$$where $$g(P_{i,j}) =  \begin{cases} 
      1 & P_{i,j} \leq 0.95, \\
      1.5 & \text{otherwise.}
   \end{cases}
$$
    \item[Setting 3 (Cosine2).] $\mu^{(g)}_0=6$, $\theta^{(g)}_1 = 0.4$, $\theta^{(g)}_2 = \pi/3$, $m^{(g)}_{i,0} \sim \mathrm{N}(0, 1)$, $c^{(g)}_{i,1} \sim \mathrm{TN}(0, 0.5, -0.3, 0.3)$, $c^{(g)}_{i,2} \sim \mathrm{TN}(0, \pi^2/16, -\pi, \pi)$, $M = 10$, $n_i = 6$ (for all $i$), $X_{i,j}=4j$, $\epsilon^{(g)}_{i,j} \sim \text{N}(0, 0.25)$, and $$Y^{(g)}_{i,j} = (\mu^{(g)}_0+m^{(g)}_{i,0}) + (\theta^{(g)}_1+c^{(g)}_{i,1})\left\{\cos\left(\frac{\pi X_{i,j}}{12} - \theta^{(g)}_2 + c^{(g)}_{i,2} \right) + \frac{\cos\left(\frac{\pi X_{i,j}}{4} - \frac{\pi}{2} - \theta^{(g)}_2 + c^{(g)}_{i,2} \right)}{2}\right\}+\epsilon^{(g)}_{i,j}.$$
    \item[Setting 4 (Peak).] $\mu^{(g)}_0=6$, $\theta^{(g)}_1 = 0.4$, $\theta^{(g)}_2 = \pi/2$, $m^{(g)}_{i,0} \sim \mathrm{N}(0, 1)$, $c^{(g)}_{i,1} \sim \mathrm{TN}(0, 0.5, -0.3, 0.3)$, $c^{(g)}_{i,2} \sim \mathrm{TN}(0, \pi^2/16, -\pi, \pi)$, $M = 10$, $n_i = 12$ (for all $i$), $X_{i,j}=2j$, $\epsilon^{(g)}_{i,j} \sim \text{N}(0, 0.25)$, and $$Y_{i,j} = (\mu^{(g)}_0+m^{(g)}_{i,0}) + (\theta^{(g)}_1+c^{(g)}_{i,1})\left[-1+2\left\{\cos\left(\frac{\pi X_{i,j}}{24} + \frac{\theta^{(g)}_2}{2} + c^{(g)}_{i,2} \right)\right\}^{10} \right]+\epsilon^{(g)}_{i,j}.$$
    \item[Setting 5 (Triangle).] $\mu^{(g)}_0=6$, $\theta^{(g)}_1 = 0.3$, $\theta^{(g)}_2 = 2\pi/3$, $m^{(g)}_{i,0} \sim \mathrm{N}(0, 1)$, $c^{(g)}_{i,1} \sim \mathrm{TN}(0, 0.5, -0.3, 0.3)$, $c^{(g)}_{i,2} \sim \mathrm{TN}(0, \pi^2/9, -\pi, \pi)$, $M = 10$, $n_i = 8$ (for all $i$), $X_{i,j}=3j$, $\epsilon^{(g)}_{i,j} \sim \text{N}(0, 0.25)$, and 
    \begin{align*}
        Y^{(g)}_{i,j} &= (\mu^{(g)}_0+m^{(g)}_{i,0}) \\
        & \quad + \frac{8(\theta^{(g)}_1+c^{(g)}_{i,1})}{\pi^2}\left[\sin\left(\frac{\pi X_{i,j}}{12} - \frac{\pi}{2} - \theta^{(g)}_2 + c^{(g)}_{i,2} \right) \right. \\
        & \quad \quad \quad \left. - \left(\frac{1}{9}\right)\sin\left\{3\left(\frac{\pi X_{i,j}}{12} - \frac{\pi}{2} - \theta^{(g)}_2 + c^{(g)}_{i,2} \right)\right\} \right. \\
        & \quad \quad \quad \left. + \left(\frac{1}{25}\right)\sin\left\{5\left(\frac{\pi X_{i,j}}{12} - \frac{\pi}{2} - \theta^{(g)}_2 + c^{(g)}_{i,2} \right)\right\}\right] \\
        & \quad +\epsilon^{(g)}_{i,j}.
    \end{align*} 
    \item[Setting 6 (Square).] $\mu^{(g)}_0=6$, $\theta^{(g)}_1 = 0.3$, $\theta^{(g)}_2 = 5\pi/6$, $m^{(g)}_{i,0} \sim \mathrm{N}(0, 1)$, $c^{(g)}_{i,1} \sim \mathrm{TN}(0, 0.5, -0.3, 0.3)$, $c^{(g)}_{i,2} \sim \mathrm{TN}(0, \pi^2/9, -\pi, \pi)$, $M = 10$, $n_i = 6$ (for all $i$), $X_{i,j}=4j$, $\epsilon^{(g)}_{i,j} \sim \text{N}(0, 0.25)$, and \begin{align*}
        Y^{(g)}_{i,j} &= (\mu^{(g)}_0+m^{(g)}_{i,0}) \\
        & \quad + \frac{4(\theta^{(g)}_1+c^{(g)}_{i,1})}{\pi}\left[\sin\left(\frac{\pi X_{i,j}}{12} - \frac{\pi}{2} - \theta^{(g)}_2 + c^{(g)}_{i,2} \right) \right. \\
        & \quad \quad \quad \left. + \left(\frac{1}{3}\right)\sin\left\{3\left(\frac{\pi X_{i,j}}{12} - \frac{\pi}{2} - \theta^{(g)}_2 + c^{(g)}_{i,2} \right)\right\} \right. \\
        & \quad \quad \quad \left. + \left(\frac{1}{5}\right)\sin\left\{5\left(\frac{\pi X_{i,j}}{12} - \frac{\pi}{2} - \theta^{(g)}_2 + c^{(g)}_{i,2} \right)\right\}\right] \\
        & \quad +\epsilon^{(g)}_{i,j}.
    \end{align*} 
 \end{description}
Here, $\mathrm{N}(\mu, \sigma^2)$ denotes a normal distribution with mean $\mu$ and variance $\sigma^2$; $\mathrm{TN}(\mu, \sigma^2, a, b)$ a truncated normal distribution with mean $\mu$, variance $\sigma^2$, lower bound $a$, and upper bound $b$; and $\mathrm{Unif}(a,b)$ denotes a uniform distribution with support from $a$ to $b$. 

In each simulation setting, the true population amplitude parameter $\theta_1^{(g)} = 0.3$, which has been reported as being frequently observed across genes \citep{MllerLevet2013}. Further, for Settings 1 and 4, a sample is collected from each individual once every two hours over a 24 hour period; Settings 2 and 5 once every three hours over a 24 hour period; and Settings 3 and 6 once every four hours over a 24 hour period. The purpose of this design is to emulate scenarios where there is a limited budget for sample collection \citep{Brooks2023}.

\bibliographystyle{apalike} 
\bibliography{bibliography}

\end{document}